\documentclass[english,12pt]{article}
\usepackage[T1]{fontenc}
\usepackage[latin9]{inputenc}
\usepackage{color}
\usepackage{babel}
\usepackage{array}
\usepackage{float}
\usepackage{booktabs}
\usepackage{mathrsfs}
\usepackage{mathtools}
\usepackage{url}
\usepackage{multirow}
\usepackage{amsmath}
\usepackage{amsthm}
\usepackage{amssymb}
\usepackage{graphicx}
\usepackage{setspace}
\usepackage[authoryear]{natbib}
\usepackage[unicode=true,pdfusetitle,
 bookmarks=true,bookmarksnumbered=false,bookmarksopen=false,
 breaklinks=true,pdfborder={0 0 0},pdfborderstyle={},backref=false,colorlinks=true]
 {hyperref}
\hypersetup{
 citecolor=blue}

\makeatletter

\providecommand{\tabularnewline}{\\}
\floatstyle{ruled}
\newfloat{algorithm}{tbp}{loa}
\providecommand{\algorithmname}{Algorithm}
\floatname{algorithm}{\protect\algorithmname}

\theoremstyle{plain}
\newtheorem{prop}{\protect\propositionname}
\theoremstyle{plain}
\newtheorem{thm}{\protect\theoremname}
\theoremstyle{remark}
\newtheorem*{rem*}{\protect\remarkname}
\theoremstyle{plain}
\newtheorem{assumption}{\protect\assumptionname}
\theoremstyle{plain}
\newtheorem{lem}{\protect\lemmaname}

\setcitestyle{round}

\usepackage{algorithmic}

\usepackage{titling}
\usepackage[compact]{titlesec}

\date{}

\@ifundefined{showcaptionsetup}{}{%
 \PassOptionsToPackage{caption=false}{subfig}}
\usepackage{subfig}
\makeatother

\providecommand{\assumptionname}{Assumption}
\providecommand{\lemmaname}{Lemma}
\providecommand{\propositionname}{Proposition}
\providecommand{\remarkname}{Remark}
\providecommand{\theoremname}{Theorem}

\begin{document}
\onehalfspacing
\title{\textbf{Efficient Multimodal Sampling via Tempered Distribution Flow}}
\author{Yixuan Qiu\\
{\normalsize{}School of Statistics and Management, Shanghai University
of Finance and Economics}\\
{\normalsize{}Shanghai 200433, P.R. China, }\texttt{\normalsize{}qiuyixuan@sufe.edu.cn}\\
\\
Xiao Wang\\
{\normalsize{}Department of Statistics, Purdue University}\\
{\normalsize{}West Lafayette, IN 47907, U.S.A., }\texttt{\normalsize{}wangxiao@purdue.edu}}
\maketitle
\begin{abstract}
Sampling from high-dimensional distributions is a fundamental problem
in statistical research and practice. However, great challenges emerge
when the target density function is unnormalized and contains isolated
modes. We tackle this difficulty by fitting an invertible transformation
mapping, called a transport map, between a reference probability measure
and the target distribution, so that sampling from the target distribution
can be achieved by pushing forward a reference sample through the
transport map. We theoretically analyze the limitations of existing
transport-based sampling methods using the Wasserstein gradient flow
theory, and propose a new method called TemperFlow that addresses
the multimodality issue. TemperFlow adaptively learns a sequence of
tempered distributions to progressively approach the target distribution,
and we prove that it overcomes the limitations of existing methods.
Various experiments demonstrate the superior performance of this novel
sampler compared to traditional methods, and we show its applications
in modern deep learning tasks such as image generation. The programming
code for the numerical experiments is available at \url{https://github.com/yixuan/temperflow}.
\end{abstract}

\newpage{}

\setstretch{1.8}
\addtolength{\abovedisplayskip}{-3.5pt}
\addtolength{\belowdisplayskip}{-3.5pt}

\section{Introduction}

Sampling from probability distributions is a fundamental task in statistical
research and practice, with omnipresent applications in Bayesian data
analysis \citep{gelman2014bayesian}, uncertainty quantification \citep{sullivan2015introduction},
and generative machine learning models \citep{salakhutdinov2015learning,bond2021deep},
among many others. It has long been a central topic in statistical
computing and simulation, with many standard approaches developed,
including rejection sampling, importance sampling, Markov chain Monte
Carlo (MCMC), etc. The general sampling problem can be described in
the following way. Let $p(x),x\in\mathbb{R}^{d}$ be the target density
function from which one wishes to sample, and suppose that it is in
an unnormalized form. In other words, one only has access to the energy
function $E(x)=-\log p(x)+C$, where $C$ is an unknown constant free
of $x$. Then the aim of sampling is to generate an independent sample
of points $X_{1},\ldots,X_{n}$, each following the $p(x)$ distribution.

Despite the availability of standard sampling algorithms, great challenges
emerge when the target distribution has complicated structures, including
high dimensions, isolated modes, and strong correlations between components.
As a result, many sampling methods are not fully applicable to such
sophisticated problems. For example, rejection sampling requires a
strict upper bound for the density function of interest, which is
hard to obtain in general, let alone the unnormalized case. Importance
sampling, on the other hand, is mainly used to approximate integrals
rather than generating independent samples, and may have a high variance
that grows exponentially with the dimension of variables (see Examples
9.3 of \citealp{owen2013monte}). To this end, MCMC becomes one of
the most popular general-purpose samplers for high-dimensional distributions,
with good statistical properties and many variants available \citep{gilks1995markov,brooks2011handbook,dunson2020hastings}.

Nevertheless, MCMC methods are not without limitations. First, the
random variates generated by MCMC are correlated, which lead to reduced
effective sample size in approximating expectations. Also, they are
not appropriate for scenarios in which an independent sample is required.
Second, users typically need to manually tune hyperparameters in most
MCMC algorithms, which require much expertise and experience. Moreover,
it is non-trivial to test the convergence of a Markov chain, making
it challenging to determine the stopping rule. Last but not least,
the target density function typically contains many isolated modes,
which lead to MCMC samples that are highly sensitive to the initial
state, as sampled points may be trapped in one mode and can hardly
escape. All these aspects severely harm the application of MCMC in
practice.

More recently, the Schr\"odinger--F\"ollmer sampler \citep{huang2021schrodinger}
is a new sampling method based on a diffusion process, which transports
particles following a degenerate distribution at time zero to the
target distribution at time one. The Schr\"odinger--F\"ollmer sampler
has an easy implementation, and shows promising results on multimodal
Gaussian mixture models. However, the drift term of the diffusion
process does not always have closed forms, and may need to be approximated
by Monte Carlo methods, which further increases the computing burden.

To overcome the difficulties mentioned above, in this article we advocate
a new framework for efficient statistical sampling and simulation.
Our method is motivated by the measure transport framework introduced
in \citet{marzouk2016sampling}, whose central idea is to estimate
a deterministic transport map between a base probability measure and
the target distribution. The base measure is typically a fixed and
convenient distribution, and then independent random variables from
the target distribution can be obtained by simply pushing forward
an i.i.d. reference sample through the transport map. Despite its
appealing features, there are many unresolved challenges in \citet{marzouk2016sampling},
among which the flexibility of the map and the difficulty of the computation
are the biggest concerns. More importantly, we point out in Section
\ref{subsec:challenges} that there is an intrinsic limitation of
the existing methods when the target distribution has multiple modes.
In particular, we show that these methods fail even for simple distributions
such as a bimodal normal mixture.

To uncover the reason for such failures, we analyze the existing measure
transport methods using gradient flows of probability measures, also
known as Wasserstein gradient flows \citep{ambrosio2008gradient,santambrogio2017euclidean}.
A Wasserstein gradient flow can be viewed as a continuous evolution
of probability measures, and serves as a powerful tool to study the
convergence property of optimization problems involving distributions.
To address the multimodal sampling problem, we borrow ideas from the
simulated and parallel tempering algorithms \citep{kirkpatrick1983optimization,swendsen1986replica,geyer1991markov,marinari1992simulated,geyer1995annealing,neal1996sampling},
and develop a new method to estimate the transport map between a base
measure and the target multimodal distribution. We show that under
very mild conditions, the proposed algorithm converges fast and can
generate high-quality samples from the target distribution.

It is worth mentioning that the proposed sampling method greatly benefits
from recent advances in deep learning. In particular, we primarily
use deep neural networks (DNNs, \citealp{goodfellow2016deep}) to
construct transport maps between measures, due to their superior expressive
power to approximate highly nonlinear relationships. DNNs have already
been broadly applied to statistical modeling \citep{yuan2020deep,pang2020learning,qiu2021almond,sun2021consistent,liu2021density}
and sampling given training data \citep{romano2020deep,zhou2021deep},
and in this article we also show their great potentials for sampling
given energy functions. Combined with powerful optimization techniques
such as stochastic approximation \citep{robbins1951stochastic} and
adaptive gradient-based optimization \citep{kingma2014adam}, the
proposed method is scalable to complex and high-dimensional distributions,
and can be accelerated by modern computing hardware such as graphics
processing units (GPUs).

The remainder of this article is organized as follows. In Section
\ref{sec:background} we provide the background on the measure transport
framework proposed in \citet{marzouk2016sampling}, introduce the
existing samplers, discuss their issues and challenges, and describe
their connections with MCMC. Section \ref{sec:flow_sampler} is dedicated
to the theoretical framework of Wasserstein gradient flows and the
analysis of existing samplers under this framework. We propose the
new TemperFlow sampler in Section \ref{sec:temperflow}, and describe
its theoretical properties in Section \ref{sec:convergence_properties}.
Various numerical experiments are conducted in Section \ref{sec:simulation}
to demonstrate the effectiveness of the proposed sampler, and we show
its applications in modern deep generative models in Section \ref{sec:application}.
Finally, we conclude this article in Section \ref{sec:conclusion}
with discussion.

\section{Background and Related Work}

\label{sec:background}

\subsection{Measure transport}

Let $\mu:\mathcal{B}(\mathbb{R}^{d})\rightarrow[0,1]$ be the probability
measure from which we want to sample, defined over the Borel $\sigma$-algebra
on $\mathbb{R}^{d}$. In this article, we focus on sampling from continuous
distributions, so we assume that $\mu$ admits a density function
$p(x)$ with respect to $\lambda$, the Lebesgue measure on $\mathbb{R}^{d}$.
In what follows, we would use $\mu$ and $p(x)$ interchangeably to
indicate the target distribution. Let $\mu_{0}:\mathcal{B}(\mathbb{R}^{d})\rightarrow[0,1]$
be another probability measure from which we can easily generate independent
samples. In most cases, $\mu_{0}$ can be chosen to be a simple fixed
distribution such as the standard multivariate normal $N(0,I_{d})$,
and we follow this convention. A mapping $T:\mathbb{R}^{d}\rightarrow\mathbb{R}^{d}$
is said to push forward $\mu_{0}$ to $\mu$ if $\mu(A)=\mu_{0}(T^{-1}(A))$
for any set $A\in\mathcal{B}(\mathbb{R}^{d})$, denoted by $T_{\sharp}\mu_{0}=\mu$.
In this case, $T$ is also called a transport map. Throughout this
article, we assume that the transport map is invertible and differentiable,
\emph{i.e.}, $T$ is a diffeomorphism.

When both $\mu$ and $\mu_{0}$ are absolutely continuous with respect
to $\lambda$, the transport map from $\mu_{0}$ to $\mu$ always
exists, although it is not necessarily unique \citep{villani2008optimal}.
To generate a random sample from the target distribution $\mu$, one
only needs to simulate $Z_{1},\ldots,Z_{n}\overset{iid}{\sim}\mu_{0}$,
and then it follows that $T(Z_{1}),\ldots,T(Z_{n})\overset{iid}{\sim}\mu$.
Therefore, the key step of this framework is to estimate the transport
map $T$ given the energy function $E(x)$.

\citet{marzouk2016sampling} proposes to solve $T$ via the following
optimization problem
\begin{equation}
\min_{T}\ \mathrm{KL}(T_{\sharp}\mu_{0}\Vert\mu)\quad\text{subject to}\quad\det\nabla T>0,\ T\in\mathcal{T},\label{eq:map_optimization}
\end{equation}
where $\mathrm{KL}(\nu\Vert\mu)=\int\log(\mathrm{d}\nu/\mathrm{d}\mu)\mathrm{d}\nu$
is the Kullback--Leibler (KL) divergence from $\mu$ to $\nu$, $\nabla T$
is the Jacobian matrix of $T$, and $\mathcal{T}$ is a suitable set
of diffeomorphisms. In \citet{marzouk2016sampling}, $\mathcal{T}$
is the set of triangular maps based on the Knothe--Rosenblatt rearrangement,
which means that the $i$-th output variable of $T$ depends only
on the first $i$ input variables. Section \ref{subsec:flow_sampler}
shows that there are also more general specifications for $\mathcal{T}$.
Condition $\det\nabla T>0$ guarantees that the pushforward density
$T_{\sharp}\mu_{0}$ is positive on the support of $\mu$. It is worth
noting that optimizing (\ref{eq:map_optimization}) requires only
the energy function of $\mu$, since (\ref{eq:map_optimization})
is equivalent to
\begin{equation}
\min_{T}\ \int p_{0}(x)\left[E(T(x))-\log\det\nabla T(x)\right]\mathrm{d}x\quad\text{subject to}\quad\det\nabla T>0,\ T\in\mathcal{T},\label{eq:energy_optimization}
\end{equation}
where $p_{0}(x)$ is the density function of $\mu_{0}$. In practice,
$T$ is parameterized by a finite-dimensional vector $\theta$, and
the integral in (\ref{eq:energy_optimization}) can be unbiasedly
estimated by Monte Carlo samples. For example, \citet{marzouk2016sampling}
uses polynomials to approximate the $\mathcal{T}$ space, and optimizes
(\ref{eq:energy_optimization}) through stochastic gradient descent.
However, maintaining the invertibility of $T$ and the constraint
$\det\nabla T>0$ is nontrivial for polynomials, and hence the applicability
of the method in \citet{marzouk2016sampling} is severely limited.

\subsection{The neural transport samplers}

\label{subsec:flow_sampler}

To overcome such issues, \citet{hoffman2019neutra} extends \citet{marzouk2016sampling}
by using inverse autoregressive flows \citep{kingma2016improving}
to represent the transport map $T$, resulting in more scalable computation.
Inverse autoregressive flows belong to the broader class of normalizing
flows \citep{tabak2010density,tabak2013family,rezende2015variational},
which can be described as transformations of a probability density
through a sequence of invertible mappings. Normalizing flows include
many other variants such as affine coupling flows (\citealp{dinh2014nice};
\citealp{dinh2016density}), masked autoregressive flows \citep{papamakarios2017masked},
neural spline flows \citep{durkan2019neural}, and linear rational
spline flows \citep{dolatabadi2020invertible}. These models typically
use neural networks to construct invertible mappings, so they are
also referred to as invertible neural networks in the literature.
To clarify, the term ``flow'' in normalizing flows has a conceptual
gap with that in gradient flows, where the latter is the focus of
this article. Therefore, to avoid ambiguity, we use invertible neural
networks to refer to normalizing flow models hereafter. A more detailed
introduction to invertible neural networks is given in Section \ref{sec:inn}
of the supplementary material.

It should be noted that all the invertible neural networks mentioned
above are guaranteed to be diffeomorphisms that satisfy $\det\nabla T>0$,
by properly designing their neural network architectures. Therefore,
they are natural and powerful tools to construct the transport map
$T$ in (\ref{eq:map_optimization}). We use the term \emph{neural
transport sampler} to stand for measure transport sampling methods
that construct the transport map $T$ using invertible neural networks.
Since in (\ref{eq:map_optimization}) $T$ is learned by minimizing
the KL divergence, we also call such a method the KL neural transport
sampler, or KL sampler for short, when we need to emphasize this objective
function. The outline of the KL sampler is shown in Algorithm \ref{alg:kl_sampler}.

\begin{algorithm}[h]
\caption{\label{alg:kl_sampler}The KL neural transport sampler under the measure
transport framework.}


\begin{algorithmic}[1]

\REQUIRE Target distribution $\mu$ with energy function $E(x)$,
invertible neural network $T_{\theta}$ with initial parameter value
$\theta^{(0)}$, batch size $M$, step sizes $\{\alpha_{k}\}$

\ENSURE Neural network parameters $\hat{\theta}$ such that $T_{\hat{\theta}\sharp}\mu_{0}\approx\mu$

\FOR{ $k=1,2,\ldots$ }

\STATE Generate $Z_{1},\ldots,Z_{M}\overset{iid}{\sim}\mu_{0}$

\STATE Define $L(\theta)=M^{-1}\sum_{i=1}^{M}\left[E(T_{\theta}(Z_{i}))-\log\det\nabla_{z}T_{\theta}(z)|_{z=Z_{i}}\right]$ 

\STATE Set $\theta^{(k)}\leftarrow\theta^{(k-1)}-\alpha_{k}\nabla_{\theta}L(\theta)|_{\theta=\theta^{(k-1)}}$

\IF{ $L(\theta)$ converges }

\RETURN $\hat{\theta}=\theta^{(k)}$

\ENDIF

\ENDFOR

\end{algorithmic}

\end{algorithm}

\subsection{Issues and challenges}

\label{subsec:challenges}

Although the global minimizer of (\ref{eq:map_optimization}) is guaranteed
to satisfy $T_{\sharp}\mu_{0}=\mu$ provided that $\mathcal{T}$ is
rich enough, in practice the convergence may be extremely slow if
$\mu$ contains multiple isolated modes. Below, we use a motivating
example to illustrate this issue. Specifically, we apply the neural
transport sampler to two simple univariate distributions: the first
one has a unimodal and log-concave density function $p_{u}(x)=\exp\{x-\exp(x/3)\}/6$,
and the second is a mixture of two normal distributions $p_{m}(x)\sim0.7\cdot N(1,1)+0.3\cdot N(8,0.25)$.
The base measure is $\mu_{0}=N(0,1)$, $T$ is initialized as the
identity map, and the stochastic gradient descent is used to update
the parameters in $T$. Let $T^{(k)}$ denote the estimated transport
map after $k$ iterations, and then in Figure \ref{fig:demo_uni_multi_modal}
we plot the density function of $T_{\sharp}^{(k)}\mu_{0}$ with $k=0,10,50,100,500$.

\begin{figure}[h]
\begin{centering}
\includegraphics[width=0.85\textwidth]{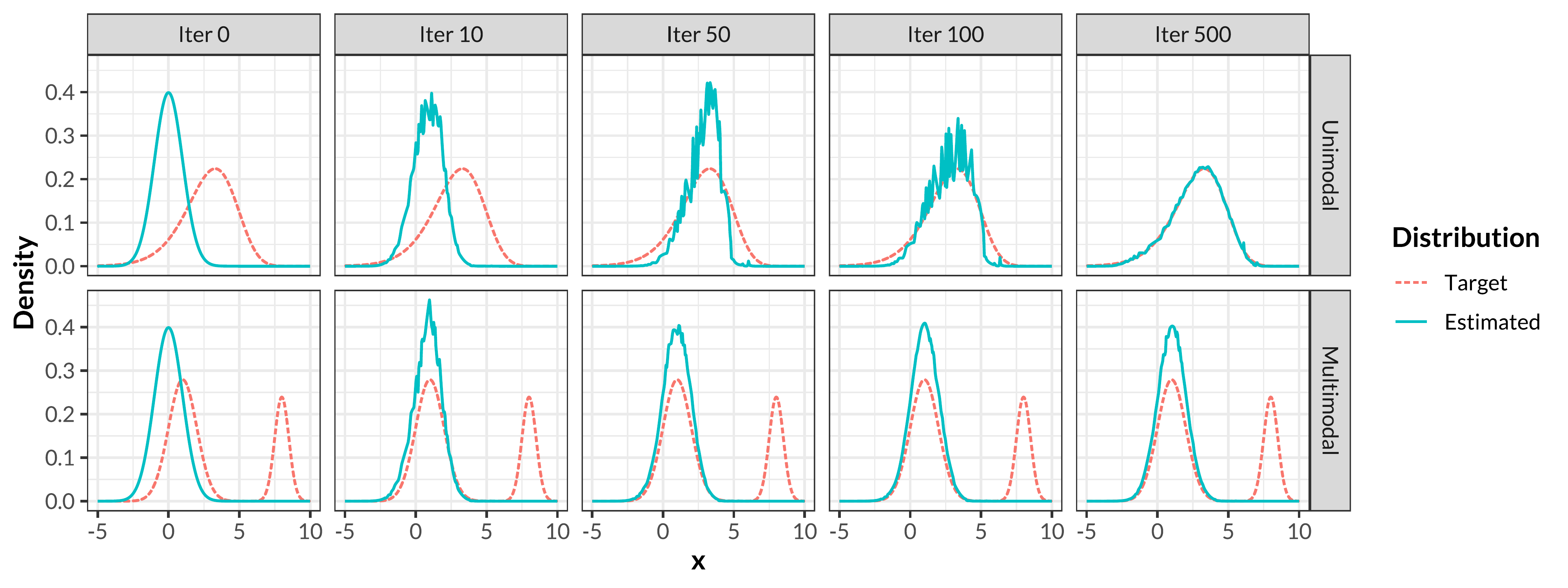}
\par\end{centering}
\caption{\label{fig:demo_uni_multi_modal}Applying the KL sampler to two distributions.
The first row shows the evolution of the estimated density for a unimodal
distribution, and the second row shows that of a bimodal distribution.
Each column stands for one iteration in the optimization process.}
\end{figure}

It is clear that for the unimodal density $p_{u}(x)$, $T_{\sharp}^{(k)}\mu_{0}$
converges very fast and approximates the target density well, whereas
in the case of bimodal normal mixture, $T_{\sharp}^{(k)}\mu_{0}$
only captures the first mode, and has very little progress after fifty
iterations. This motivating example suggests that the optimization
problem (\ref{eq:energy_optimization}) has intrinsic difficulties
when the target density $p(x)$ is multimodal. Therefore, we need
to carefully analyze the dynamics of the optimization process in order
to uncover the reason of failure.

\subsection{Connections with MCMC}

MCMC was listed as one of the top algorithms in the 20th century \citep{dongarra2000guest}.
However, for distributions that are far from being log-concave and
have many isolated modes, additional techniques are necessary for
current MCMC techniques. For example, simulated tempering \citep{marinari1992simulated}
is an attractive method of this kind. The simulated tempering swaps
between Markov chains that are different temperature variants of the
original chain. The intuition behind this is that the Markov chains
at higher temperatures can cross energy barriers more easily, and
hence they act as a bridge between isolated modes. Provable results
of this heuristic are few and far between. The lower bound of the
spectral gap for generic simulated chains is established by \citet{zheng2003swapping,woodard2009conditions},
but the spectral gap bound in \citet{woodard2009conditions} is exponentially
small as a function of the number of modes. For Langevin dynamics
\citep{bhattacharya1978criteria}, the effect of discretization for
arbitrary non-log-concave distributions has been analyzed by \citet{raginsky2017non,cheng2018sharp,vempala2019rapid},
but the mixing time is exponential in general. \citet{ge2018simulated}
combined Langevin diffusion with simulated tempering and theoretically
analyzed the Markov chain for a mixture of log-concave distributions
of the same shape. In particular, \citet{ge2018simulated} proved
fast mixing for multimodal distributions in this setting. Other related
techniques include parallel tempering \citep{geyer1991markov,falcioni1999biased}
and simulated annealing \citep{kirkpatrick1983optimization}. 

We emphasize that measure transport and MCMC are complementary to
each other. In particular, measure transport requires extensive training
to obtain the transport map, whereas MCMC does not need any training
to generate samples. But after obtaining the transport map, the generation
step becomes trivial and extremely fast for measure transport; instead,
MCMC requires running Markov chains iteratively. 

\section{Theoretical Analysis of the Neural Transport Sampler}

\label{sec:flow_sampler}

\subsection{Wasserstein gradient flow}

\label{subsec:gradient_flow}

In this section, we introduce the Wasserstein gradient flow as a tool
to analyze the neural transport sampler. Assume that the transport
map $T$ belongs to the Hilbert space $\mathbb{H}=\{T:\int\Vert T(x)\Vert^{2}\mathrm{d}\mu_{0}(x)<\infty\}$
equipped with the inner product $\langle T,M\rangle_{\mathbb{H}}=\int\langle T(x),M(x)\rangle\mathrm{d}\mu_{0}(x)$,
where $\Vert\cdot\Vert$ and $\langle\cdot,\cdot\rangle$ are the
usual Euclidean norm and inner product, respectively. For a functional
$\mathcal{G}(T)$, define its functional derivative $\delta\mathcal{G}/\delta T$
as the mapping $M:\mathbb{R}^{d}\rightarrow\mathbb{R}^{d}$ such that
$\left.\frac{\mathrm{d}}{\mathrm{d}\varepsilon}\mathcal{G}(T+\varepsilon\Psi)\right|_{\varepsilon=0}=\langle M,\Psi\rangle_{\mathbb{H}}$
for any $\Psi\in\mathbb{H}$, assuming that the relevant quantities
are well-defined. Then the main optimization problem (\ref{eq:map_optimization})
can be solved via gradient descent methods as follows. For the objective
functional $\mathcal{G}(T)=\mathrm{KL}(T_{\sharp}\mu_{0}\Vert\mu)$,
at iteration $k$ we modify the current transport map $T^{(k)}$ by
a small perturbation to obtain the new map $T^{(k+1)}=T^{(k)}-\alpha_{k}M^{(k)}$,
where $\alpha_{k}>0$ is the step size, and $M^{(k)}=(\delta\mathcal{G}/\delta T)|_{T=T^{(k)}}$
is the functional derivative of $\mathcal{G}(T)$ evaluated at $T^{(k)}$.
In Proposition \ref{prop:vector_field} we show that under mild smoothness
assumptions, $\delta\mathcal{G}/\delta T$ exists and has a simple
form.
\begin{prop}
\label{prop:vector_field}Let $p(x)$ and $q(x)$ be the density functions
of $\mu$ and $T_{\sharp}\mu_{0}$, respectively, and denote by $\Vert\cdot\Vert_{\mathrm{op}}$
the operator norm of a matrix. For any $T,\Psi\in\mathbb{H}$, if
(a) $\Vert\nabla(\Psi\circ T^{-1})(x)\Vert_{\mathrm{op}}\le C_{\Psi,T}$
and $\Vert\nabla^{2}\log p(x)\Vert_{\mathrm{op}}\le c$ for some constants
$C_{\Psi,T},c>0$ not depending on $x$; (b) $[\nabla(\log p-\log q)]\circ T\in\mathbb{H}$;
and (c) $\lim_{\Vert x\Vert\rightarrow\infty}\Vert q(x)\Psi(T^{-1}(x))\Vert=0$,
then
\[
\left.\frac{\mathrm{d}}{\mathrm{d}\varepsilon}\mathcal{G}(T+\varepsilon\Psi)\right|_{\varepsilon=0}=\left\langle \frac{\delta\mathcal{G}}{\delta T},\Psi\right\rangle _{\mathbb{H}},
\]
where $\delta\mathcal{G}/\delta T$ is the functional derivative of
$\mathcal{G}(T)$ with respect to $T$, given by $(\delta\mathcal{G}/\delta T)(x)=M(T(x))$,
$M(x)=\nabla\left[\log q(x)-\log p(x)\right]$.
\end{prop}
In Proposition \ref{prop:vector_field}, assumption (a) is used to
control the smoothness of mappings and energy functions, assumption
(b) guarantees that the computed derivative is well-defined, and assumption
(c) basically means that the density of $q$ vanishes at infinity.
They do not impose real restrictions on the form of the objects we
would analyze below.

To gain more insights about the optimization algorithm and to avoid
unnecessary technicalities, we consider the case of an infinitesimal
step size, and then the optimization process becomes a continuous-time
dynamics. Let $T^{(t)}$ be the transport map at time $t$, and then
$T^{(t)}$ evolves according to a differential equation,
\begin{equation}
\frac{\partial}{\partial t}T^{(t)}(x)=-\left(\left.\frac{\delta\mathcal{G}}{\delta T}\right|_{T=T^{(t)}}\right)(x)=\mathbf{u}_{t}(T^{(t)}(x)),\quad\mathbf{u}_{t}(x)=-\nabla\left[\log p_{t}(x)-\log p(x)\right],\label{eq:ode_map}
\end{equation}
where $p_{t}(x)$ is the density function of $T_{\sharp}^{(t)}\mu_{0}$,
and the whole process $T^{(t)}$ is typically called the \emph{gradient
flow} with respect to the functional $\mathcal{G}$.

However, the vector field $\mathbf{u}_{t}$ depends on $T^{(t)}$
indirectly via $p_{t}(x)$, making the convergence property of $T^{(t)}$
hard to analyze. Instead, we introduce the notion of Wasserstein gradient
flow in the space of probability measures. At a high level, instead
of studying the transport map $T$, we view (\ref{eq:map_optimization})
as an optimization problem with respect to the probability measure
$T_{\sharp}\mu_{0}$, and a differential equation comparable to (\ref{eq:ode_map})
is set up for this measure.

Denote by $\mathscr{P}(\mathbb{R}^{d})$ the space of Borel probability
measures on $\mathbb{R}^{d}$ with finite second moment. Theorem 8.3.1
of \citet{ambrosio2008gradient} shows that any absolutely continuous
curve $\mu_{t}:(0,T)\rightarrow\mathscr{P}(\mathbb{R}^{d})$ is a
solution to the continuity equation
\begin{equation}
\frac{\partial}{\partial t}\mu_{t}+\nabla\cdot(\mathbf{v}_{t}\mu_{t})=0\label{eq:continuity_eq}
\end{equation}
for some vector field $\mathbf{v}_{t}:\mathbb{R}^{d}\rightarrow\mathbb{R}^{d}$.
(\ref{eq:continuity_eq}) holds in the sense of distributions, \emph{i.e.},
\[
\int_{0}^{T}\int_{\mathbb{R}^{d}}\left(\partial_{t}\varphi(x,t)+\langle\mathbf{v}_{t}(x),\nabla_{x}\varphi(x,t)\rangle\right)\mathrm{d}\mu_{t}\mathrm{d}t=0,\quad\forall\varphi\in C^{\infty}(\mathbb{R}^{d}\times(0,T)),
\]
where $C^{\infty}$ is the class of infinitely differentiable functions.
Under some mild conditions, the curve $\mu_{t}$ is uniquely determined
by the vector field $\mathbf{v}_{t}$, and vice versa (\citealp{ambrosio2008gradient},
Proposition 8.1.8, Theorem 8.3.1).

Let $\mathcal{F}:\mathscr{P}(\mathbb{R}^{d})\rightarrow\mathbb{R}$
be a functional of the form $\mathcal{F}(\rho)=\int F(x,\rho(x),\nabla\rho(x))\mathrm{d}x$,
where $\rho$ is a probability measure on $\mathbb{R}^{d}$ equipped
with a differentiable density function $\rho(x)$, and $F$ is a smooth
function of its three arguments. Define the first variation of $\mathcal{F}$
at $\rho$, $(\mathrm{\delta}\mathcal{F}/\delta\nu)(\rho)$, as the
function such that $\left.\frac{\mathrm{d}}{\mathrm{d}\varepsilon}\mathcal{F}(\rho+\varepsilon\psi)\right|_{\varepsilon=0}=\int(\mathrm{\delta}\mathcal{F}/\delta\nu)(\rho)\mathrm{d}\psi$
for any perturbation $\psi$ \citep{santambrogio2017euclidean}. Then
$\mu_{t}$ is called the \emph{Wasserstein gradient flow} with respect
to the functional $\mathcal{F}$ if
\begin{equation}
\mathbf{v}_{t}=-\nabla\left(\frac{\delta\mathcal{F}}{\delta\nu}(\mu_{t})\right).\label{eq:velocity_field}
\end{equation}
To summarize, a Wasserstein gradient flow $\mu_{t}$ characterizes
the evolution of a probability measure over time by the continuity
equation (\ref{eq:continuity_eq}), whose vector field $\mathbf{v}_{t}$
is determined by a functional $\mathcal{F}$ based on equation (\ref{eq:velocity_field}).

For the measure transport problem (\ref{eq:map_optimization}), $\mathcal{F}(\cdot)=\mathcal{F}_{\mathrm{KL}}(\cdot)\coloneqq\mathrm{KL}(\cdot\Vert\mu)$.
Then Lemma 2.1 of \citet{gao2019deep} shows that for $\mathcal{F}_{\mathrm{KL}}$,
$\mathbf{v}_{t}(x)=-\nabla\left[\log p_{t}(x)-\log p(x)\right]$,
where $p_{t}(x)$ is the density function of $\mu_{t}$. It is interesting
to note that this $\mathbf{v}_{t}$ is exactly the same as $\mathbf{u}_{t}$
in (\ref{eq:ode_map}). In fact, the Wasserstein gradient flow theory
indicates that under suitable regularity conditions on $\mathbf{v}_{t}$,
there is an equivalence between the gradient flow of the transport
map $T^{(t)}$ and the Wasserstein gradient flow of the probability
measure $\mu_{t}$, in the sense that $T^{(t)}$ and $\mu_{t}$ are
connected by the relation $\mu_{t}=T_{\sharp}^{(t)}\mu_{0}$ (\citealp{ambrosio2008gradient},
Lemma 8.1.4, Proposition 8.1.8). Therefore, the behavior of $T^{(t)}$
can be understood by studying $\mu_{t}$, and vice versa. In this
paper, we focus on the latter, since it provides much convenience
in analyzing the convergence behavior of optimization algorithms.

\subsection{When the neural transport sampler fails}

\label{subsec:diagnosis}

The neural transport sampler in Algorithm \ref{alg:kl_sampler} uses
the discrete-time gradient descent method to update the sampler distribution
$\{T^{(k)}\mu_{0}\}$, which can be viewed as the discretized version
of the Wasserstein gradient flow $\mu_{t}$ with respect to the objective
functional $\mathcal{F}=\mathcal{F}_{\mathrm{KL}}$. For theoretical
analysis, we focus on continuous $\mu_{t}$, although it has to be
discretized in practical implementation. We also call $\mu_{t}$ the
sampler distribution at time $t$ for convenience.

An important fact about Wasserstein gradient flows is that the curve
$\mu_{t}$ governed by the vector field in (\ref{eq:velocity_field})
decreases the functional $\mathcal{F}$ over time (\citealp{ambrosio2008gradient},
p.233):
\begin{equation}
\frac{\mathrm{d}}{\mathrm{d}t}\mathcal{F}(\mu_{t})=-\int p_{t}(x)\Vert\mathbf{v}_{t}(x)\Vert^{2}\mathrm{d}x\le0,\label{eq:kl_derivative}
\end{equation}
implying that $\mu_{t}$ eventually converges if $\mathcal{F}\ge0$.
However, its convergence speed is significantly affected by the shape
of the target density $p(x)$. Recall that in Figure \ref{fig:demo_uni_multi_modal},
the sampler distribution accurately estimates the target distribution
$p_{u}(x)$ within a few hundreds of iterations. In fact, this good
property applies to all log-concave densities, and the rate of convergence
can be quantified. Let $\chi^{2}(\nu\Vert\mu)=\int(\mathrm{d}\nu/\mathrm{d}\mu)^{2}\mathrm{d}\mu-1$
and $H^{2}(\mu,\nu)=\frac{1}{2}\int(\sqrt{\mathrm{d}\mu/\mathrm{d}\lambda}-\sqrt{\mathrm{d}\nu/\mathrm{d}\lambda})^{2}\mathrm{d}\lambda$
be the $\chi^{2}$ divergence and the squared Hellinger distance between
$\mu$ and $\nu$, respectively. Then we have the following results.
\begin{thm}
\label{thm:kl_decay}Assume $p(x)$ is log-concave, i.e., $\log p(x)$
is concave. Let $\Sigma$ be the covariance matrix of $p(x)$, and
denote by $\sigma^{2}$ the largest eigenvalue of $\Sigma$. Then
for any $t>0$:
\end{thm}
\begin{enumerate}
\item $\mathcal{F}_{\mathrm{KL}}(\mu_{t})\le e^{-t/(2\sigma^{2}C(d))}\chi^{2}(\mu_{0}\Vert\mu)$,
where $C(d)=\exp\left\{ c\sqrt{l_{d}\cdot[\log(l_{d})+1]}\right\} $,
$l_{d}=\log(d)+1$, and $c>0$ is a universal constant.
\item ${\displaystyle \left|\frac{\mathrm{d}\mathcal{F}_{\mathrm{KL}}(\mu_{t})}{\mathrm{d}t}\right|}\ge[\sigma^{2}C(d)]^{-1}\cdot\left[2-H^{2}(\mu_{t},\mu)\right]\cdot H^{2}(\mu_{t},\mu)\ge[\sigma^{2}C(d)]^{-1}H^{2}(\mu_{t},\mu)$.
\end{enumerate}
The first part of Theorem \ref{thm:kl_decay} shows that for log-concave
densities, the difference between the sampler distribution and the
target distribution decays exponentially fast along the gradient flow,
which explains the rapid convergence of the neural transport sampler
for the unimodal distribution in Figure \ref{fig:demo_uni_multi_modal}.
The second part provides a more delicate interpretation of this phenomenon.
Specifically, it shows that the time derivative $\mathrm{d}\mathcal{F}_{\mathrm{KL}}(\mu_{t})/\mathrm{d}t$
does not vanish as long as $\mu_{t}$ is different from $\mu$, which
we call the \emph{non-vanishing gradient} property. Intuitively, it
means that the velocity field $\mathbf{v}_{t}$ unceasingly pushes
the current distribution $\mu_{t}$ closer to $\mu$, and the rate
of change, measured by $|\mathrm{d}\mathcal{F}_{\mathrm{KL}}(\mu_{t})/\mathrm{d}t|$,
would not be too small, unless $\mu_{t}$ is already close to $\mu$.
If the latter situation happens, then we have essentially achieved
the goal of sampling. As a consequence, for log-concave densities,
a small derivative in magnitude implies a small Hellinger distance
between $\mu_{t}$ and $\mu$.

However, when the target density contains multiple modes, the non-vanishing
gradient property may no longer hold. We consider a simple yet insightful
model to illustrate this issue. Let $h(x)=e^{-V(x)}$ be a base density
function, and define the target mixture density as $p(x)=\alpha h(x)+(1-\alpha)h(x-\mu)$,
where $\alpha\in(0,1)$ is the mixing probability, and $\Vert\mu\Vert$
quantifies the distance between the two components. Suppose that at
some time $t^{*}$, the sampler distribution $\mu_{t^{*}}$ also takes
the form of a mixture but with a different mixing probability $\gamma\neq\alpha$,
\emph{i.e.}, $p_{t^{*}}(x)=\gamma h(x)+(1-\gamma)h(x-\mu)$. Then
we will show that the gradient flow $\mu_{t}$ exhibits an almost
opposite property to the log-concave case, and the conclusion can
even be generalized to a much broader class of objective functionals.
Define a general $\phi$-divergence-based functional $\mathcal{F}_{\phi}(\cdot)=\mathcal{D}_{\phi}(\cdot\Vert\mu)$,
where $\mathcal{D}_{\phi}(\nu\Vert\mu)=\int\phi(\mathrm{d}\nu/\mathrm{d}\mu)\mathrm{d}\mu$,
and $\phi$ is a convex function with $\phi(1)=0$. It is easy to
show that KL divergence is a special case of $\phi$-divergence with
$\phi(x)=x\log(x)$. Then we have the following result.
\begin{thm}
\label{thm:f_divergence_grad}Let $X$ be a random vector following
$h(x)=e^{-V(x)}$, and define $Y_{\mu,1}=V(X-\mu)-V(X)$, $Y_{\mu,2}=V(X)-V(X+\mu)$.
Assume that (a) for any fixed $x\in\mathbb{R}^{d}$, $h(x-\mu)\rightarrow0$
as $\Vert\mu\Vert\rightarrow\infty$; (b) there exist constants $c_{1},C_{1},k\ge0$
such that $\Vert\nabla V(x)-\nabla V(y)\Vert\le C_{1}\Vert x-y\Vert^{k}$
for all $x,y\in\mathbb{R}^{d}$ satisfying $\Vert x-y\Vert\ge c_{1}$;
(c) there exists a constant $C_{2}>0$ such that $C_{1}\Vert\mu\Vert^{k}\cdot P(\Vert\mu\Vert^{-1}|Y_{\mu,i}|\le C_{2})\rightarrow0$
as $\Vert\mu\Vert\rightarrow\infty$, $i=1,2$. Then we have the following.
\end{thm}
\begin{enumerate}
\item ${\displaystyle \lim_{\Vert\mu\Vert\rightarrow\infty}\mathcal{F}_{\phi}(\mu_{t^{*}})=\alpha\phi\left(\frac{\gamma}{\alpha}\right)+(1-\alpha)\phi\left(\frac{1-\gamma}{1-\alpha}\right)}>0$.
\item ${\displaystyle \lim_{\Vert\mu\Vert\rightarrow\infty}\left.\frac{\mathrm{d}\mathcal{F}_{\phi}(\mu_{t})}{\mathrm{d}t}\right|_{t=t^{*}}=0}$.
\end{enumerate}
Theorem \ref{thm:f_divergence_grad} indicates that when the two components
of $p(x)$ are distantly isolated, there exists a configuration of
$\mu_{t}$ such that it is different from the target distribution,
but meanwhile the gradient vanishes. Since the gradient characterizes
how fast the current objective functional is going to decrease at
an instant, the vanishing gradient implies an extremely slow convergence
speed. Theorem \ref{thm:f_divergence_grad} also explains the multimodal
example in Figure \ref{fig:demo_uni_multi_modal}. In the 100th iteration,
the sampler distribution concentrates on the first mode, leading to
a mixture distribution $p_{t^{*}}(x)$ with $\gamma\approx1$, and
hence Theorem \ref{thm:f_divergence_grad} holds. Furthermore, since
the result applies to a general $\phi$-divergence, the intrinsic
limitation of the neural transport sampler cannot be easily fixed
by simply changing the divergence measure.

In summary, Theorem \ref{thm:kl_decay} and Theorem \ref{thm:f_divergence_grad}
show that the shape of the target density greatly affects the convergence
speed of the neural transport sampler. Unfortunately, most target
densities in real-world problems are not log-concave, so the neural
transport sampler needs to be modified to adapt to more sophisticated
sampling problems.

\section{The TemperFlow Sampler}

\label{sec:temperflow}

\subsection{Overview}

\label{subsec:temperflow_overview}

The analysis in Section \ref{subsec:diagnosis} indicates that the
Wasserstein gradient flow $\mu_{t}$ with respect to the $\phi$-divergence
$\mathcal{D}_{\phi}(\cdot\Vert\mu)$ is inefficient to approximate
the target distribution $\mu$ when $p(x)$ is multimodal, due to
the phenomenon of vanishing gradients. Therefore, it is crucial to
find an alternative curve that quickly converges to $\mu$. Inspired
by the simulated annealing methods in optimization and MCMC, consider
the tempered density function of $\mu$, defined by
\begin{equation}
q_{\beta_{t}}(x)=\frac{1}{Z(\beta_{t})}e^{-\beta_{t}E(x)},\label{eq:tempering_curve}
\end{equation}
where $\beta_{t}$ is an increasing and continuous function of $t$
satisfying $\beta_{0}>0$ and $\lim_{t\rightarrow\infty}\beta_{t}=1$,
and $Z(\beta_{t})$ is a normalizing constant to make $q_{\beta_{t}}(x)$
a density function. In the simulated annealing literature, $\beta_{t}$
is often called the inverse temperature parameter. Let $\rho_{t}$
be the curve of probability measure corresponding to $q_{\beta_{t}}(x)$,
and then it is easy to find that $\rho_{t}$ converges to $\mu$ as
$t\rightarrow\infty$, since $q_{\beta_{t}}(x)$ reduces to $p(x)$
if $\beta_{t}=1$.

The tempering curve $\rho_{t}$ has several benign properties. First,
for any $t>0$, $q_{\beta_{t}}(x)$ has the same local maxima as $p(x)$,
so all modes of $p(x)$ are preserved along the curve $\rho_{t}$.
This is important since the loss of modes is a common issue of many
existing sampling methods, as will be demonstrated in numerical experiments.
Second, for $t<s$, $q_{\beta_{t}}(x)$ is a ``flatter'' distribution
than $q_{\beta_{s}}(x)$ with better connectivity between modes, so
it is typically much easier to sample from an intermediate distribution
$q_{\beta_{t}}(x)$ than from $p(x)$ directly. Third, similar to
the Wasserstein gradient flow, the tempering curve $\rho_{t}$ also
decreases the KL divergence relative to the target distribution, as
shown in Theorem \ref{thm:monotonicity_kl}.
\begin{thm}
\label{thm:monotonicity_kl}Let $\rho_{t}$ be the probability measure
corresponding to $q_{\beta_{t}}(x)$ as defined in (\ref{eq:tempering_curve}),
and then $\mathrm{d}\mathcal{F}_{\mathrm{KL}}(\rho_{t})/\mathrm{d}t\le0$.
The equal sign holds if and only if $\rho_{t}=\mu$.
\end{thm}
However, tracing the exact path of $\rho_{t}$ is intractable, so
a more realistic approach is to let the sampler learn a sequence of
distributions $\{r_{k}\}$ extracted from the tempering curve $\rho_{t}$.
This is done by discretizing $\beta_{t}$ into a sequence $\{\beta_{k}\}$,
and defining
\begin{equation}
r_{k}(x)=q_{\beta_{k}}(x),\ k=0,1,\ldots,K,\ \beta_{0}<\beta_{1}<\cdots<\beta_{K}=1.\label{eq:tempering_discretize}
\end{equation}
Then we seek transport maps to connect each $r_{k-1}$ with $r_{k}$.
We call the sequence $\{r_{k}\}$ a tempered distribution flow for
$p(x)$ to reflect such an idea, and name our proposed algorithm the
TemperFlow sampler for brevity. The core idea of TemperFlow is to
learn a sequence of transport maps $T^{(k)},k=0,1,\ldots,K$ such
that $T_{\sharp}^{(k)}\mu_{0}\approx r_{k}$, where each $T^{(k)}$
is optimized using $T^{(k-1)}$ as a warm start. The outline of the
TemperFlow sampler is given in Algorithm \ref{alg:temperflow}, in
which the subroutine \texttt{kl\_sampler} applies Algorithm \ref{alg:kl_sampler}
to the tempered energy function $\beta_{0}E(x)$, and \texttt{l2\_sampler}
and \texttt{estimate\_beta} will be introduced in subsequent sections.

\begin{algorithm}[h]
\caption{\label{alg:temperflow}The outline of the TemperFlow sampler.}


\begin{algorithmic}[1]

\REQUIRE Target distribution $\mu$ with energy function $E(x)$,
invertible neural network $T_{\theta}$ parameterized by $\theta$,
initial inverse temperature $\beta_{0}$

\ENSURE Estimated transport map $\hat{T}$ such that $\hat{T}_{\sharp}\mu_{0}\approx\mu$

\STATE Initialize $\theta\leftarrow\theta^{(-1)}$ such that $T^{(-1)}\coloneqq T_{\theta^{(-1)}}\approx I$,
the identity map

\STATE Compute $\theta^{(0)}\leftarrow\mathtt{kl\_sampler}(\beta_{0}E,\theta^{(-1)})$\hfill{}(Algorithm
\ref{alg:kl_sampler})

\FOR{ $k=1,2,\ldots$ }

\STATE Compute $\beta_{k}\leftarrow\min\{1,\mathtt{estimate\_beta}(\beta_{k-1},\theta^{(k-1)})\}$\hfill{}(Algorithm
\ref{alg:estimate_beta})

\STATE Compute $\theta^{(k)}\leftarrow\mathtt{l2\_sampler}(\beta_{k}E,\theta^{(k-1)})$\hfill{}(Algorithm
\ref{alg:l2_flow_sampler})

\IF{ $\beta_{k}=1$ }

\RETURN $\hat{T}=T_{\theta^{(k)}}$

\ENDIF

\ENDFOR

\end{algorithmic}

\end{algorithm}

\subsection{Transport between temperatures}

\label{subsec:transport_temperatures}

The problem of smoothly changing $T^{(k-1)}$ to $T^{(k)}$ is the
central part of TemperFlow, and can be viewed as finding a curve connecting
$r_{k-1}$ to $r_{k}$. Theoretically, the existing KL sampler in
Algorithm \ref{alg:kl_sampler} can be used, but the analysis in Section
\ref{subsec:diagnosis} implies that it may suffer from the vanishing
gradient problem for multimodal distributions. Instead, we propose
a novel method called $L^{2}$ sampler, which profoundly improves
the sampling quality.

The idea behind the $L^{2}$ sampler is simple. Instead of minimizing
the KL divergence between the sampler distribution and the target
distribution as in Algorithm \ref{alg:kl_sampler}, the $L^{2}$ sampler
uses the squared $L^{2}$-distance between the two densities as the
objective function, defined as $\mathcal{L}(g)=\int(g-f)^{2}\mathrm{d}\lambda$,
where $f$ is the target density, and $g$ stands for the sampler
density. Given a parameterized transport map $T_{\theta}$, typically
an invertible neural network, the $L^{2}$ sampler minimizes $\mathcal{L}(p_{\theta})$
using stochastic gradient descent, where $p_{\theta}$ is the density
function of $T_{\theta\sharp}\mu_{0}$. We defer the theoretical analysis
of the $L^{2}$ sampler to Section \ref{sec:convergence_properties},
and here focus on the computational aspect.

The main challenge of computing $\mathcal{L}(g)$ is evaluating the
density function $f$, since it is typically given in the unnormalized
form $f(x)=e^{-E(x)}/U$, where $U$ is the unknown normalizing constant.
Fortunately, $U$ can be efficiently estimated by importance sampling
\[
U=\int\exp\{-E(x)\}\mathrm{d}x=\mathbb{E}_{X\sim h(x)}\exp\{-E(X)-\log h(X)\},
\]
where $h\approx f$ is a proposal distribution close to the target,
and the exact expectation can be approximated by Monte Carlo samples.
In TemperFlow, we take $f=r_{k}$ and $h=r_{k-1}$, so the estimator
for $U$ does not suffer from a high variance as long as the adjacent
distributions in the tempered distribution flow $\{r_{k}\}$ are close
to each other. Then we have
\begin{align*}
\mathcal{L}(g) & =\int[g(x)-f(x)]^{2}\mathrm{d}x=\int[g(x)]^{2}\mathrm{d}x-2\int f(x)g(x)\mathrm{d}x+\int[f(x)]^{2}\mathrm{d}x\\
 & =\mathbb{E}_{X\sim g(x)}\left[g(X)-2f(X)\right]+\int[f(x)]^{2}\mathrm{d}x,
\end{align*}
where the last term does not depend on $g$. The full algorithm is
given in Algorithm \ref{alg:l2_flow_sampler}.

\begin{algorithm}[h]
\caption{\label{alg:l2_flow_sampler}The $L^{2}$ sampler.}


\begin{algorithmic}[1]

\REQUIRE Target distribution $f$ with energy function $E(x)$, invertible
neural network $T_{\theta}$ with initial parameter value $\theta^{(0)}$,
batch size $M$, step sizes $\{\alpha_{k}\}$

\ENSURE Neural network parameters $\hat{\theta}$ such that $T_{\hat{\theta}\sharp}\mu_{0}\approx f$

\STATE Let $h(x)$ be the density function of $T_{\theta^{(0)}\sharp}\mu_{0}$

\STATE Generate $Z_{1},\ldots,Z_{M}\overset{iid}{\sim}\mu_{0}$ and
set $X_{i}\leftarrow T_{\theta^{(0)}}(Z_{i})$, $i=1,\ldots,M$

\STATE Set $\hat{U}\leftarrow M^{-1}\sum_{i=1}^{M}\exp\{-E(X_{i})-\log h(X_{i})\}$

\FOR{ $k=1,2,\ldots$ }

\STATE Let $g_{\theta}(x)$ be the density function of $T_{\theta\sharp}\mu_{0}$

\STATE Generate $Z_{1},\ldots,Z_{M}\overset{iid}{\sim}\mu_{0}$ and
define $X_{i}=T_{\theta}(Z_{i})$, $i=1,\ldots,M$

\STATE Define $L(\theta)=M^{-1}\sum_{i=1}^{M}\left[g_{\theta}(X_{i})-2\exp\{-E(X_{i})\}/\hat{U}\right]$

\STATE Set $\theta^{(k)}\leftarrow\theta^{(k-1)}-\alpha_{k}\nabla_{\theta}L(\theta)|_{\theta=\theta^{(k-1)}}$

\IF{ $L(\theta)$ converges }

\RETURN $\hat{\theta}=\theta^{(k)}$

\ENDIF

\ENDFOR

\end{algorithmic}

\end{algorithm}

\begin{rem*}
In practical implementation, it is helpful to optimize the logarithm
of $\mathcal{L}(g)$ to enhance numerical stability, especially when
the dimension of $f$ is high. In this case, we can apply importance
sampling again to obtain
\begin{align*}
\log[\mathcal{L}(g)] & =\log\left[\mathbb{E}_{X\sim h(x)}\frac{\left[g(X)-U^{-1}\exp\{-E(X)\}\right]^{2}}{h(X)}\right]\approx\log\left[M^{-1}\sum_{i=1}^{M}\exp\{W_{i}\}\right],\\
W_{i} & =2\log\left|g(X_{i})-U^{-1}\exp\{-E(X_{i})\}\right|-\log h(X_{i})\\
 & =2\log g(X_{i})-\log h(X_{i})+2\log\left|1-U^{-1}\exp\{-E(X_{i})-\log g(X_{i})\}\right|,
\end{align*}
where $X_{i}\overset{iid}{\sim}h(x)$.
\end{rem*}

\subsection{Adaptive temperature selection}

\label{subsec:temperature_selection}

There remains the problem of how to choose the $\beta_{k}$ sequence
in (\ref{eq:tempering_discretize}). Obviously, using a dense sequence
of $\beta_{k}$ makes it easier to transport from $r_{k-1}$ to $r_{k}$,
but it also increases the computing time. On the other hand, a coarse
grid of $\beta_{k}$ points reduces the number of transitions, but
results in harder transports. 

For parallel tempering MCMC, there were extensive discussions on the
choice of temperatures for parallel chains, and one widely-used criterion
is to make the acceptance probability of swapping uniform across different
chains \citep{kofke2002acceptance,kone2005selection,earl2005parallel}.
This leads to a geometrically spaced ladder based on calculations
on normal distributions, and there are also adaptive schemes to dynamically
adjust temperatures \citep{vousden2016dynamic}. These heuristics,
however, are not directly applicable to TemperFlow, since there is
no acceptance probability in our framework.

For TemperFlow, we develop an adaptive scheme to determine $\beta_{k+1}$
given the current $\beta_{k}$. As a starting point, Theorem \ref{thm:monotonicity_kl}
indicates that the tempering curve decreases the KL divergence over
time, so a natural way to stabilize the optimization process is to
select $\beta_{k}$'s that reduce the KL divergence smoothly. Let
$\ell(\beta)=\mathrm{KL}(q_{\beta}\Vert\mu)$, and Theorem \ref{thm:monotonicity_kl}
shows that $\ell(\beta)$ is a decreasing function of $\beta$. Suppose
that we are given $r_{k}$ and $\beta_{k}$, and a natural selection
of $\beta_{k+1}$ is such that $\ell(\beta_{k+1})=\alpha\ell(\beta_{k})$,
where $0<\alpha<1$ is a predefined discounting factor, \emph{e.g.},
0.9. Of course, $\ell(\beta_{k+1})$ is unknown, but by Taylor expansion
we have
\[
\ell(\beta_{k+1})-\ell(\beta_{k})\approx\left.\frac{\mathrm{d}\ell(\beta)}{\mathrm{d}(\log\beta)}\right|_{\beta=\beta_{k}}\left(\log\beta_{k+1}-\log\beta_{k}\right).
\]
With the reparameterization $\gamma=\log(\beta)$ and the rearrangement
of terms, we get
\begin{equation}
\gamma_{k+1}\approx\gamma_{k}-(1-\alpha)\left.\left\{ \frac{\mathrm{d}[\log\ell(e^{\gamma})]}{\mathrm{d}\gamma}\right\} ^{-1}\right|_{\gamma=\gamma_{k}}.\label{eq:log_beta_update}
\end{equation}
In the supplementary material, we show that
\begin{equation}
\frac{\mathrm{d}[\log\ell(e^{\gamma})]}{\mathrm{d}\gamma}=\frac{-\beta(1-\beta)\left\{ \int q_{\beta}(x)E^{2}(x)\mathrm{d}x-\left[\int q_{\beta}(x)E(x)\mathrm{d}x\right]^{2}\right\} }{\int q_{\beta}(x)\left[\log q_{\beta}(x)+E(x)\right]\mathrm{d}x+\log\int q_{\beta}(x)\exp\left[-E(x)-\log q_{\beta}(x)\right]\mathrm{d}x},\label{eq:delta_gamma}
\end{equation}
so (\ref{eq:log_beta_update}) can be estimated based on $r_{k}$
and $\beta_{k}$. The overall method of computing $\beta$ is given
in Algorithm \ref{alg:estimate_beta}.

\begin{algorithm}[h]
\caption{\label{alg:estimate_beta}Adaptive selection of $\beta$ parameter
in the TemperFlow sampler.}


\begin{algorithmic}[1]

\REQUIRE Current $\beta_{k}$ and neural network parameter value
$\theta^{(k)}$, batch size $M$, discounting factor $\alpha$

\ENSURE Next parameter $\beta_{k+1}$

\STATE Generate $Z_{1},\ldots,Z_{M}\overset{iid}{\sim}\mu_{0}$ and
let $X_{i}=T_{\theta^{(k)}}(Z_{i})$, $i=1,\ldots,M$

\STATE Set $c_{1}\leftarrow M^{-1}\sum_{i=1}^{M}E(X_{i})$, $c_{2}\leftarrow M^{-1}\sum_{i=1}^{M}[E(X_{i})]^{2}$

\STATE Compute $U_{i}=\log p_{\theta}(X_{i})+E(X_{i})$, $i=1,\ldots,M$

\STATE Set $c_{3}\leftarrow M^{-1}\sum_{i=1}^{M}U_{i}$, $c_{4}\leftarrow\log\left[M^{-1}\sum_{i=1}^{M}e^{-U_{i}}\right]$

\STATE Set $\gamma^{(k+1)}\leftarrow\log\beta_{k}+(1-\alpha)\beta_{k}^{-1}(1-\beta_{k})^{-1}(c_{3}+c_{4})/(c_{2}-c_{1}^{2})$

\RETURN $\beta_{k+1}=\exp\{\gamma^{(k+1)}\}$

\end{algorithmic}

\end{algorithm}

\section{Convergence Properties of TemperFlow}

\label{sec:convergence_properties}

As is introduced in Section \ref{subsec:temperflow_overview}, the
tempered distribution flow $\{r_{k}\}$ can be viewed as a discretized
version of the tempering curve $\rho_{t}$, whose convergence property
is given by Theorem \ref{thm:monotonicity_kl}. To connect the adjacent
distributions in the sequence $\{r_{k}\}$, we have proposed the $L^{2}$
sampler in Section \ref{subsec:transport_temperatures}, which is
the key to studying the theoretical properties of TemperFlow. The
Wasserstein gradient flow with respect to the squared $L^{2}$-distance
has been used to estimate generative models from data \citep{gao2022deep},
and below we show that it can also be used to analyze the $L^{2}$
sampler. In particular, we theoretically show that this Wasserstein
gradient flow does not suffer from the vanishing gradient problem.

Let $g_{t}$ be the Wasserstein gradient flow with respect to $\mathcal{L}$.
Assume that the target probability measure $f$ and every $g_{t}$
are absolutely continuous with respect to the Lebesgue measure $\lambda$,
so they admit densities. For simplicity, we use the same symbol $f$
or $g_{t}$ to represent both the probability measure and the density
function. Assume that $f$ and $g_{t}$ are differentiable, supported
on $\mathbb{R}^{d}$, and belong to the $L^{2}(\mathbb{R}^{d})$ space,
so that $\int f^{2}\mathrm{d}\lambda<\infty$ and $\int g_{t}^{2}\mathrm{d}\lambda<\infty$.
Denote by $\mathcal{Q}$ the set of balls in $\mathbb{R}^{d}$, and
for a measurable set $A\subset\mathbb{R}^{d}$, let $|A|=\lambda(A)$.
To analyze the property of $g_{t}$, we first make the following assumptions:
\begin{assumption}
\label{assu:balls} For every $\varepsilon\in(0,\varepsilon_{0}]$,
where $\varepsilon_{0}<1$ is a fixed constant, there exist mutually
exclusive balls $Q_{1}^{\varepsilon},\ldots,Q_{K}^{\varepsilon}\in\mathcal{Q}$,
possibly depending on $\varepsilon$, such that for all $i=1,\ldots,K$:
\end{assumption}
\begin{enumerate}
\item[(a)] $\inf_{Q_{i}^{\varepsilon}}f\le\varepsilon^{1/4}$ and $\int_{Q^{\varepsilon}}f\mathrm{d}\lambda\ge1-\varepsilon^{1/4}$,
where $Q^{\varepsilon}=\bigcup_{i=1}^{K}Q_{i}^{\varepsilon}$.
\item[(b)] $\sup\,_{\begin{subarray}{c}
Q\subset Q_{i}^{\varepsilon}\\
Q\in\mathcal{Q}
\end{subarray}}\,|Q|^{2/d}\left[\frac{1}{|Q|}\int_{Q}f^{-2}\mathrm{d}\lambda\right]^{1/2}\le M_{1}\varepsilon^{-1/4}[\log(1/\varepsilon)]^{c}$ holds for some $c,M_{1}>0$.
\item[(c)] $\int_{Q_{i}^{\varepsilon}}\left[\varepsilon^{1/4}-f(x)\right]_{+}^{2}\mathrm{d}x\ge M_{2}\varepsilon^{3/4}[\log(1/\varepsilon)]^{c}$
holds for sufficiently large $M_{2}>0$.
\end{enumerate}
Intuitively, Assumption \ref{assu:balls}(a) means that we can use
$K$ balls to cover the modes of the target distribution, with only
minimal mass outside the balls. Assumptions \ref{assu:balls}(b) and
(c) are technical conditions for the tail behavior of each mode. As
a starting point, we first show that the standard normal distribution
satisfies Assumption \ref{assu:balls}.
\begin{prop}
\label{prop:assump1_normal}Let $f(x)=(2\pi)^{-1/2}e^{-x^{2}/2}$
be the standard normal distribution. Then for some $\varepsilon_{0}<1$
and every $\varepsilon\in(0,\varepsilon_{0}]$, $Q^{\varepsilon}=[-x_{\varepsilon}-M\varepsilon^{1/12},x_{\varepsilon}+M\varepsilon^{1/12}]$
satisfies Assumption \ref{assu:balls}, where $x_{\varepsilon}=\sqrt{-(\log\varepsilon)/2-\log(2\pi)}$
and $M$ is a sufficiently large constant.
\end{prop}
Proposition \ref{prop:assump1_normal} can be easily extended to more
complicated distributions. In fact, Assumption \ref{assu:balls} is
mostly concerned with the tail behavior of the distribution, so a
density function can be modified within a compact set without violating
the assumptions. Moreover, for multimodal distributions, each mode
can be covered by a ball with suitable size as long as the balls are
disjoint, or multiple modes are covered by a single ball, depending
on the shape of the density function.
\begin{assumption}
\label{assu:lower_bound}There exists a constant $\alpha>0$ such
that $g_{t}(x)\ge\alpha f(x)$ for all $x\in Q^{\varepsilon}$ and
$t>0$.
\end{assumption}
Assumption \ref{assu:lower_bound} essentially requires that the sampler
density has the same support as the target density. This assumption
is natural in a tempered distribution flow, since in that case $g_{t}$
is typically a flatter distribution than $f$. Under these conditions,
we show that the gradient of $\mathcal{L}(g_{t})$ does not vanish.
\begin{thm}
\label{thm:grad_l2}Suppose that Assumptions \ref{assu:balls} and
\ref{assu:lower_bound} hold. Then there exists a constant $C>0$
such that whenever $|\mathrm{d}\mathcal{L}(g_{t})/\mathrm{d}t|\coloneqq\varepsilon_{t}\le\varepsilon_{0}$,
we have
\begin{equation}
\left|\frac{\mathrm{d}\mathcal{L}(g_{t})}{\mathrm{d}t}\right|=\varepsilon_{t}\ge C\cdot\max_{i=1,\ldots,K}\left\{ \left[\int_{Q_{i}^{\varepsilon_{t}}}(g_{t}-f)^{2}\mathrm{d}\lambda\right]^{2}\cdot\min(|Q_{i}^{\varepsilon_{t}}|,|Q_{i}^{\varepsilon_{t}}|^{-1})\right\} .\label{eq:l2_gradient}
\end{equation}
\end{thm}
Proposition \ref{prop:assump1_normal} implies that the $|Q_{i}^{\varepsilon_{t}}|$
term in (\ref{eq:l2_gradient}) is typically at the order of $O(\sqrt{\log(1/\varepsilon_{t})})$,
so in most cases it can be viewed as a constant. In addition, since
the union of $Q_{i}^{\varepsilon_{t}}$'s captures almost all the
mass of $f$ by Assumption \ref{assu:balls}, the right hand side
of (\ref{eq:l2_gradient}) can be viewed as the squared $L^{2}$-distance
between $f$ and $g_{t}$ up to some constant. As a whole, Theorem
\ref{thm:grad_l2} implies that the $L^{2}$ Wasserstein gradient
flow enjoys the non-vanishing gradient property.

To visually demonstrate the difference between the KL-based and $L^{2}$-based
samples, we use both methods to sample from the bimodal distribution
$p_{m}$ in Figure \ref{fig:demo_uni_multi_modal}, with the initial
sampler distribution set to $0.9\cdot N(1,1)+0.1\cdot N(8,0.25)$.
As can be seen from Figure \ref{fig:demo_l2}, the behavior of the
KL sampler is predicted by Theorem \ref{thm:f_divergence_grad}, with
very little progress even after 100 iterations. On the contrary, the
$L^{2}$ sampler successfully estimates the target distribution after
50 steps.

\begin{figure}[h]
\begin{centering}
\includegraphics[width=0.85\textwidth]{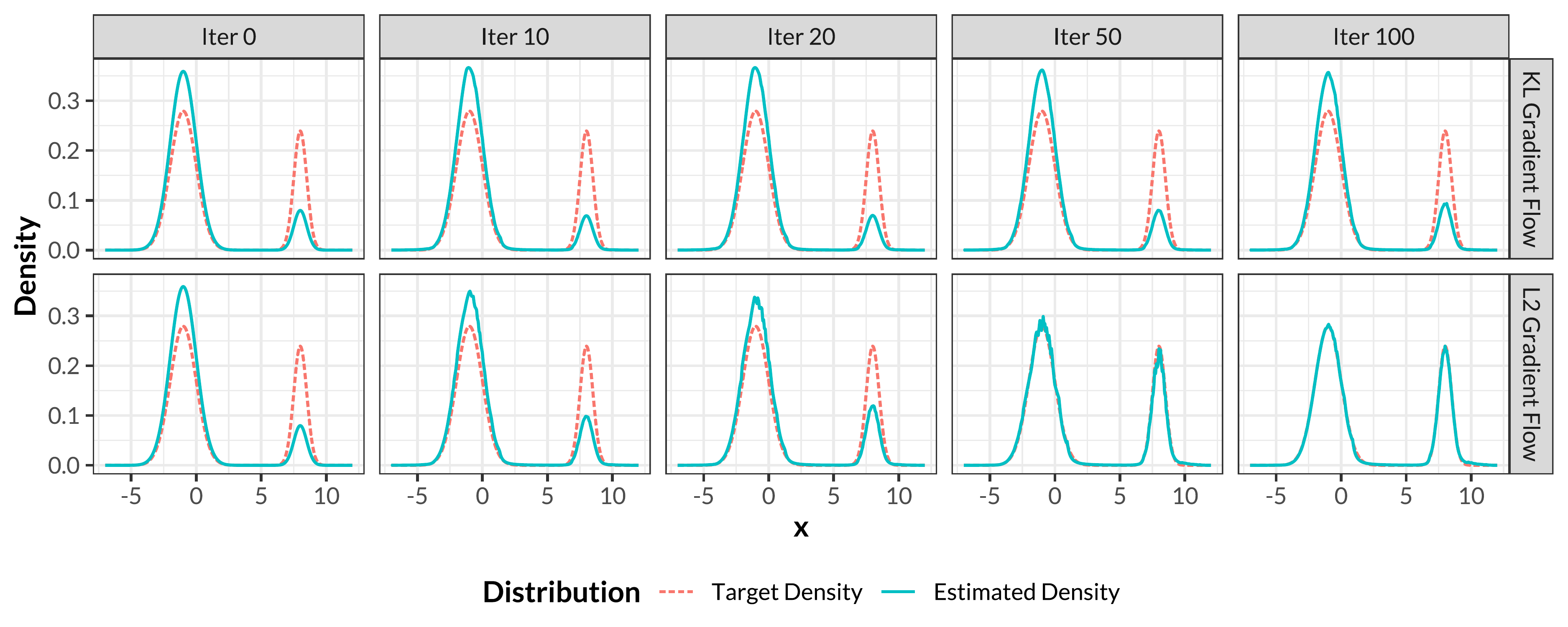}
\par\end{centering}
\caption{\label{fig:demo_l2}A comparison of the KL-based and $L^{2}$-based
samplers. The target density is $p_{m}(x)$ in Section \ref{subsec:challenges},
and the initial sampler distribution has the same components as $p_{m}(x)$
but with a different mixing probability. Each column stands for one
iteration in the optimization process.}
\end{figure}

\section{Simulation Study}

\label{sec:simulation}

In this section, we use simulation experiments to compare TemperFlow
with widely-used MCMC samplers, including the Metropolis--Hastings
algorithm \citep{metropolis1953equation,hastings1970monte}, Hamiltonian
Monte Carlo \citep{brooks2011handbook}, and parallel tempering. We
also consider a variant of the TemperFlow algorithm, which post-processes
the TemperFlow samples by a rejection sampling refinement. Since our
focus is on multimodal distributions, we first consider a number of
bivariate Gaussian mixture models, in which the modes can be easily
visualized, and then study multivariate distributions that have arbitrary
dimensions based on copulas. Details on the hyperparameter setting
of algorithms are given in Section \ref{subsec:hyperparameter} of
the supplementary material. Additional experiments to validate the
proposed method on other aspects are given in Section \ref{sec:additional_exp}.

\subsection{Gaussian mixture models}

\label{subsec:gmm}

We study three Gaussian mixture models shown in the last column of
Figure \ref{fig:experiments_2d}. For each method, we compute the
sampling errors based on the 1-Wasserstein distance and the maximum
mean discrepancy (MMD, \citealp{gretton2006kernel}) between the generated
points and the true sample. The precise definitions of the metrics,
which we call the adjusted Wasserstein distance and the adjusted MMD,
are given in Section \ref{subsec:def_metrics} of the supplementary
material. These two metrics can take negative values, and the principle
is that smaller values indicate higher sampling quality. Figure \ref{fig:experiments_2d}
shows the result of one simulation run.

\begin{figure}[h]
\begin{centering}
\includegraphics[width=0.83\textwidth]{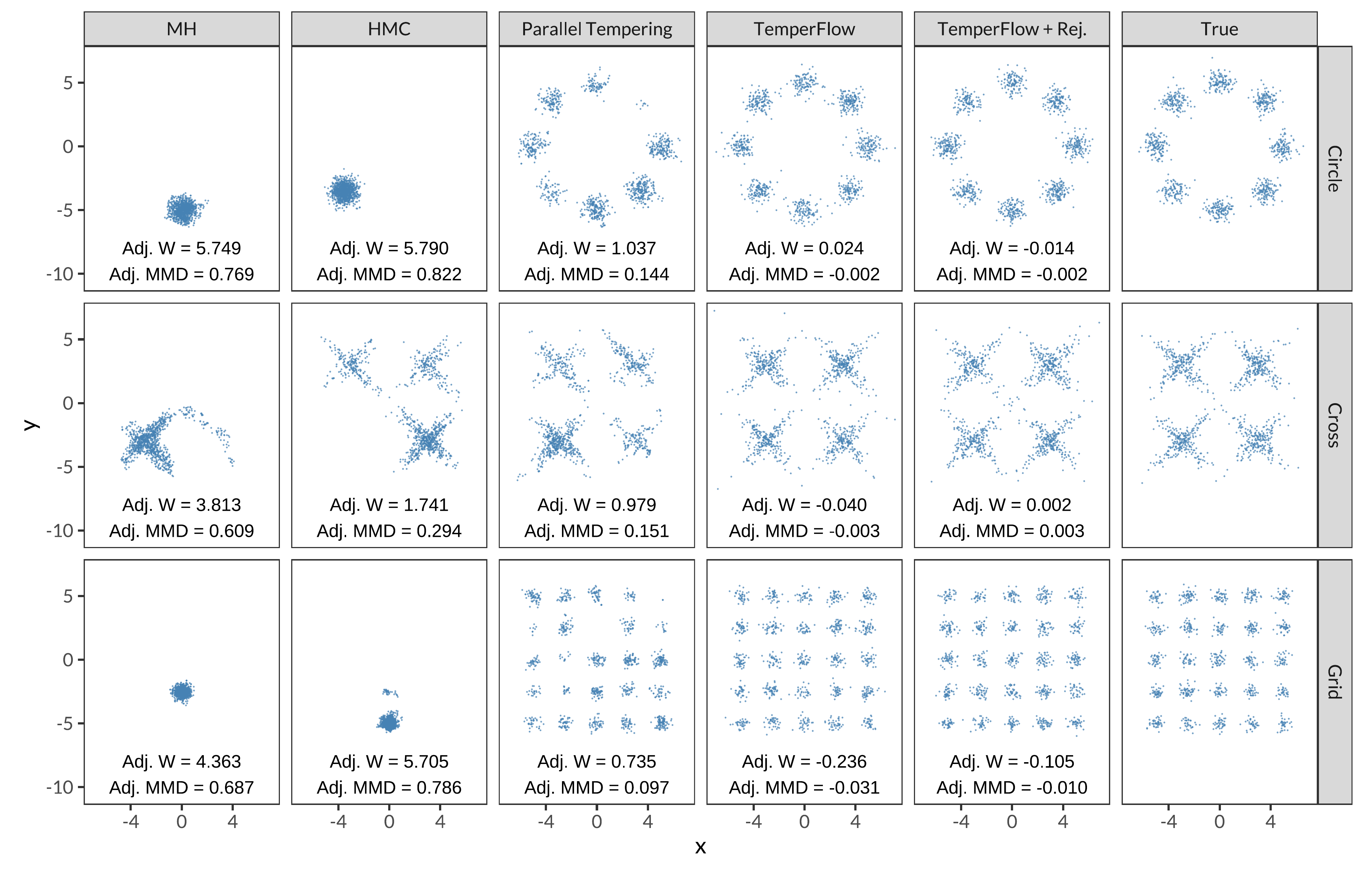}
\par\end{centering}
\caption{\label{fig:experiments_2d}A comparison of different sampling methods
on three Gaussian mixture models. MH stands for the Metropolis--Hastings
MCMC method, HMC is Hamiltonian Monte Carlo, and \textquotedblleft TemperFlow
+ Rej.\textquotedblright{} means post-processing the TemperFlow samples
by rejection sampling. The numbers under each panel are the adjusted
Wasserstein distance and adjusted MMD between the generated points
and the true sample, respectively.}
\end{figure}

It is clear that for basic MCMC methods such as Metropolis--Hastings
and Hamiltonian Monte Carlo, the sampling results severely lose most
of the modes. Parallel tempering greatly improves the quality of the
samples, but it also has difficulty in achieving the correct proportion
for each mode. In contrast, TemperFlow successfully captures all modes,
and the sample quality can be further improved by a rejection sampling
refinement. To account for randomness in sampling, we repeat the experiment
100 times for each method and each target distribution, and summarize
their sampling errors in Figure \ref{fig:experiments_2d_errors}.
The results are consistent with those in Figure \ref{fig:experiments_2d}:
among all the methods, TemperFlow and its refined version give uniformly
smaller errors than other methods.

\begin{figure}[h]
\begin{centering}
\includegraphics[width=0.9\textwidth]{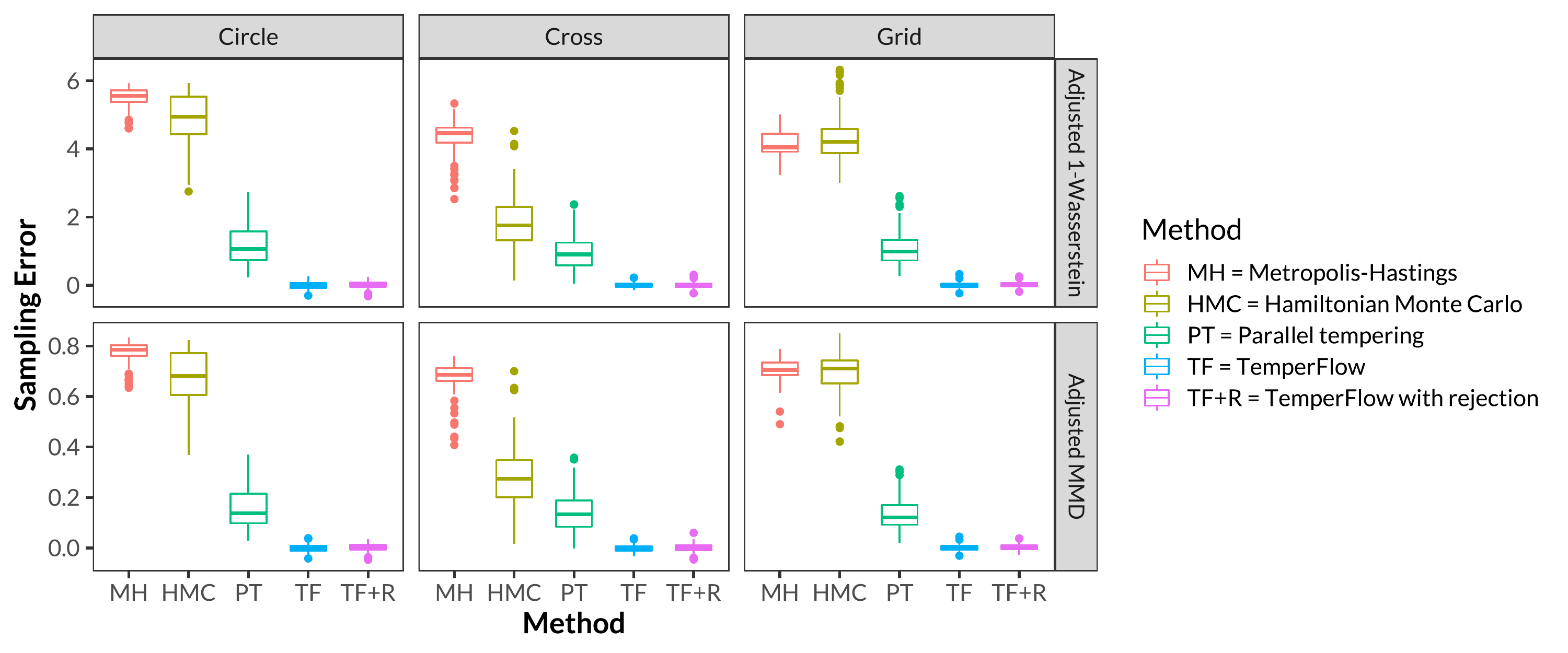}
\par\end{centering}
\caption{\label{fig:experiments_2d_errors}A comparison of sampling errors
of different methods on three normal mixture distributions. The boxplots
are drawn based on 100 simulation runs.}
\end{figure}

\subsection{Copula-generated distributions}

\label{subsec:copula}

To study more general multimodal distributions, we use copulas to
define multivariate densities that have arbitrary dimensions. See
\citet{nelsen2006introduction} for an introduction to copula modeling.
Specifically, we define the target distribution function as $F(x_{1},\ldots,x_{d})=C(F_{1}(x_{1}),\ldots,F_{d}(x_{d}))$,
where $F_{i}(x)$ is the marginal distribution function of each component
$X_{i}$, and $C(u_{1},\ldots,u_{d})$ is the copula function. In
our experiment, we take the first $s$ marginals to be a mixture of
normal distributions, $F_{i}\sim0.7\cdot N(-1,0.2^{2})+0.3\cdot N(1,0.2^{2})$,
$i=1,\ldots,s$, and the remaining to be a normal distribution $N(0,0.25)$.
The function $C(u_{1},\ldots,u_{d})=(u_{1}^{-\theta}+\cdots+u_{d}^{-\theta}-d+1)^{-1/\theta}$
is a Clayton copula with $\theta=2$. Figure \ref{fig:distribution_clayton}
shows the scatterplot and density plot of $(X_{i},X_{j})$, $i,j\le s$,
$i\neq j$.

\begin{figure}[h]
\begin{centering}
\includegraphics[width=0.27\textwidth]{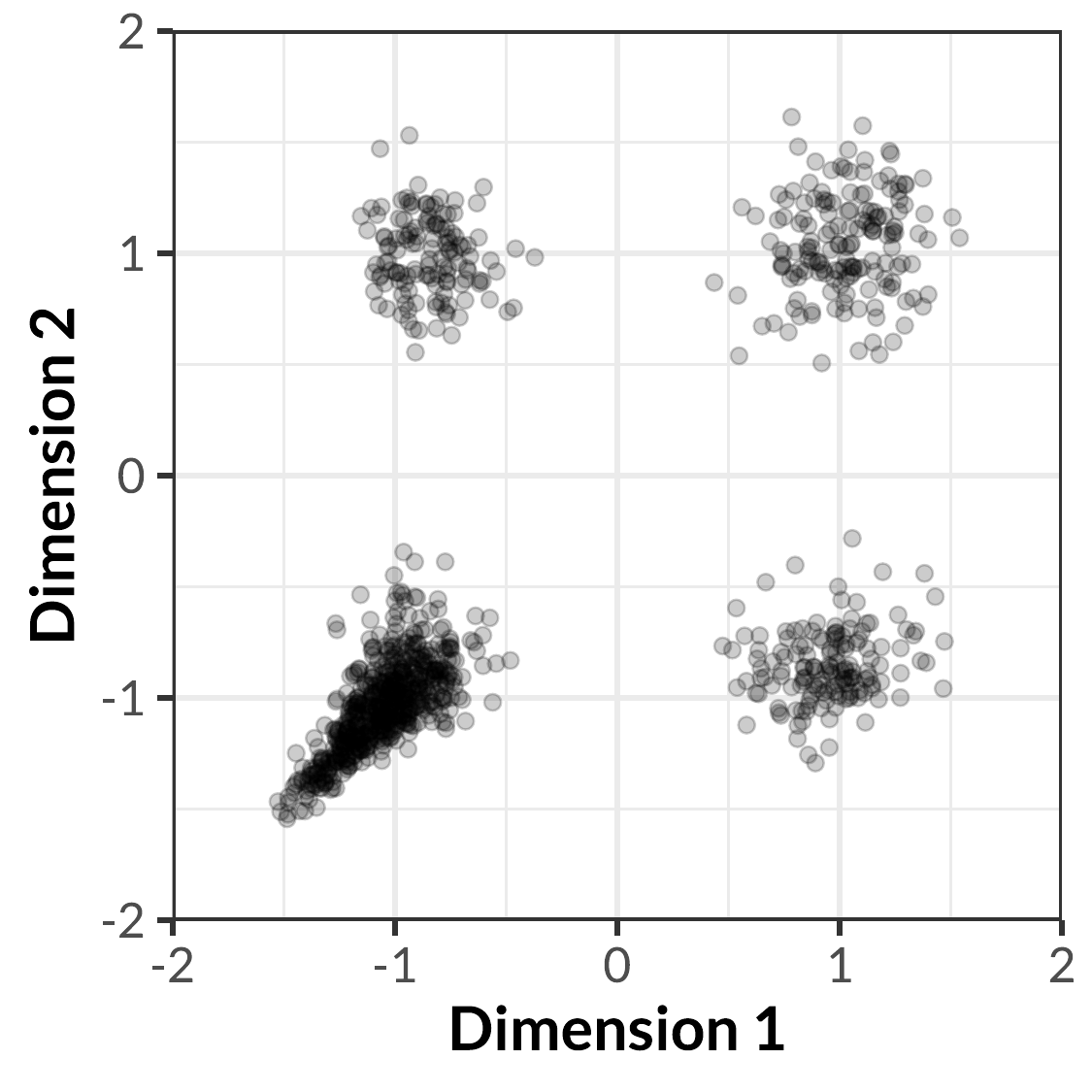}\includegraphics[width=0.36\textwidth]{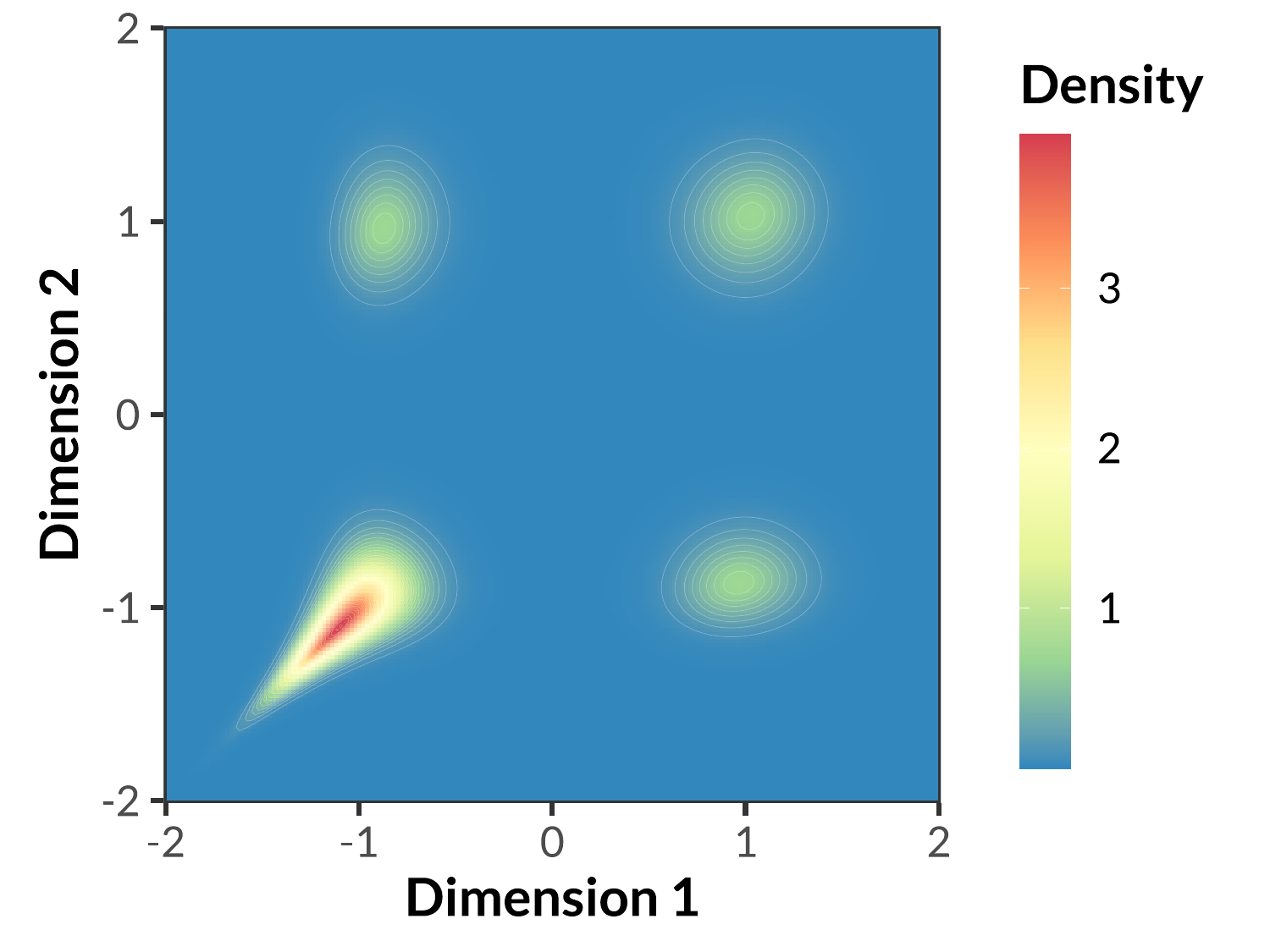}
\par\end{centering}
\caption{\label{fig:distribution_clayton}Normal mixture marginals combined
with a Clayton copula.}
\end{figure}

\begin{figure}[h]
\begin{centering}
\includegraphics[width=0.9\textwidth]{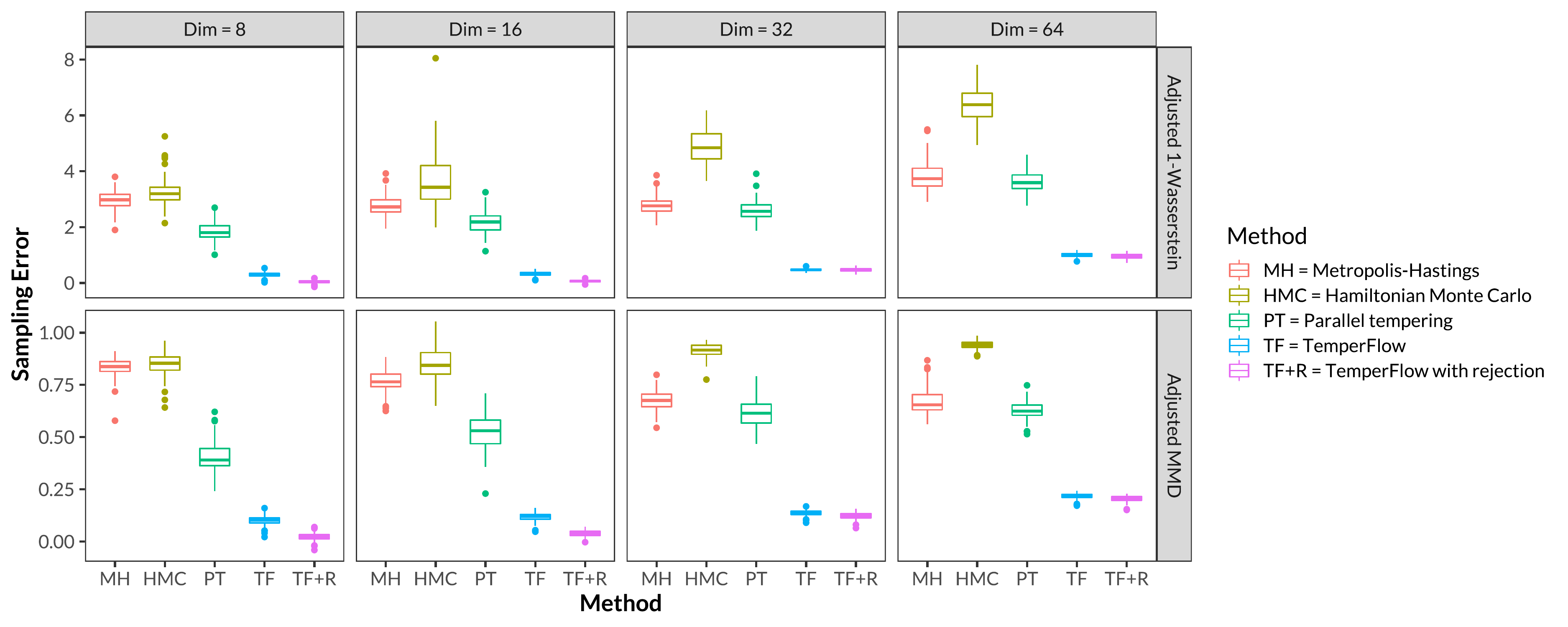}
\par\end{centering}
\caption{\label{fig:experiments_copula_error}Sampling errors of different
methods for the copula-generated distribution.}
\end{figure}

Under such a design, the target distribution $F(x_{1},\ldots,x_{d})$
has $2^{s}$ modes, which grows exponentially with $s$, and every
pair $(X_{i},X_{j})$ is correlated. Therefore, it is a very challenging
distribution to sample from even with moderate $s$ and $d$. We first
derive the density function $f(x_{1},\ldots,x_{d})$ by computing
the partial derivatives of $F$, and then use different methods to
sample from $f$ with $s=8$ and $d=8,16,32,64$. That is, each distribution
has $256$ modes in total, with varying dimensions.

Similar to the bivariate experiments, we compute the adjusted Wasserstein
distances and adjusted MMD between the true and generated samples,
and summarize the comparison results in Figure \ref{fig:experiments_copula_error}
based on 100 repetitions. Figure \ref{fig:experiments_copula} in
the supplementary material also shows the pairwise scatterplots and
density contour plots of the generated samples by parallel tempering
and TemperFlow, respectively. As expected, TemperFlow samples have
desirable quality close to the ground truth, whereas other methods
encounter great issues given the huge number of modes and high dimensions.

Of course, both TemperFlow and MCMC are iterative algorithms, so their
sampling errors would be affected by the number of iterations, which
further impact the overall computing time. We report more comparison
results in Section \ref{subsec:mcmc_more_iters} of the supplementary
material, and also show the computing time of different algorithms
in Table \ref{tab:benchmark}. We remark that MCMC only has generation
costs, whereas TemperFlow also has a training cost. The main difference
is that the generation cost of MCMC is proportional to the number
of Markov chain iterations, but for TemperFlow it is nearly fixed
and extremely small. Instead, the training cost of TemperFlow scales
linearly with the number of optimization iterations.

\section{Application: Deep Generative Models}

\label{sec:application}

Modern deep learning techniques have attracted enormous attentions
from statistical researchers and practitioners, among which deep generative
models are a class of important unsupervised learning methods \citep{salakhutdinov2015learning,bond2021deep}.
Deep generative models attempt to model the statistical distribution
of high-dimensional data using DNNs, with wide applications in image
synthesis, text generation, etc.

One general class of deep generative models has the form $X=G(Z)$,
where $X\in\mathbb{R}^{p}$ is the high-dimensional data point, for
example, an image, $Z\in\mathbb{R}^{d}$ is a latent random vector
with $d\ll p$, and $G:\mathbb{R}^{d}\rightarrow\mathbb{R}^{p}$ is
a DNN, typically called the generator. The distribution of $Z$ is
characterized by an energy function $E:\mathbb{R}^{d}\rightarrow\mathbb{R}$,
which is also a DNN. In this sense, the pair $(E_{d},G_{d,p})$ defines
a deep generative model, where the subscripts $d$ and $p$ indicate
the dimensions. There is a great deal of literature discussing the
modeling and estimation of $(E_{d},G_{d,p})$; see \citet{che2020your,pang2020learning}
for more details. In this section, we treat the functions $E$ and
$G$ as known, and focus on the sampling of $p(z)\propto\exp\{-E(z)\}$,
as it is the key to generating new data points of $X$. 

We first consider generative models for the Fashion-MNIST data set
\citep{xiao2017fashion}. The Fashion-MNIST data contain a training
set of 60,000 images and a testing set of 10,000 images, each consisting
of $28\times28$ grey-scale pixels. The plot on the left of Figure
\ref{fig:fmnist_2d} shows 60 images from the data set. We then build
a deep generative model $(E_{2},G_{2,784})$ with $d=2$ based on
existing literature \citep{che2020your}, and the generated images
are shown in the right plot of Figure \ref{fig:fmnist_2d}.

\begin{figure}[h]
\begin{centering}
\includegraphics[width=0.49\textwidth]{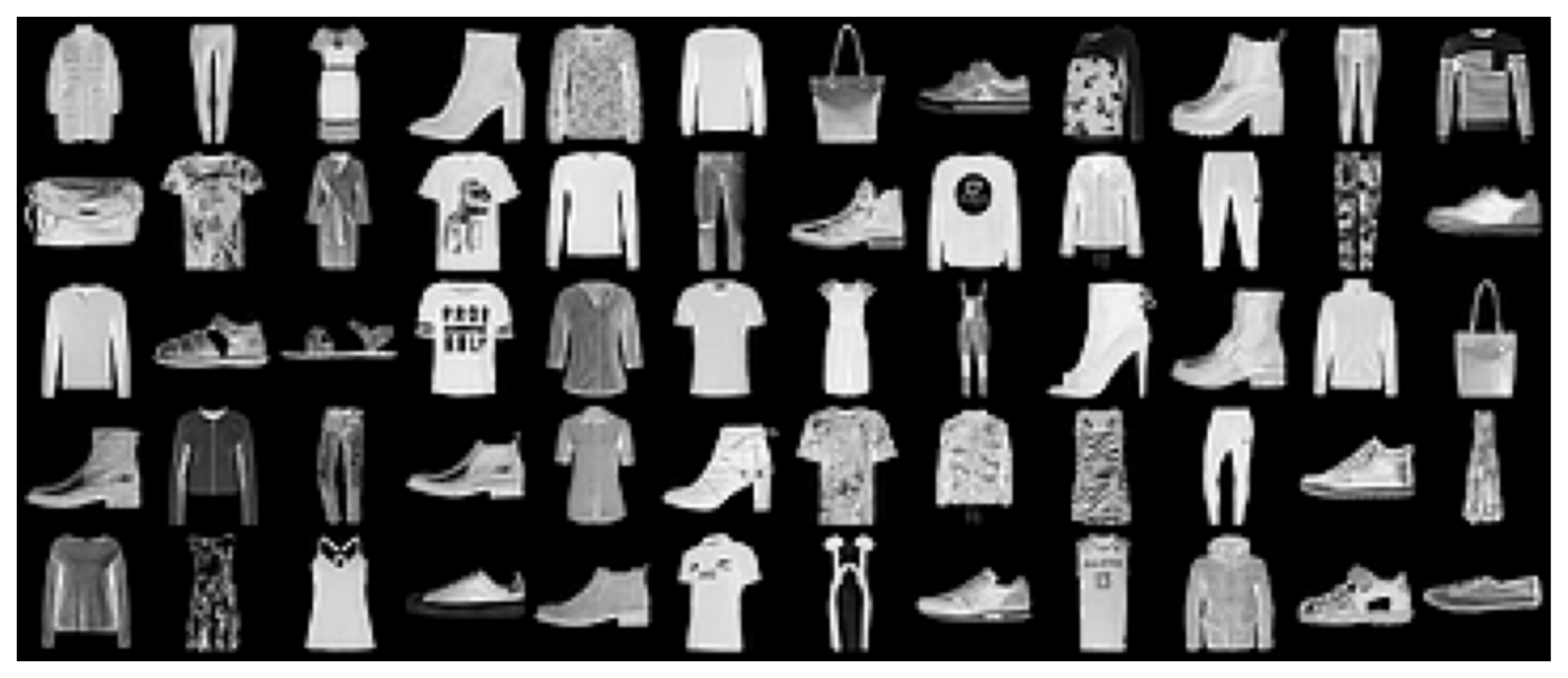} \includegraphics[width=0.49\textwidth]{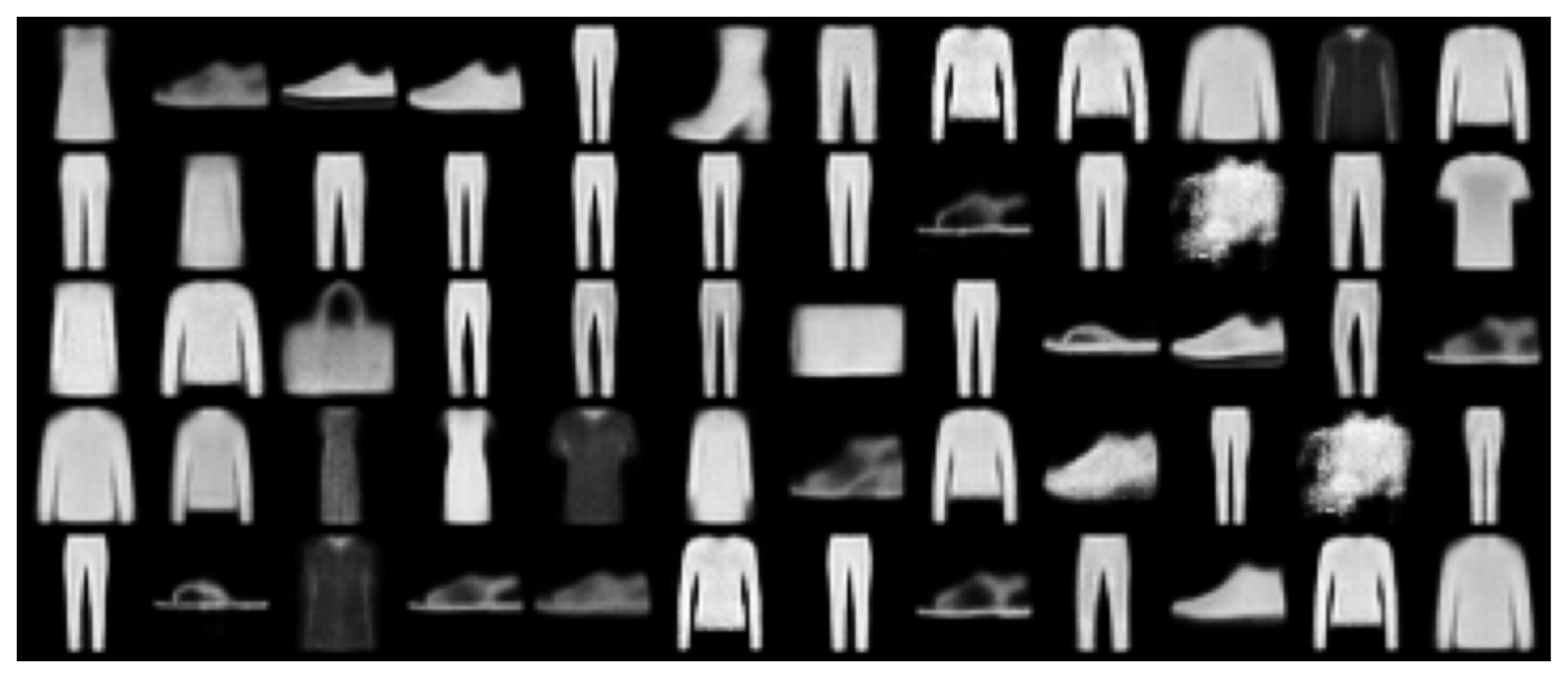}
\par\end{centering}
\caption{\label{fig:fmnist_2d}Left: true images from the Fashion-MNIST data
set. Right: images generated by a deep generative model with latent
dimension $d=2$.}
\end{figure}

\begin{figure}[h]
\begin{centering}
\includegraphics[width=0.85\textwidth]{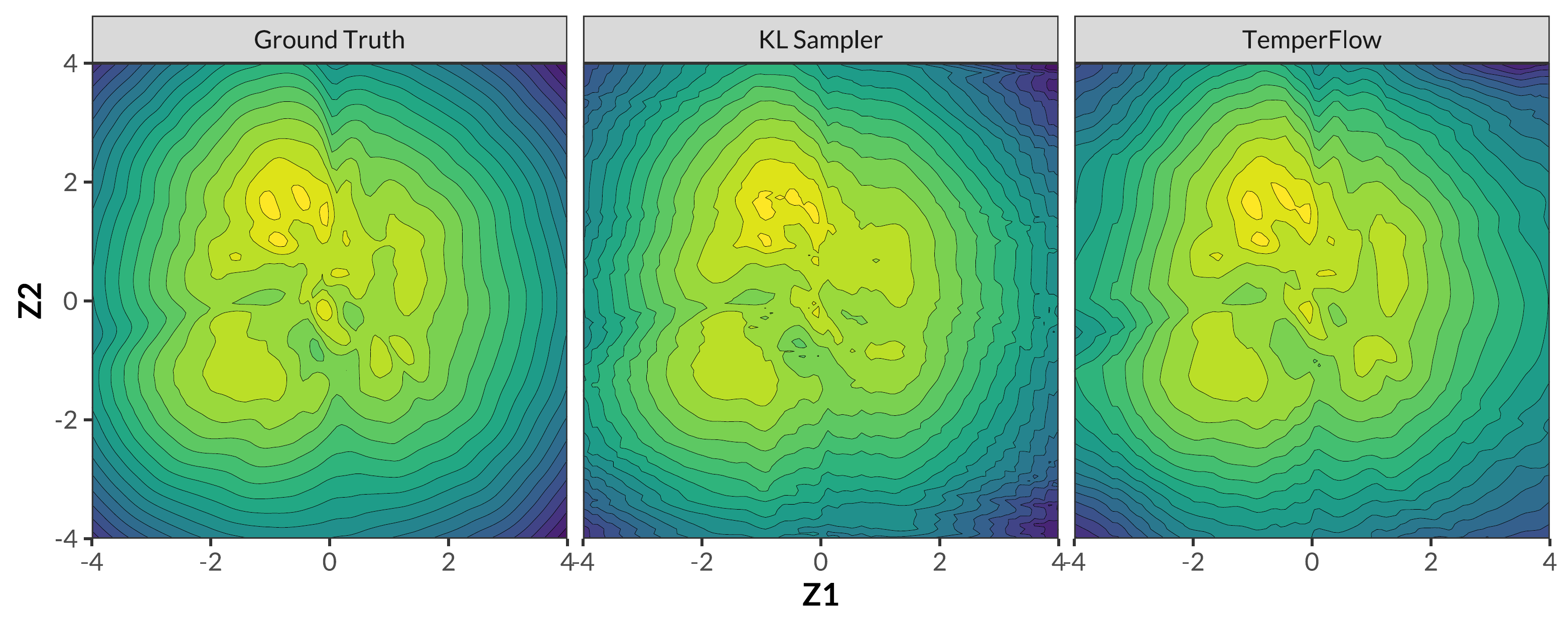}
\par\end{centering}
\caption{\label{fig:fmnist_2d_logp}Contour plots of the true log-density function
(left), the function learned by the KL sampler (middle), and the function
learned by TemperFlow (right).}
\end{figure}

\begin{figure}
\begin{centering}
\includegraphics[width=0.3\textwidth]{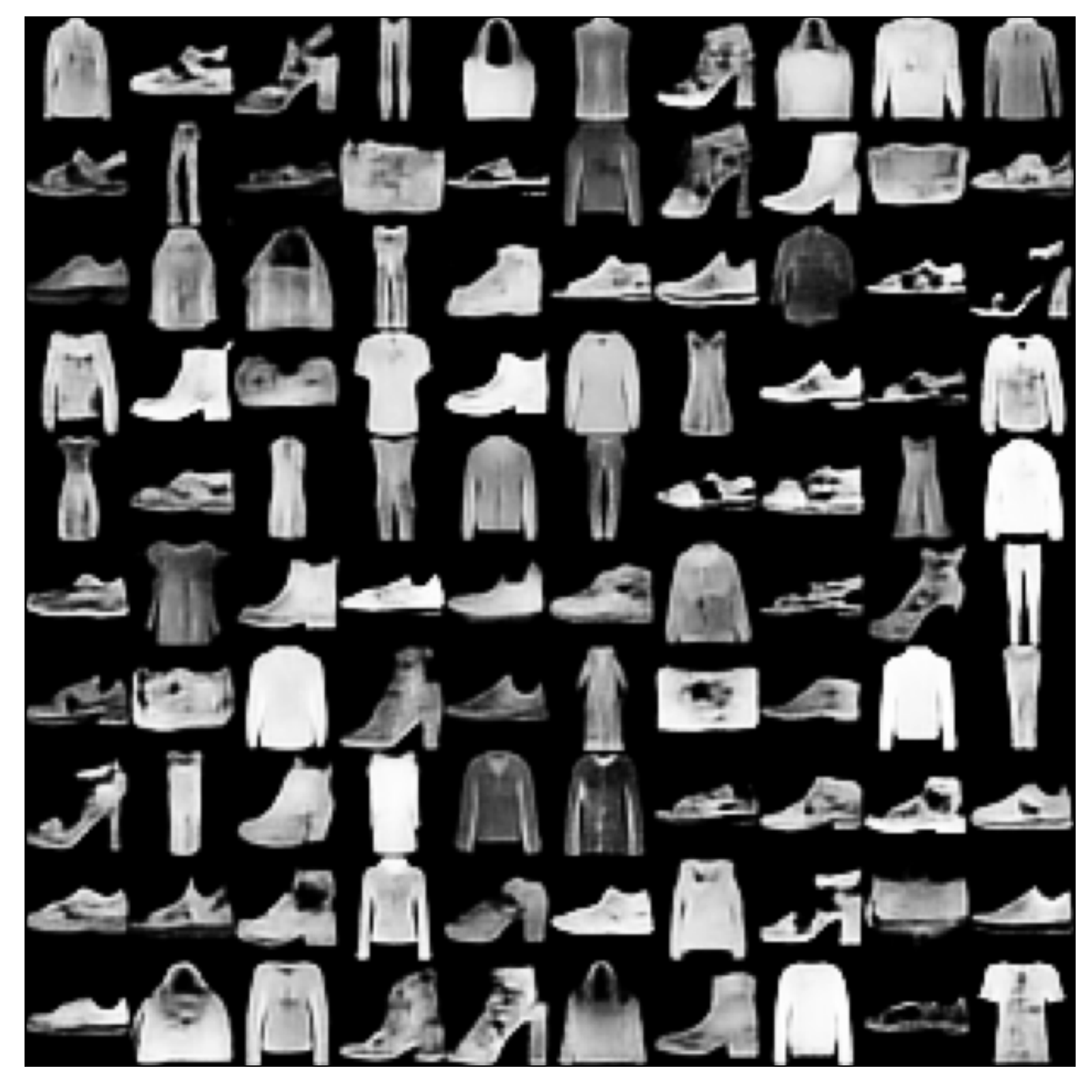} \includegraphics[width=0.3\textwidth]{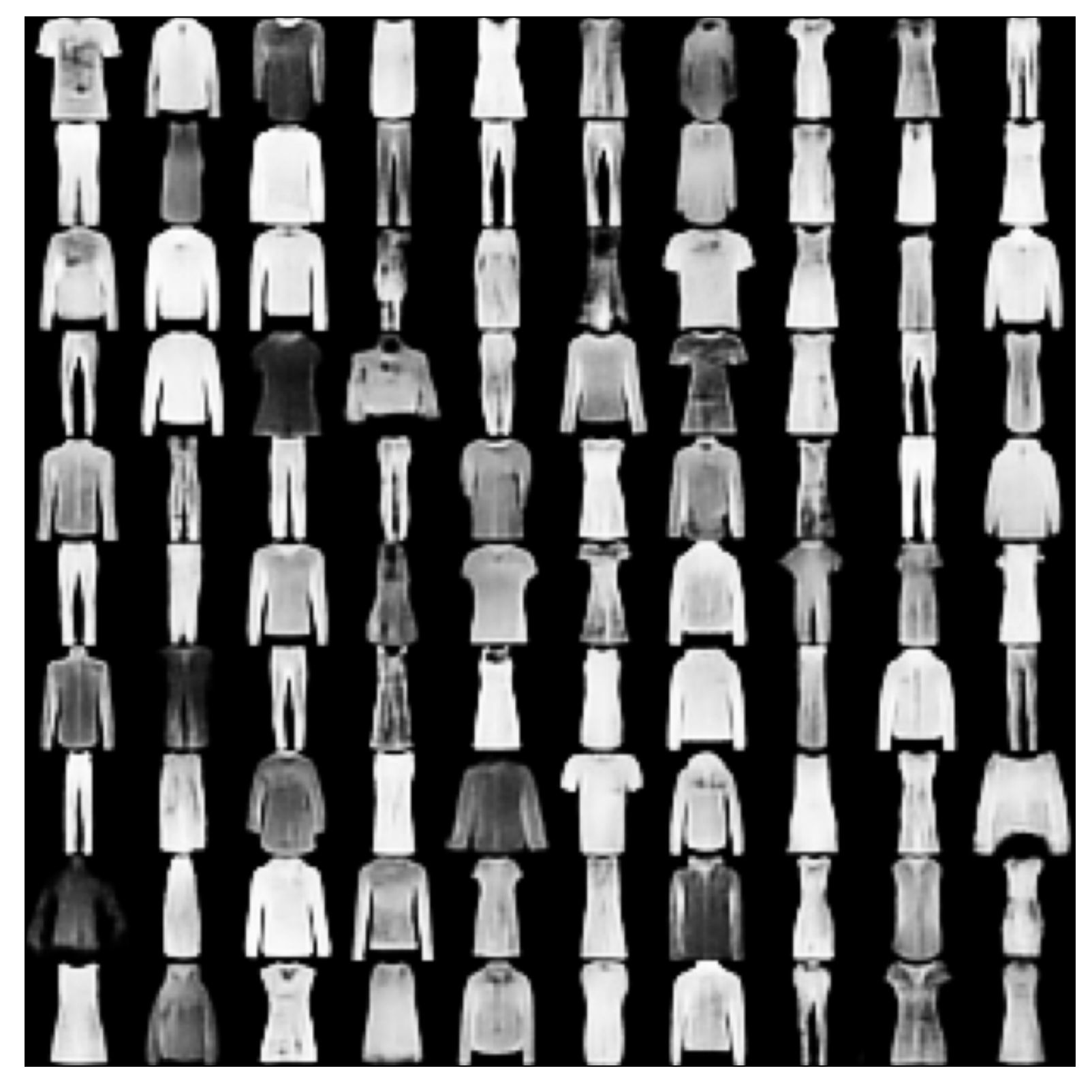}
\includegraphics[width=0.3\textwidth]{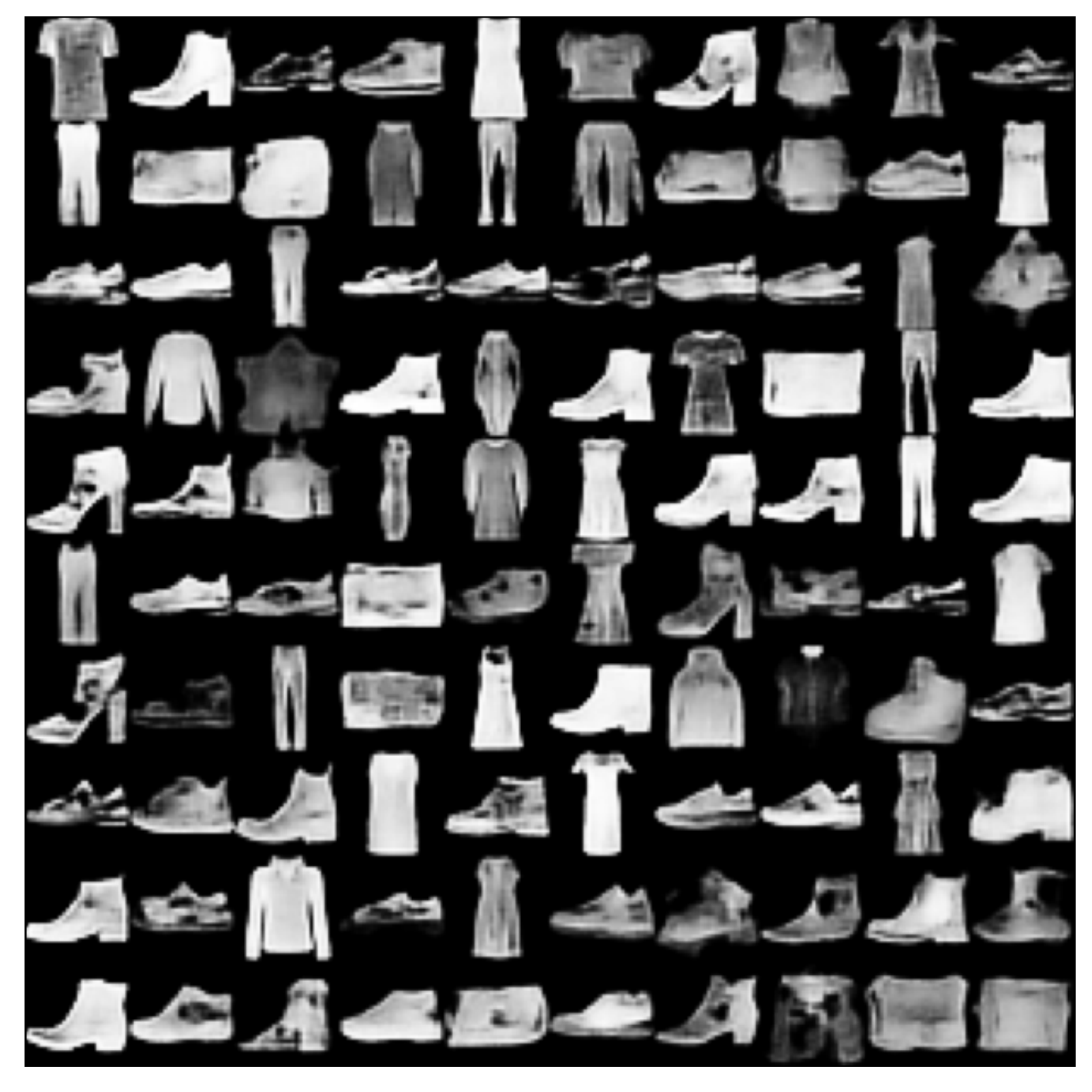}
\par\end{centering}
\caption{\label{fig:fmnist_rd}Left: images generated by a deep generative
model with latent dimension $d=32$. Middle: images sampled from the
model using the KL sampler. Right: images sampled using the TemperFlow
sampler.}
\end{figure}

Since generating images requires independent samples, the measure
transport sampling method is especially useful for this task. To this
end, we compare the existing KL sampler and our proposed TemperFlow
in sampling from $p(z)$. After both samplers are trained, we visualize
their estimated log-density functions $\log p(z)$ via contour plots
in Figure \ref{fig:fmnist_2d_logp}, and also compare them with the
ground truth.

It is clear that the true log-density function $\log p(z)$ contains
many isolated modes, but the KL sampler only captures a few of them.
In contrast, the distribution given by TemperFlow is very close to
the ground truth. The difference becomes more evident when the latent
dimension $d$ increases. We fit another generative model $(E_{32},G_{32,784})$
with a larger latent dimension $d=32$, and show the generated images
in the left plot of Figure \ref{fig:fmnist_rd}. Similar to the $d=2$
case, we train both the KL sampler and the TemperFlow sampler from
the energy function $E(z)$. In the higher-dimensional case, we are
no longer able to visualize the density function directly. Instead,
we simulate latent variables $Z_{1},\ldots,Z_{100}\in\mathbb{R}^{32}$
from both samplers, and pass them to the generator $G$ to form images
$X_{1}=G(Z_{1}),\ldots,X_{100}=G(Z_{100})\in\mathbb{R}^{784}$. The
middle and right plots of Figure \ref{fig:fmnist_rd} show the generated
images by KL sampler and TemperFlow, respectively. It is obvious that
the KL sampler only generates coat-like and trousers-like images,
indicating that it loses a lot of modes of the true distribution.
On the other hand, TemperFlow successfully preserves major modes of
$p(z)$, further validating the superior performance of the proposed
sampler.

Finally, we illustrate a much larger generative model trained from
the CelebA data set \citep{liu2015deep}. Each data point in CelebA
is a color human face image of $64\times64$ size, and we take a subset
of the original data set to form the following four categories: females
wearing glasses (2677 images), females without glasses (20000 images),
males wearing glasses (10346 images), and males without glasses (20000
images). A deep generative model $(E_{100},G_{100,12288})$ is trained
on this subset, and Figure \ref{fig:face}(a) shows images generated
by the model. Naturally, the latent energy function $E$ would reflect
these four major modes, and we test the performance of the KL sampler,
the MCMC samplers, and TemperFlow by looking at the images generated
by them.

\begin{figure}[h]
\begin{centering}
\subfloat[True model sample]{\begin{centering}
\includegraphics[width=0.33\textwidth]{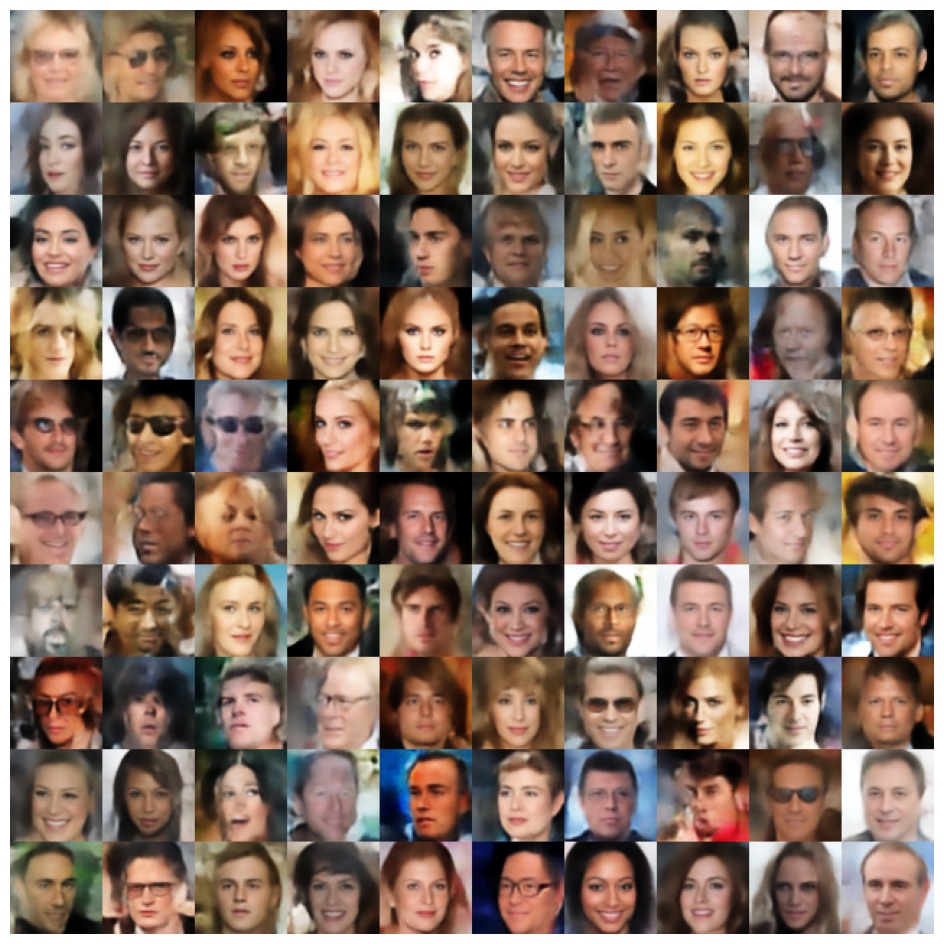}
\par\end{centering}
} \subfloat[KL sampler]{\begin{centering}
\includegraphics[width=0.33\textwidth]{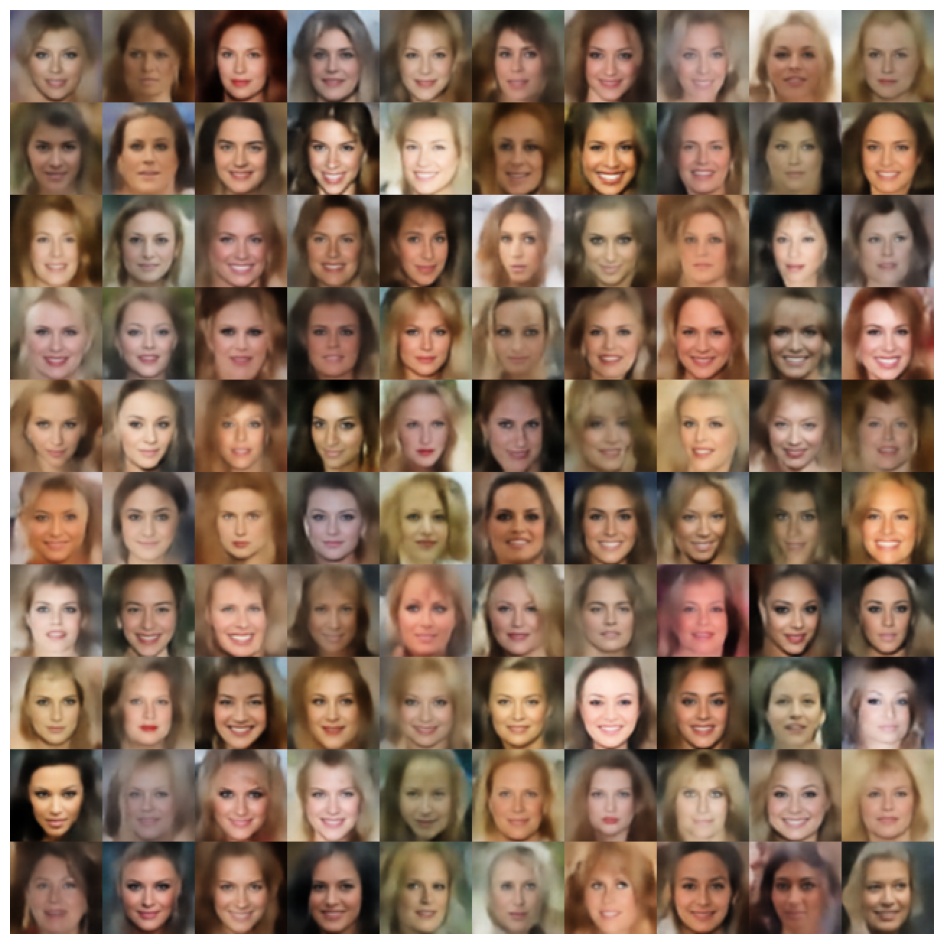}
\par\end{centering}
} \subfloat[TemperFlow]{\begin{centering}
\includegraphics[width=0.33\textwidth]{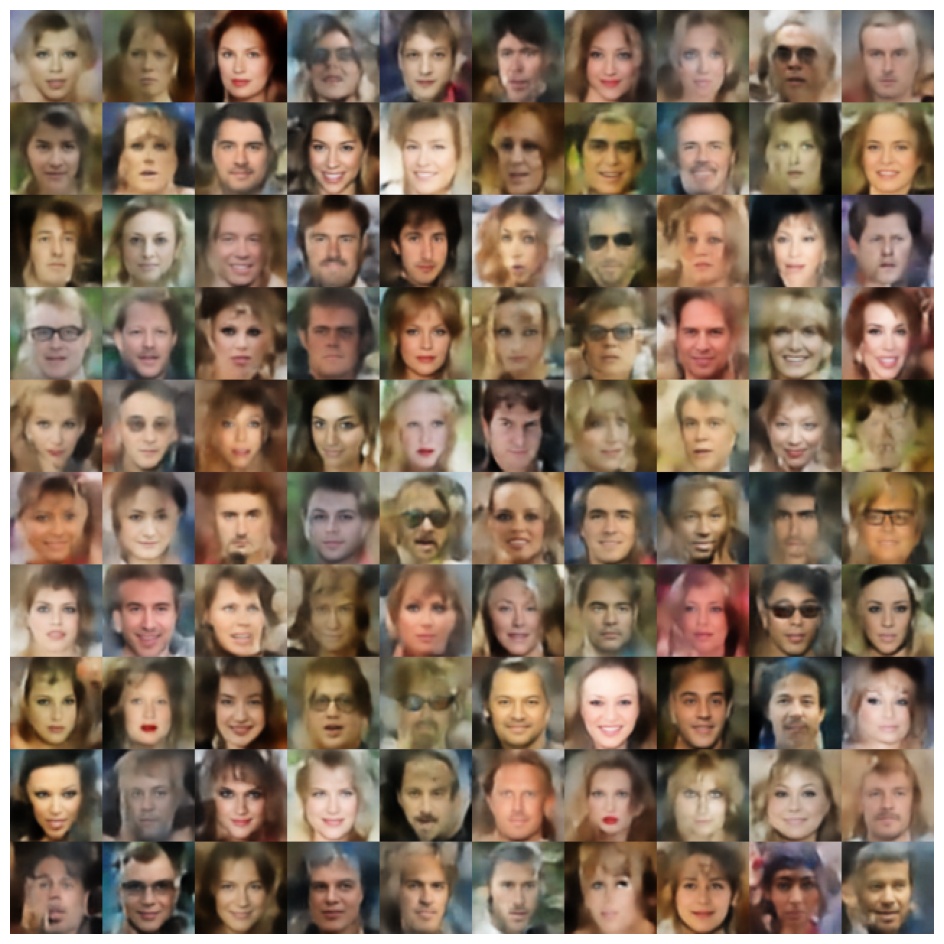}
\par\end{centering}
}
\par\end{centering}
\begin{centering}
\subfloat[Metropolis--Hastings]{\begin{centering}
\includegraphics[width=0.33\textwidth]{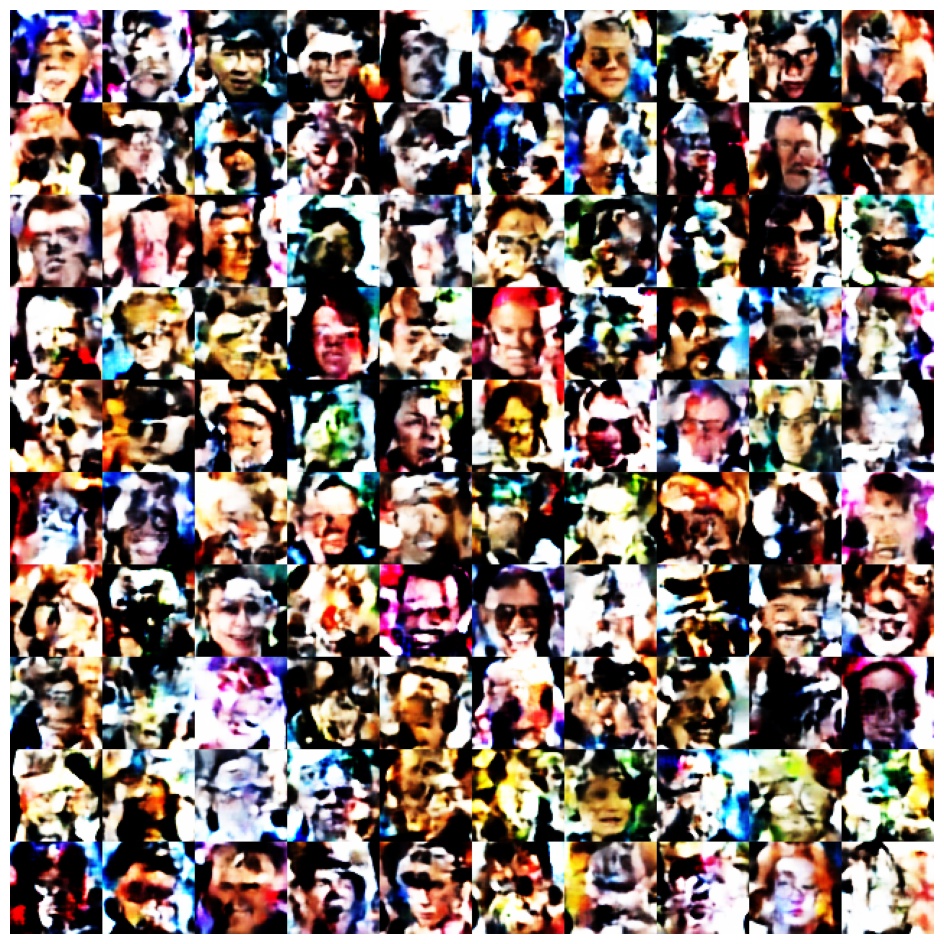}
\par\end{centering}
} \subfloat[Hamiltonian Monte Carlo]{\begin{centering}
\includegraphics[width=0.33\textwidth]{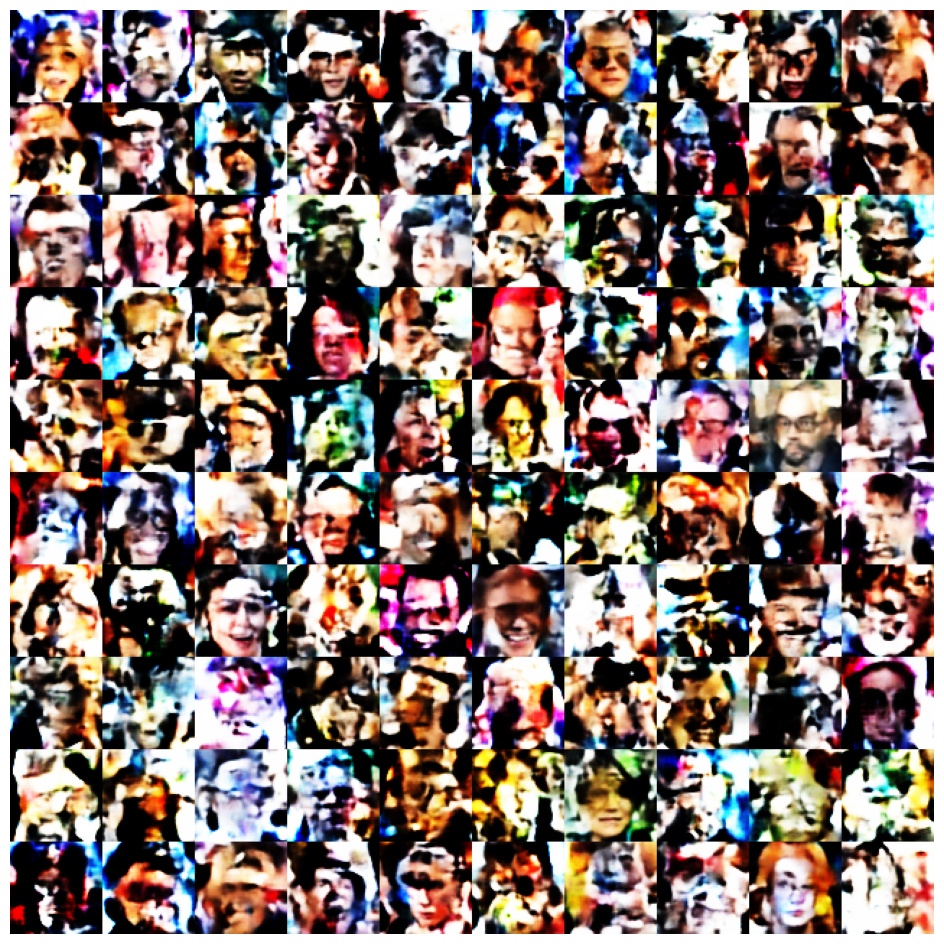}
\par\end{centering}
} \subfloat[Parallel tempering]{\begin{centering}
\includegraphics[width=0.33\textwidth]{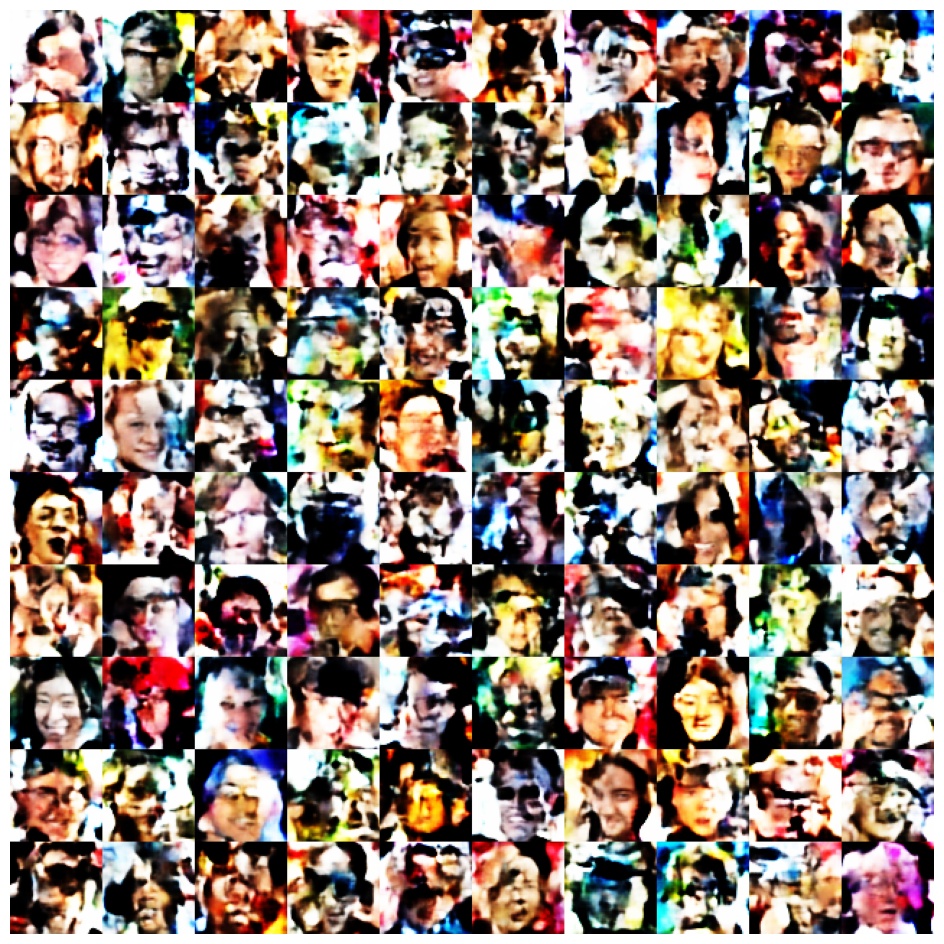}
\par\end{centering}
}
\par\end{centering}
\caption{\label{fig:face}(a) Images generated by a deep generative model with
latent dimension $d=100$, trained from the CelebA data set. (b)-(f):
Images sampled from the model using different samplers.}
\end{figure}

Figures \ref{fig:face}(b)-(f) show the result. Clearly, the KL sampler
almost only outputs one class, females without glasses, whereas TemperFlow
captures all four classes. On the other hand, MCMC samples hardly
generate visible human faces. In summary, all of the examples presented
above demonstrate the accuracy and scalability of the TemperFlow sampler.

\section{Conclusion and Discussion}

\label{sec:conclusion}

In this article, we develop a general-purpose sampler, TemperFlow,
to sample from challenging statistical distributions. TemperFlow is
inspired by the measure transport framework introduced in \citet{marzouk2016sampling},
but overcomes its critical weaknesses when the target distribution
is multimodal. Their fundamental differences are rigorously analyzed
using the Wasserstein gradient flow theory, and are also clearly illustrated
by various numerical experiments.

One interesting characteristic of TemperFlow is that it transforms
the sampling problem into a series of optimization problems, which
has the following implications. First, many modern optimization techniques
can be exploited directly to accelerate the training of transport
maps. Second, it is typically easier to diagnose the convergence of
an optimization problem than that of a Markov chain. Third, as we
primarily rely on gradient-based optimization methods, the computation
of TemperFlow is highly parallelizable, and can greatly benefit from
modern computing hardware such as GPUs. All these aspects reflect
the huge advantages of TemperFlow in computational efficiency.

As a general-purpose sampling method, TemperFlow can be compared with
MCMC on many aspects, but one of the most visible differences is that
TemperFlow has a training stage, whereas MCMC directly generates random
variates. This difference serves as a guide on which method to use
when sampling is needed. Specifically, TemperFlow is especially useful
when a large number of independent random variates are requested,
and MCMC, on the other hand, may be more suitable in computing expectations
for distributions that are continually changing, for example in the
Monte Carlo EM algorithm \citep{wei1990monte,levine2001implementations}.
To this end, we do not position TemperFlow as a substitute for MCMC.
Instead, it is viewed as a complement to MCMC and other particle-based
methods. We anticipate that these methods can be combined to develop
new efficient samplers, which we leave as a future research direction.

Another promising direction is evaluating the goodness-of-fit of models
over observed data. Let $X_{1},\ldots,X_{n}\stackrel{iid}{\sim}q(x)$
be a random sample from the TemperFlow sampler. Recall that the target
density is $p(x)\propto e^{-E(x)}$. We can perform a one-sample goodness-of-fit
test $H_{0}:q=p$. The main difficulties are due to complex and high-dimensional
models and the unknown normalizing constant for $p(x)$. Traditional
goodness-of-fit tests such as the $\chi^{2}$-test and the Kolmogorov--Smirnov
test can hardly be applied. The nonparametric one-sample test based
on the Stein discrepancy was studied in \citet{chwialkowski2016kernel}.
Note that the goodness-of-fit test for the TemperFlow sampler has
some special features that are different from the previous problems.
In particular, the density $q$ has an explicit formula given the
normalizing flow model. This feature allows us to propose a more powerful
test, such as the test based on Fisher divergence. The Fisher divergence
is stronger than many other divergences, such as total variation,
Hellinger distance, and Wasserstein distance \citep{ley2013stein}.\newpage{}

\appendix
\setcounter{figure}{0}
\setcounter{table}{0}
\renewcommand\thefigure{S\arabic{figure}}
\renewcommand\thetable{S\arabic{table}}

\section{Appendix}

\subsection{Invertible Neural Networks}

\label{sec:inn}

Invertible neural networks (INNs) are a class of neural network models
that are invertible with respect to its inputs. An INN can be viewed
as a mapping $T:\mathbb{R}^{d}\rightarrow\mathbb{R}^{d}$ such that
$T^{-1}$ exists and can be efficiently evaluated, and $T$ is typically
a composition of simpler mappings, $T=T_{K}\circ T_{K-1}\circ\cdots\circ T_{1}$,
where each $T_{i}$ is invertible. INNs are mainly used to implement
normalizing flow models \citep{tabak2010density,tabak2013family,rezende2015variational},
which can be described as transformations of a probability density
through a sequence of invertible mappings.

Normalizing flows were originally developed as nonparametric density
estimators \citep{tabak2010density,tabak2013family}, and the mappings
used there were simple functions with limited expressive powers. After
normalizing flows were introduced to the deep learning community,
many powerful INN-based models were developed, including affine coupling
flows (\citealp{dinh2014nice}; \citealp{dinh2016density}), masked
autoregressive flows \citep{papamakarios2017masked}, inverse autoregressive
flows \citep{kingma2016improving}, neural spline flows \citep{durkan2019neural},
and linear rational spline flows \citep{dolatabadi2020invertible},
among many others. It is worth mentioning that the term ``flow''
in normalizing flows has a conceptual gap with that in gradient flows,
where the latter is the focus of this article. Therefore, to avoid
ambiguity, in this article we use invertible neural networks to refer
to normalizing flow models, although other forms of normalizing flows
also exist, such as the continuous normalizing flows based on ordinary
differential equations \citep{chen2018neural}, and its extension
using free-form continuous dynamics \citep{grathwohl2019scalable}.

Take the inverse autoregressive flow as an example. It is the composition
of a sequence of invertible mappings, $T=T_{K}\circ T_{K-1}\circ\cdots\circ T_{1}$.
For each $i=1,\ldots,K$, let $x=(x_{1},\ldots,x_{d})'$ be the input
vector, and denote by $y=(y_{1},\ldots,y_{d})'=T_{i}(x)$ the output
vector. Then $y$ has the following form:
\begin{align*}
y_{1} & =\mu_{i1}+\sigma_{i1}x_{1},\\
y_{j} & =\mu_{ij}(x_{1:j-1})+\sigma_{ij}(x_{1:j-1})x_{j},\quad j=2,\ldots,d,
\end{align*}
where $x_{1:r}$ means $(x_{1},\ldots,x_{r})'$, $\mu_{i1}$, $\sigma_{i1}$
are scalars, and $\mu_{ij},\sigma_{ij}:\mathbb{R}^{j-1}\rightarrow\mathbb{R}$
are neural networks for $j\ge2$. It can be easily verified that $T_{i}$
is invertible, since
\begin{align*}
x_{1} & =(y_{1}-\mu_{i1})/\sigma_{i1},\\
x_{2} & =(y_{2}-\mu_{i2}(x_{1}))/\sigma_{i2}(x_{1})=(y_{2}-\mu_{i2}((y_{1}-\mu_{i1})/\sigma_{i1}))/\sigma_{i2}((y_{1}-\mu_{i1})/\sigma_{i1}),\\
x_{j} & =(y_{j}-\mu_{ij}(x_{1:j-1}))/\sigma_{ij}(x_{1:j-1}),\quad j=3,\ldots,d,
\end{align*}
where each $x_{j}$ is a function of $y_{j}$ and $x_{1:j-1}$, and
$x_{1:j-1}$ can be recursively reduced to functions of $y_{1:j-1}$.
Furthermore, if the neural networks $\mu_{ij}(\cdot)$ and $\sigma_{ij}(\cdot)$
are differentiable, which can be easily achieved by using smooth activation
functions, then each $T_{i}$ and the whole $T$ mapping are also
differentiable. In this sense, INNs are diffeomorphisms by design
under very mild conditions.

In addition, most INNs have the desirable property that the Jacobian
matrix $\nabla T$ is a triangular matrix, so its determinant $\det\nabla T$
is simply the product of diagonal elements. Again using the example
above, we can find that
\[
\nabla T=\left(\begin{array}{cccc}
\frac{\partial y_{1}}{\partial x_{1}} & \frac{\partial y_{1}}{\partial x_{2}} & \cdots & \frac{\partial y_{1}}{\partial x_{d}}\\
\frac{\partial y_{2}}{\partial x_{1}} & \frac{\partial y_{2}}{\partial x_{2}} & \cdots & \frac{\partial y_{2}}{\partial x_{d}}\\
\vdots & \vdots & \ddots & \vdots\\
\frac{\partial y_{d}}{\partial x_{1}} & \frac{\partial y_{d}}{\partial x_{2}} & \cdots & \frac{\partial y_{d}}{\partial x_{d}}
\end{array}\right)=\left(\begin{array}{cccc}
\sigma_{i1} & 0 & \cdots & 0\\
* & \sigma_{i2}(x_{1}) & \cdots & 0\\
\vdots & \vdots & \ddots & \vdots\\
* & * & \cdots & \sigma_{id}(x_{1:d-1})
\end{array}\right),
\]
so $\det\nabla T=\sigma_{i1}\cdot\sigma_{i2}(x_{1})\cdots\sigma_{id}(x_{1:d-1})$,
and $\det\nabla T>0$ automatically holds if $\sigma_{ij}>0$.

In practice, a permutation operator $P_{\pi}(x_{1},\ldots,x_{d})=(x_{\pi(1)},\ldots,x_{\pi(d)})'$
is inserted between each pair of $T_{k}$ and $T_{k-1}$, where $(\pi(1),\ldots,\pi(d))$
is a permutation of $(1,\ldots,d)$. This is because $T_{ij}$ is
an affine mapping of $x_{j}$ conditional on $x_{1:j-1}$, and the
expressive power of $T$ would be limited if the variables do not
change the order.

\subsection{Details of Numerical Experiments}

\subsubsection{Definition of metrics}

\label{subsec:def_metrics}

Given data points $X=(X_{1},\ldots,X_{n})$ and $Y=(Y_{1},\ldots,Y_{m})$,
define the discrete 1-Wasserstein distance between $X$ and $Y$ as
\[
W(X,Y)=\min_{P\in\Pi}\ \langle P,C\rangle,\quad\Pi=\{P\in\mathbb{R}_{+}^{n\times m}:P\mathbf{1}_{m}=n^{-1}\mathbf{1}_{n},P'\mathbf{1}_{n}=m^{-1}\mathbf{1}_{m}\},
\]
where $C=(C_{ij})$ is an $n\times m$ matrix with $C_{ij}=\Vert X_{i}-Y_{j}\Vert_{1}$.
In addition, define the empirical MMD between $X$ and $Y$ as
\[
\mathrm{MMD}(X,Y)=\frac{1}{n(n-1)}\sum_{i=1}^{n}\sum_{j\neq i}K(X_{i},X_{j})+\frac{1}{m(m-1)}\sum_{i=1}^{m}\sum_{j\neq i}K(Y_{i},Y_{j})-\frac{2}{nm}\sum_{i=1}^{n}\sum_{j=1}^{m}K(X_{i},Y_{j}),
\]
where $K(\cdot,\cdot):\mathbb{R}^{d}\times\mathbb{R}^{d}\rightarrow\mathbb{R}$
is a positive definite kernel function.

Then given the target distribution $p(x)$ and points $X=(X_{1},\ldots,X_{n})$
generated by some sampling algorithm, define the adjusted 1-Wasserstein
distance and adjusted MMD between $X$ and $p$ as
\begin{align*}
\tilde{W}(X,p) & =W(X,Y)-W(Y,\tilde{Y}),\\
\tilde{\mathrm{MMD}}(X,p) & =\mathrm{MMD}(X,Y)-\mathrm{MMD}(Y,\tilde{Y}),
\end{align*}
respectively, where $Y,\tilde{Y}\sim p(x)$ are two independent random
samples of size $n$ coming from the true $p(x)$ distribution.

\subsubsection{Hyperparameter setting}

\label{subsec:hyperparameter}

\paragraph{Gaussian mixture models}

For all MCMC samplers, the initial states are generated using standard
normal distribution, the first 200 steps are dropped as burn-in, and
the subsequent 1000 points are collected as the sample. Metropolis--Hastings
uses a normal distribution with standard deviation $\sigma_{\mathrm{MH}}=0.2$
for random walks. Hamiltonian Monte Carlo uses a step size of $\varepsilon=0.2$
and five leapfrog steps in each iteration. The parallel tempering
method uses five parallel chains, with inverse temperature parameters
equally spaced in the logarithmic scale from 0.1 to 1. At each inverse
temperature $\beta$, parallel tempering uses a normal distribution
with standard deviation $\text{\ensuremath{\sigma_{\mathrm{PT},\beta}}}=0.2/\sqrt{\beta}$
for random walks.

TemperFlow is initialized by a standard normal distribution and $\beta_{0}=0.1$.
The discounting factor for adaptive $\beta$ selection is set to $\alpha=0.5$.
For each $\beta_{k}$, the $L^{2}$ sampler is run for 2000 iterations
if $\beta_{k}<0.5$, and for 1000 iterations otherwise.

\paragraph{Copula-generated distributions}

The hyperparameter setting is similar to that of the Gaussian mixture
models, except that $\sigma_{\mathrm{MH}}=\varepsilon=0.1$, $\text{\ensuremath{\sigma_{\mathrm{PT},\beta}}}=0.1/\sqrt{\beta}$,
and $\alpha=0.7$.

For the computing time benchmark in Table \ref{tab:benchmark}, each
MCMC sampler is run for 10200 iterations, with the first 200 steps
dropped as burn-in.

\subsection{Additional Experiments}

\label{sec:additional_exp}

\subsubsection{KL sampler for normal mixture distributions}

In this experiment we extend the example in Section \ref{subsec:challenges},
and study the impact of the gap between modes on the sampling quality.
The first row of Figure \ref{fig:kl_sampler_mix} is the same as the
original example, which applies the KL sampler to the mixture distribution
$p_{m}(x)\sim0.7\cdot N(1,1)+0.3\cdot N(\mu,0.25)$ with $\mu=8$.
The second and third rows of Figure \ref{fig:kl_sampler_mix} reduce
the mean gap $\mu$ to 6 and 4, respectively. It can be observed that
unless the two modes are sufficiently close to each other, the KL
sampler would only capture one mode. This finding also validates the
use of tempered distributions in the proposed framework, as tempering
has similar effects to make the modes more connected.

\begin{figure}[h]
\begin{centering}
\includegraphics[width=0.99\textwidth]{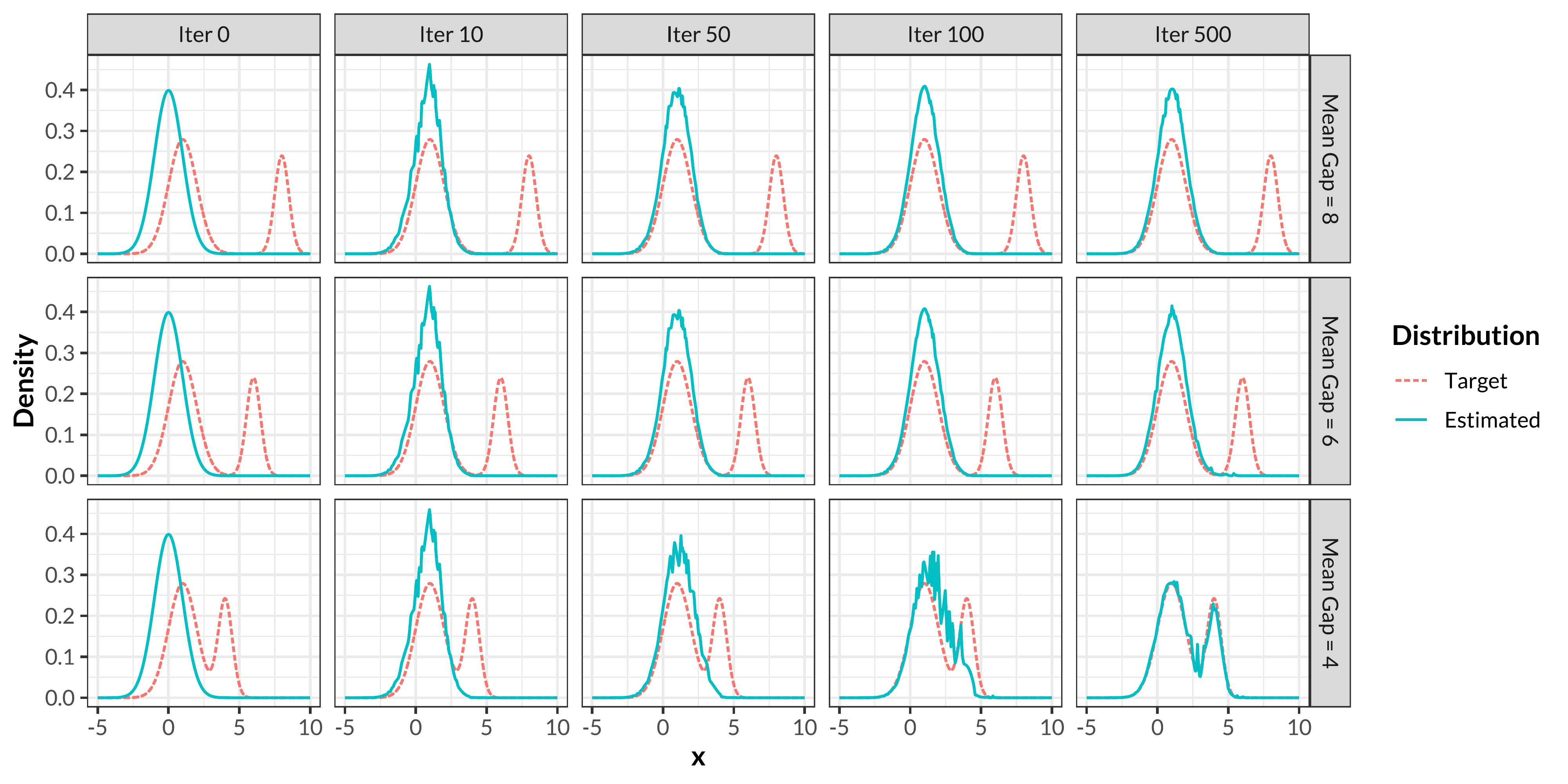}
\par\end{centering}
\caption{\label{fig:kl_sampler_mix}Applying the KL sampler to normal mixture
distributions with varying mean gaps. Each column stands for one iteration
in the optimization process.}
\end{figure}

\subsubsection{Replacing $L^{2}$ sampler with KL sampler in TemperFlow}

In this experiment, we test whether combining the tempering method
and KL sampler could achieve the same performance as TemperFlow. This
can be viewed as an ablation experiment that replaces the $L^{2}$
sampler with the KL sampler in TemperFlow. We revisit the Gaussian
mixture models in Section \ref{subsec:gmm}, and learn the samplers
using TemperFlow and its modified version, respectively. We intentionally
let both methods use the same $\beta$-sequence for factor control,
and also include the results of the KL sampler without tempering for
reference. The results are shown in Figures \ref{fig:ablation_circle},
\ref{fig:ablation_cross}, and \ref{fig:ablation_grid}.

\begin{figure}[h]
\begin{centering}
\subfloat[]{\begin{centering}
\includegraphics[width=0.25\textwidth]{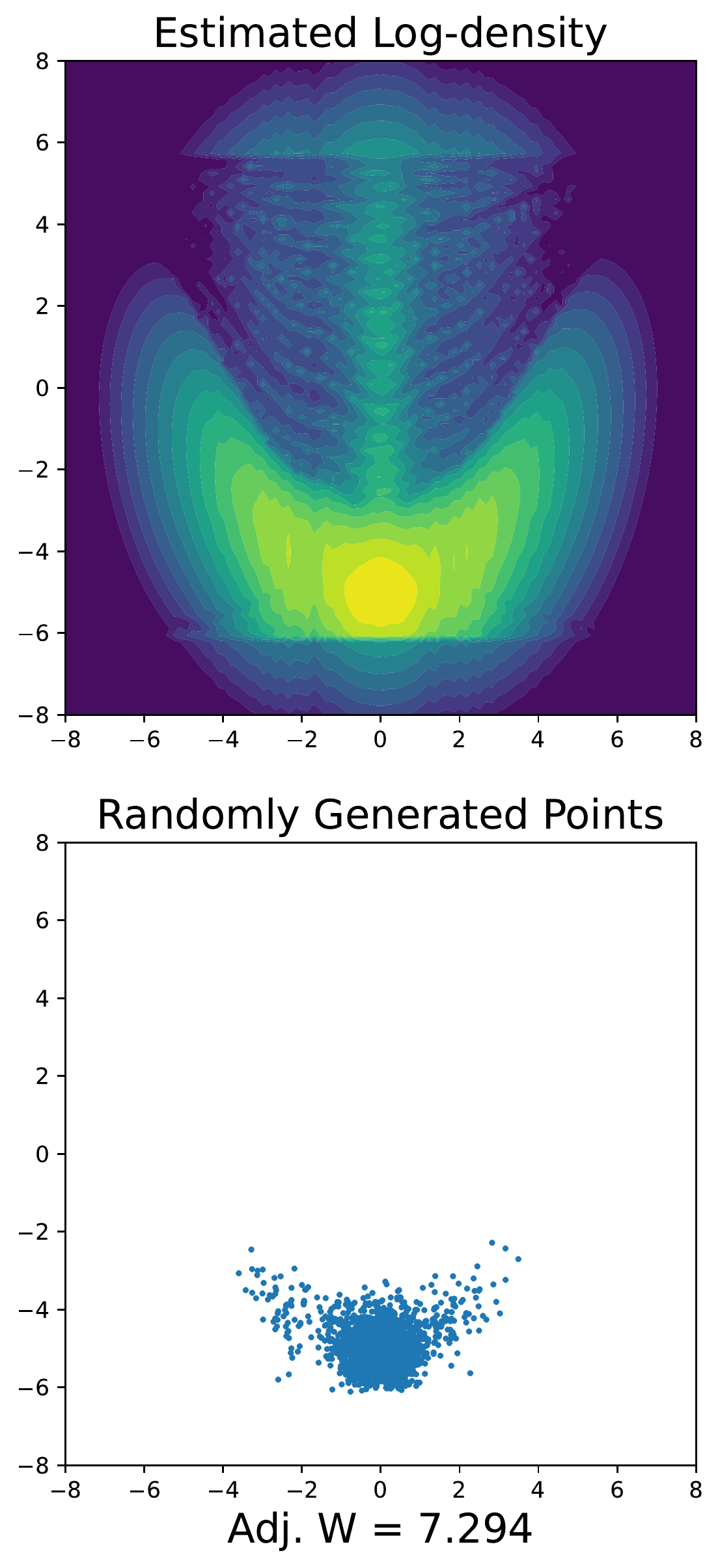}
\par\end{centering}
} \subfloat[]{\begin{centering}
\includegraphics[width=0.25\textwidth]{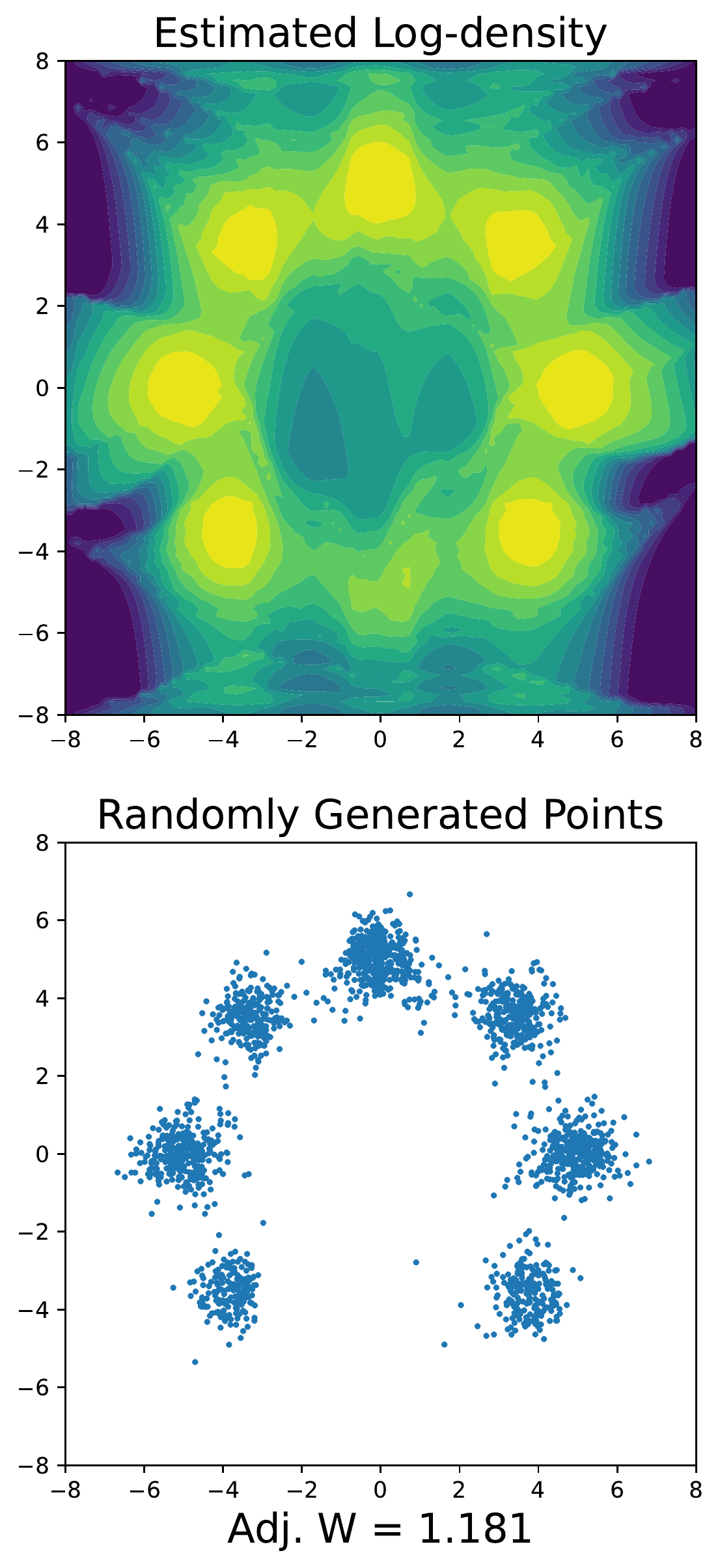}
\par\end{centering}
} \subfloat[]{\begin{centering}
\includegraphics[width=0.25\textwidth]{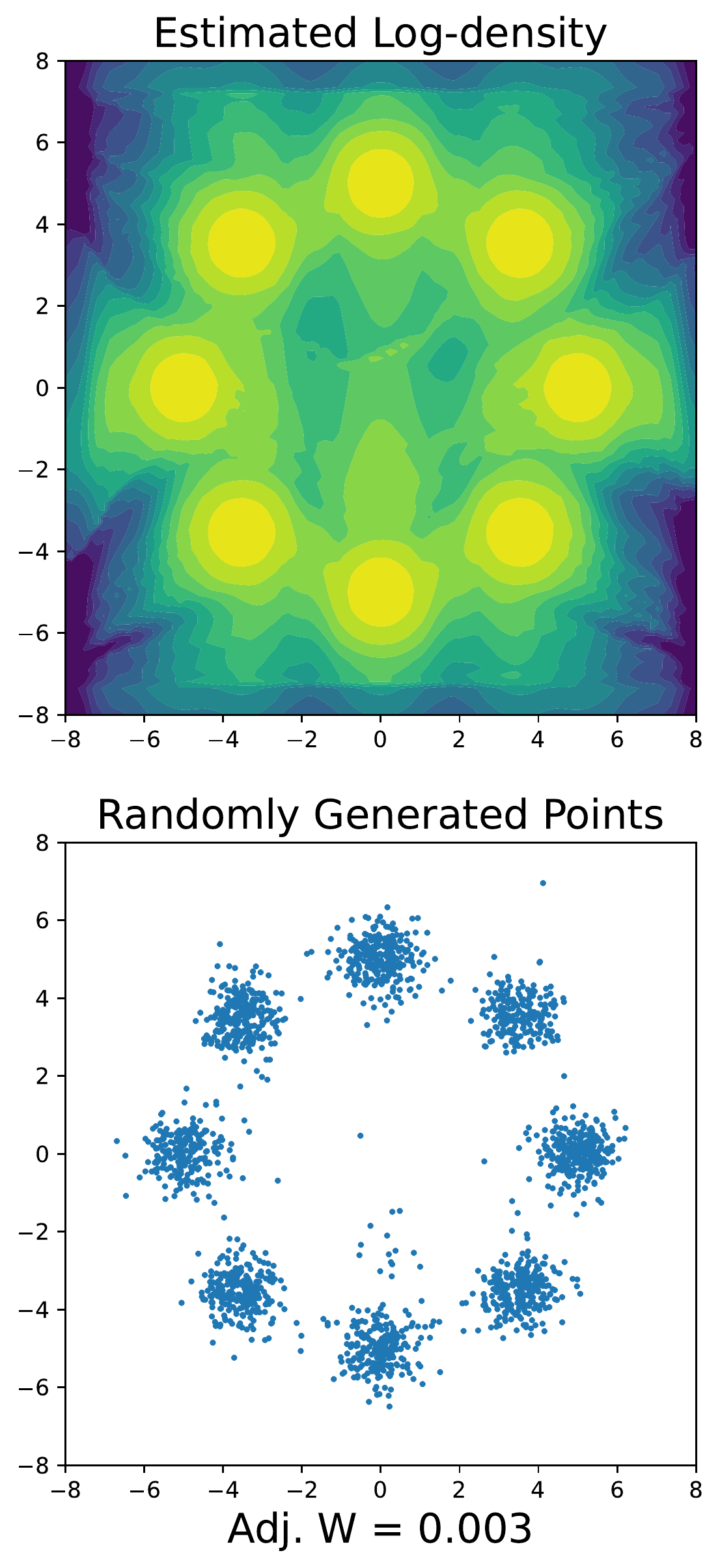}
\par\end{centering}
}
\par\end{centering}
\caption{\label{fig:ablation_circle}Sampling from the \textquotedblleft Circle\textquotedblright{}
distribution using (a) KL sampler, (b) KL sampler with tempering,
and (c) TemperFlow.}
\end{figure}

\begin{figure}
\begin{centering}
\subfloat[]{\begin{centering}
\includegraphics[width=0.25\textwidth]{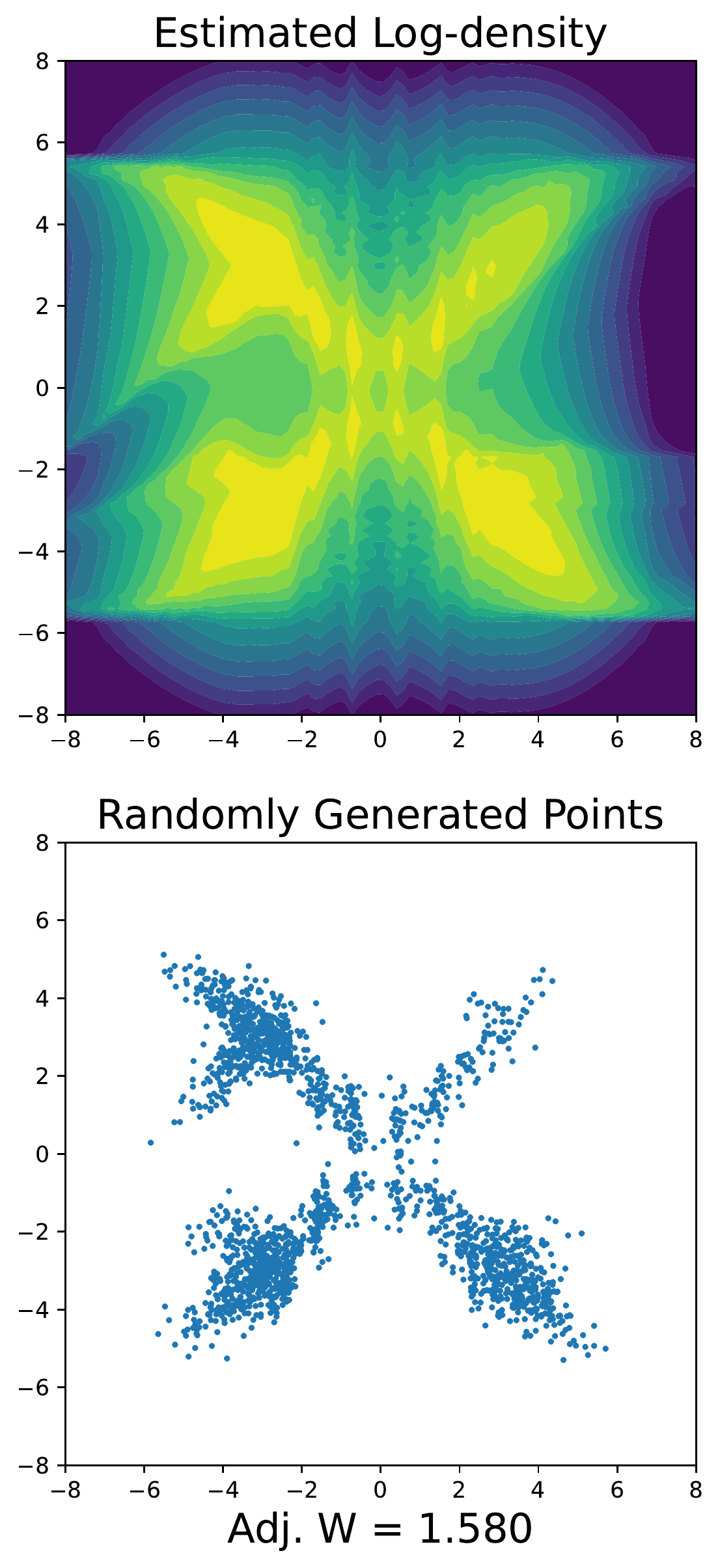}
\par\end{centering}
} \subfloat[]{\begin{centering}
\includegraphics[width=0.25\textwidth]{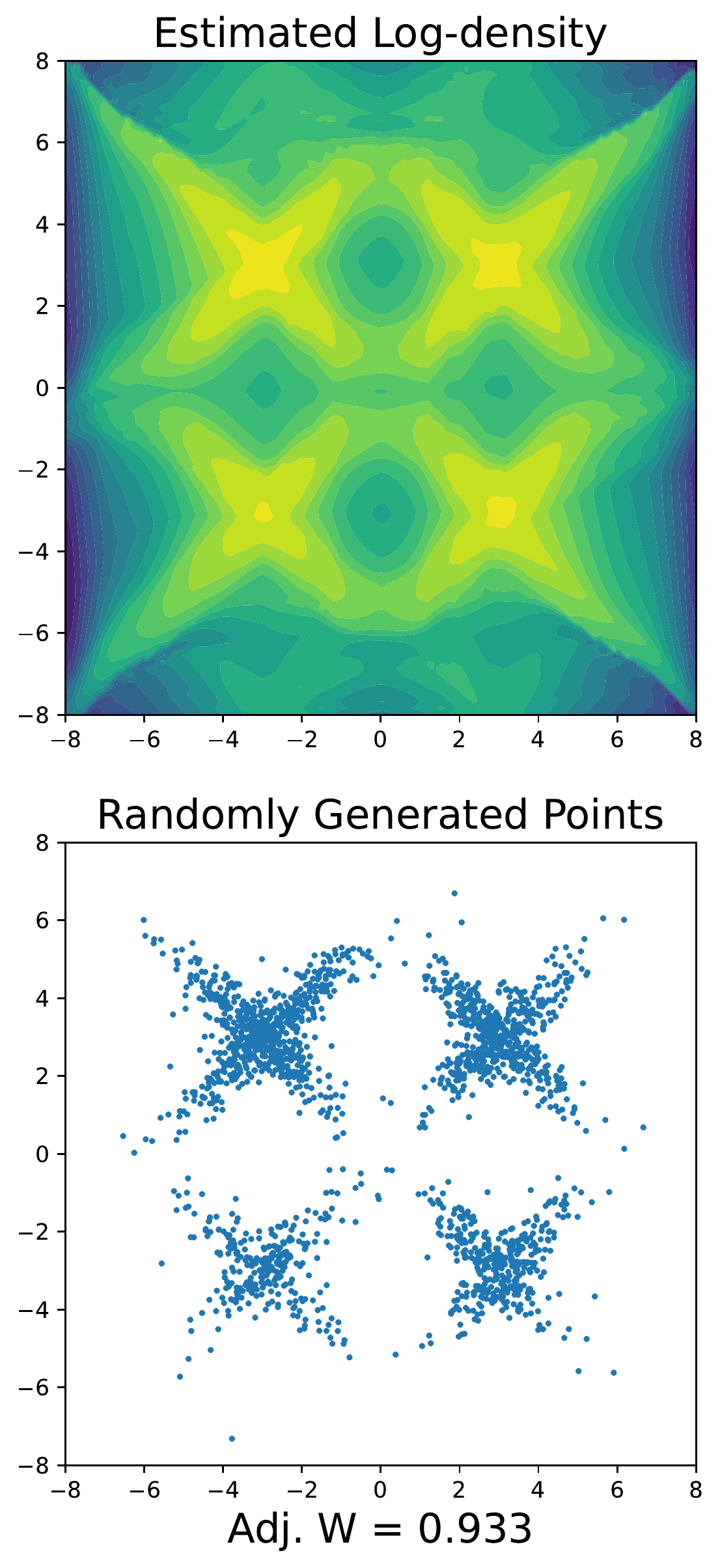}
\par\end{centering}
} \subfloat[]{\begin{centering}
\includegraphics[width=0.25\textwidth]{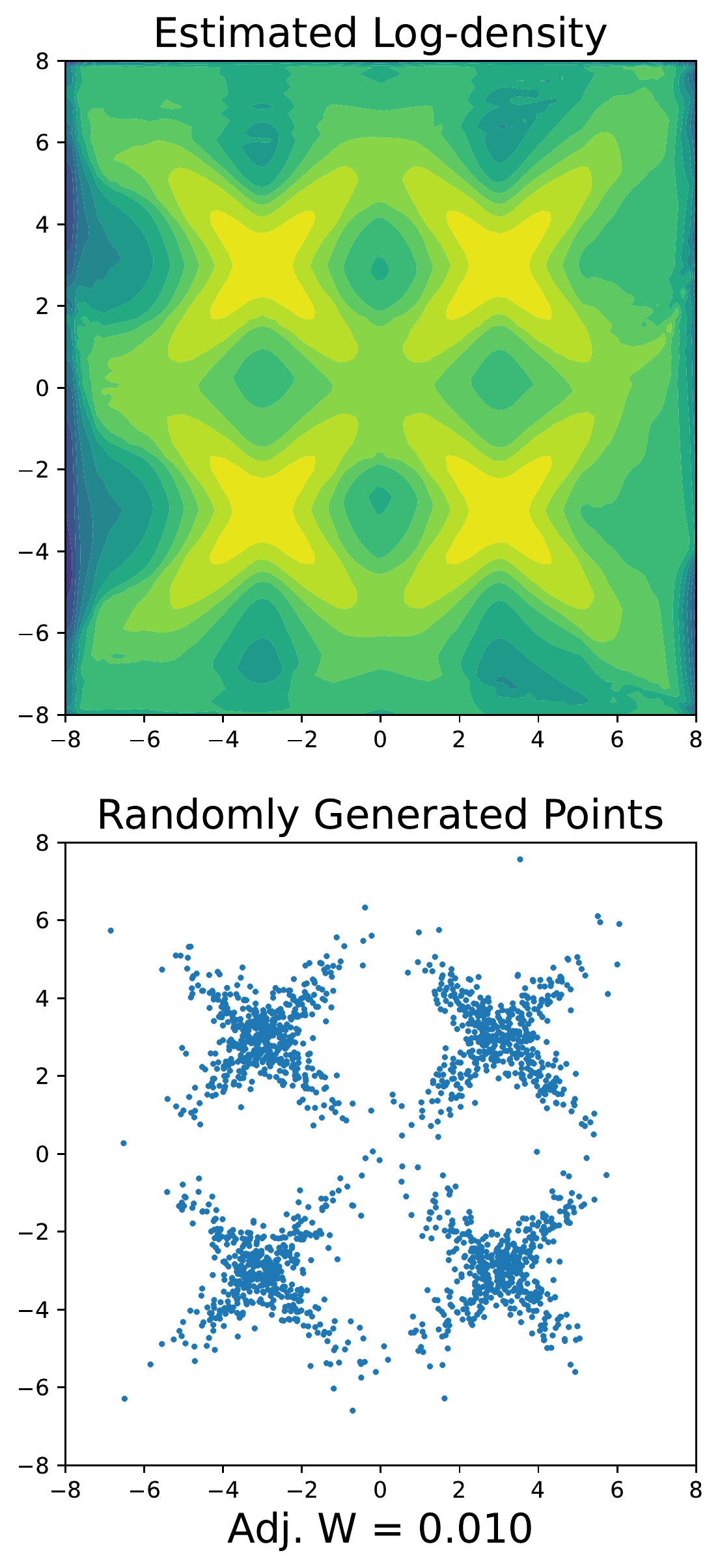}
\par\end{centering}
}
\par\end{centering}
\caption{\label{fig:ablation_cross}Sampling from the \textquotedblleft Cross\textquotedblright{}
distribution using (a) KL sampler, (b) KL sampler with tempering,
and (c) TemperFlow.}
\end{figure}

\begin{figure}
\begin{centering}
\subfloat[]{\begin{centering}
\includegraphics[width=0.25\textwidth]{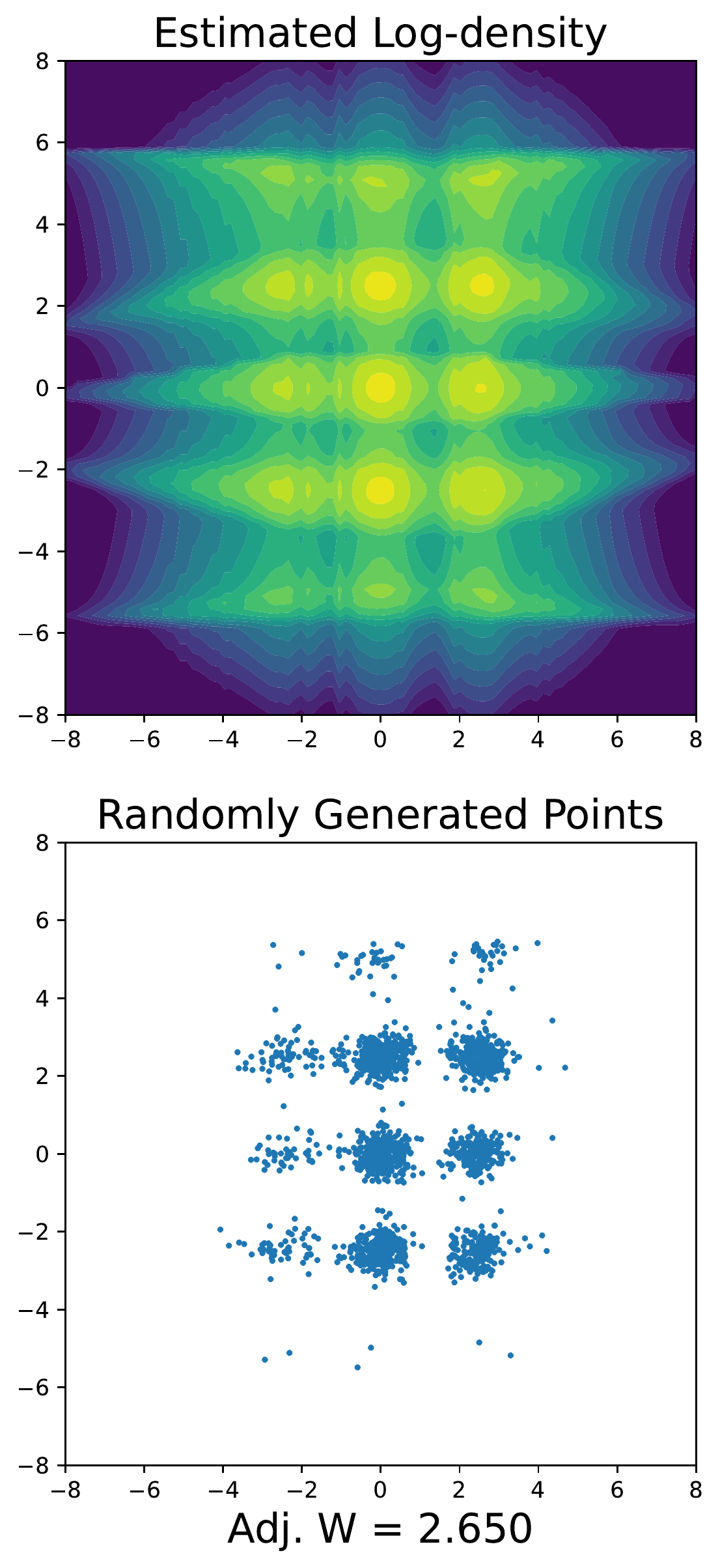}
\par\end{centering}
} \subfloat[]{\begin{centering}
\includegraphics[width=0.25\textwidth]{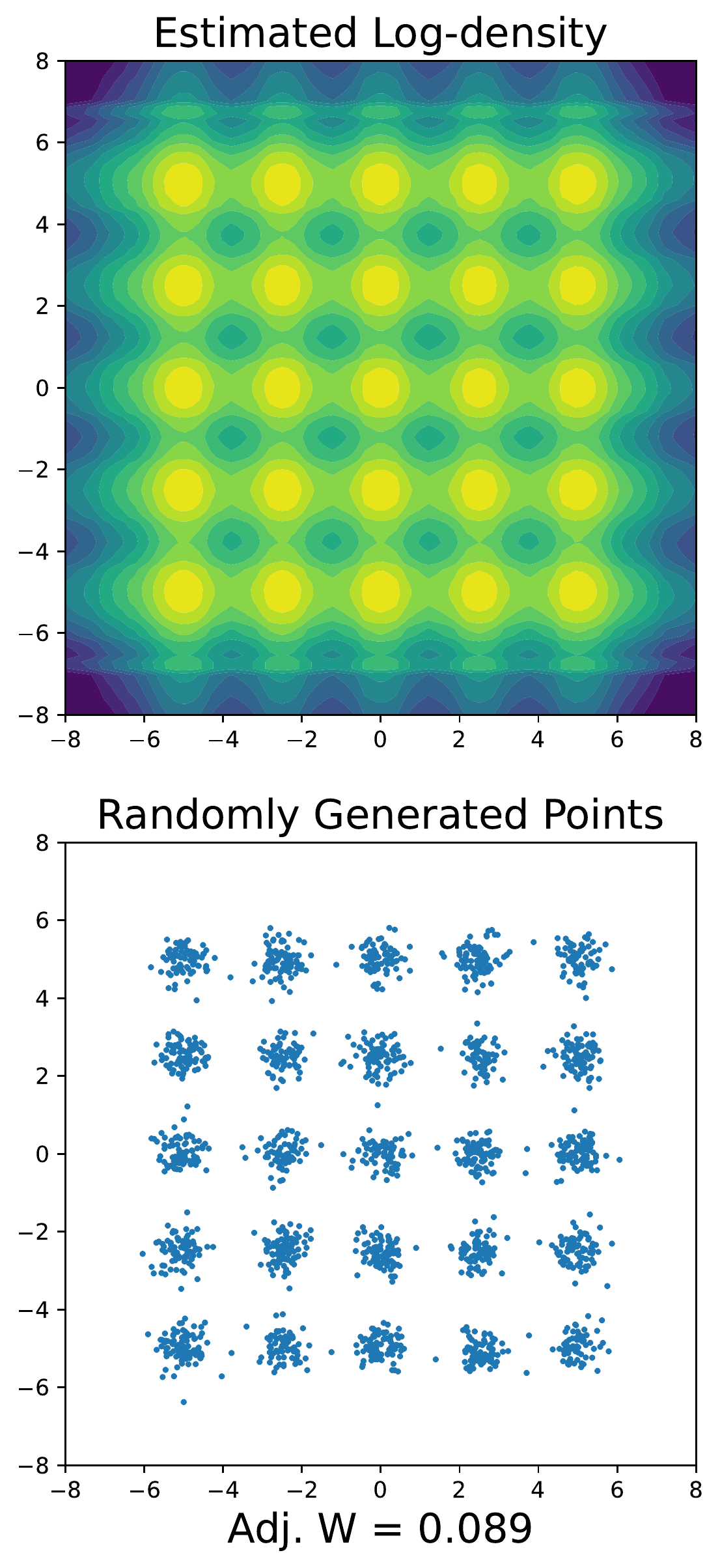}
\par\end{centering}
} \subfloat[]{\begin{centering}
\includegraphics[width=0.25\textwidth]{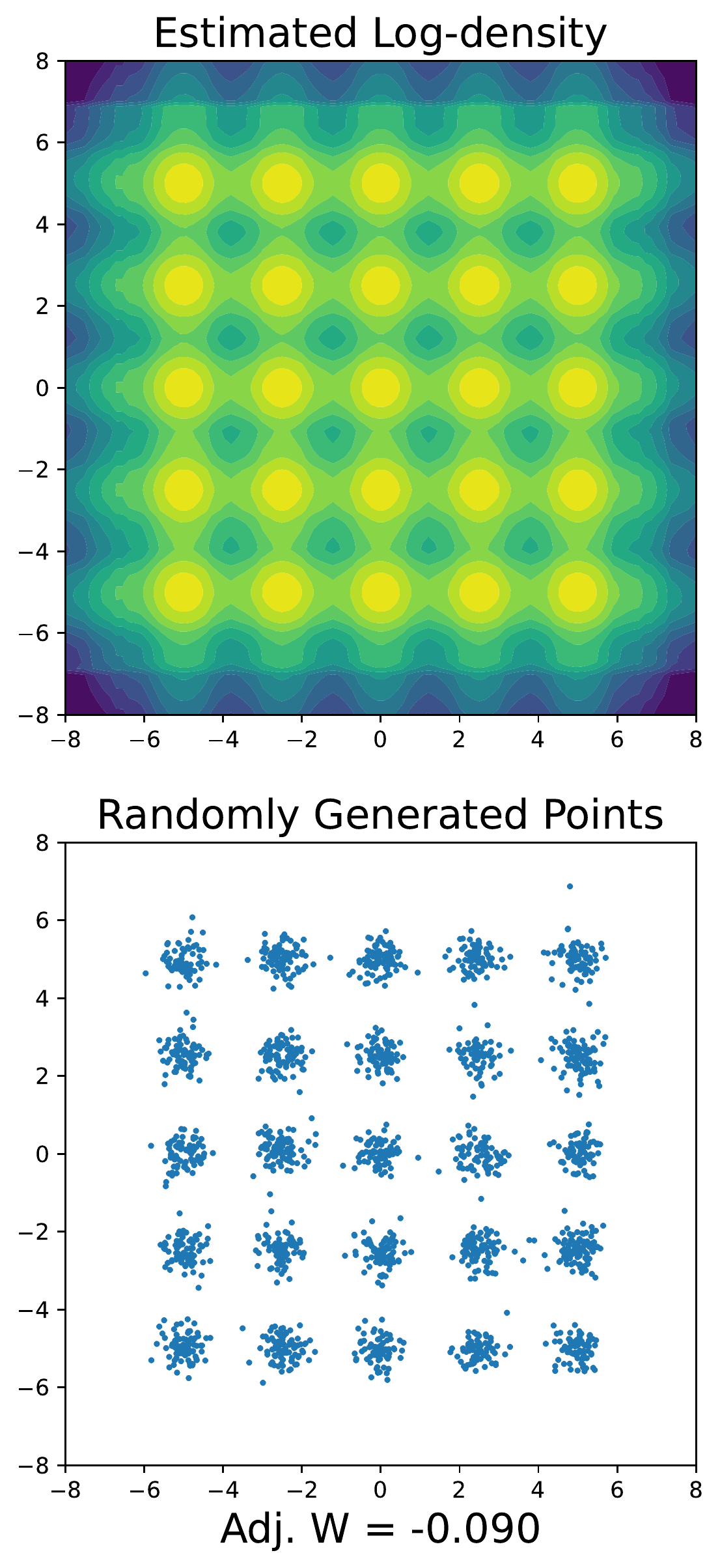}
\par\end{centering}
}
\par\end{centering}
\caption{\label{fig:ablation_grid}Sampling from the \textquotedblleft Grid\textquotedblright{}
distribution using (a) KL sampler, (b) KL sampler with tempering,
and (c) TemperFlow.}
\end{figure}

From the plots we can find that, in general, tempering greatly improves
the quality of KL samplers. However, even with tempering, KL samplers
tend to incorrectly estimate the probability mass of each mode, which
is consistent with the finding in Figure \ref{fig:demo_l2}. This
suggests that the $L^{2}$ sampler plays an important role in TemperFlow
that cannot be simply replaced by the KL sampler.

\subsubsection{Random samples of copula-generated distribution}

Figure \ref{fig:experiments_copula} shows the pairwise scatterplots
and density contour plots of the generated samples by parallel tempering
and TemperFlow, respectively, for the experiment in Section \ref{subsec:copula}
with $s=d=8$.

\begin{figure}
\begin{centering}
\includegraphics[width=0.65\textwidth]{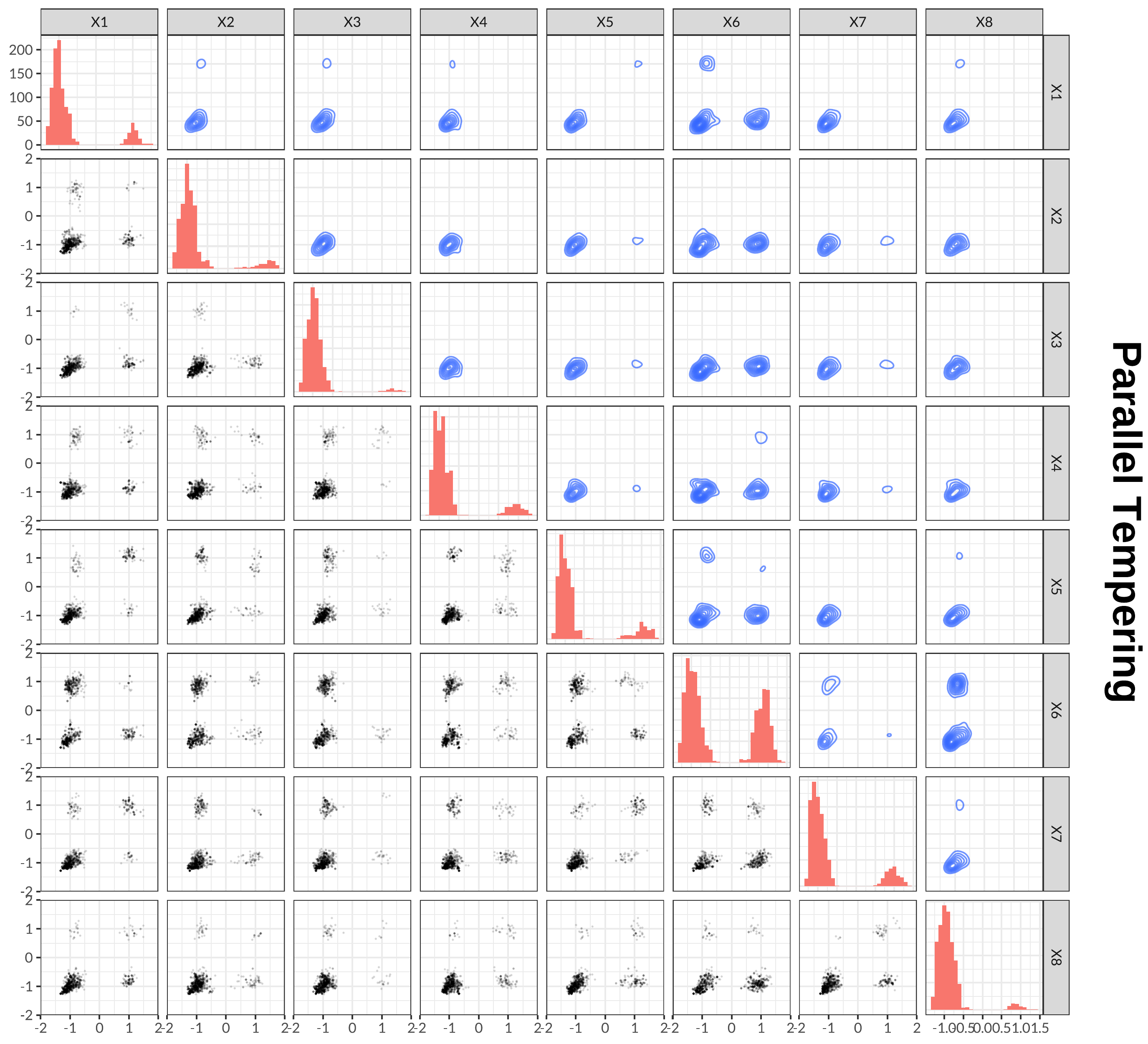}
\par\end{centering}
\begin{centering}
\includegraphics[width=0.65\textwidth]{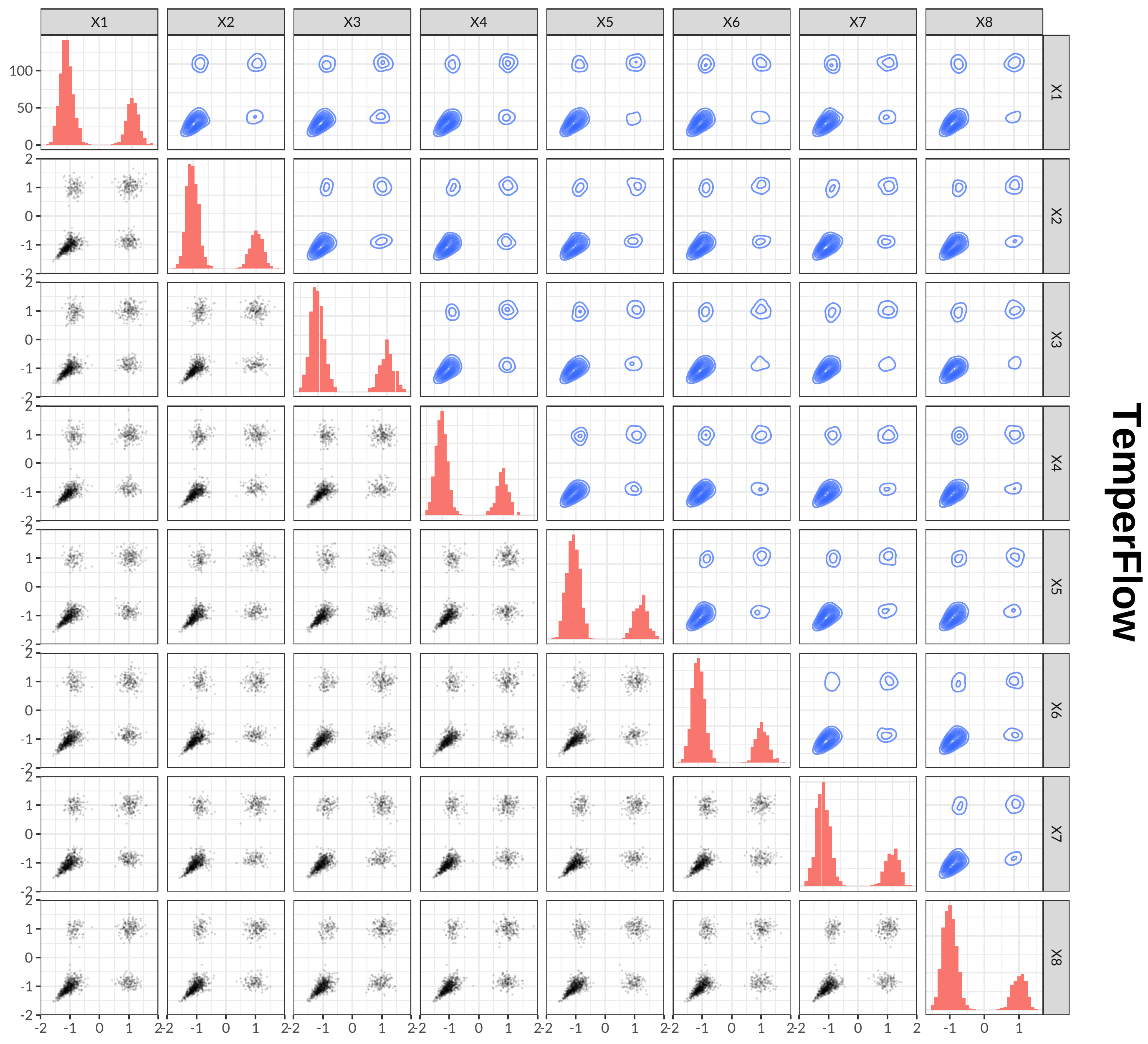}
\par\end{centering}
\caption{\label{fig:experiments_copula}Sampling results of parallel tempering
and TemperFlow for the copula-generated distribution with $s=d=8$.
The lower triangular part shows pairwise scatterplots of the generated
samples, and upper triangular part are the density contour plots.
Plots on the diagonal are histograms for the marginal distributions.}
\end{figure}

\subsubsection{Increasing the number of iterations for MCMC samplers}

\label{subsec:mcmc_more_iters}

Based on the experiment in Section \ref{subsec:copula}, here we increase
the number of iterations for MCMC samplers, and again compare their
results with TemperFlow. Similar to the previous setting, we drop
the first 200 steps as burn-in, but then take one data point every
ten iterations. Finally, 1000 data points are collected, so the total
number of iterations for MCMC samplers is 10200. Figure \ref{fig:experiments_copula_error_mcmc10}
demonstrates the sampling errors on this new setting. Compared with
Figure \ref{fig:experiments_copula_error}, we can observe that the
overall pattern is very similar, although MCMC is now run with ten
times the iterations as before.

\begin{figure}[h]
\begin{centering}
\includegraphics[width=0.9\textwidth]{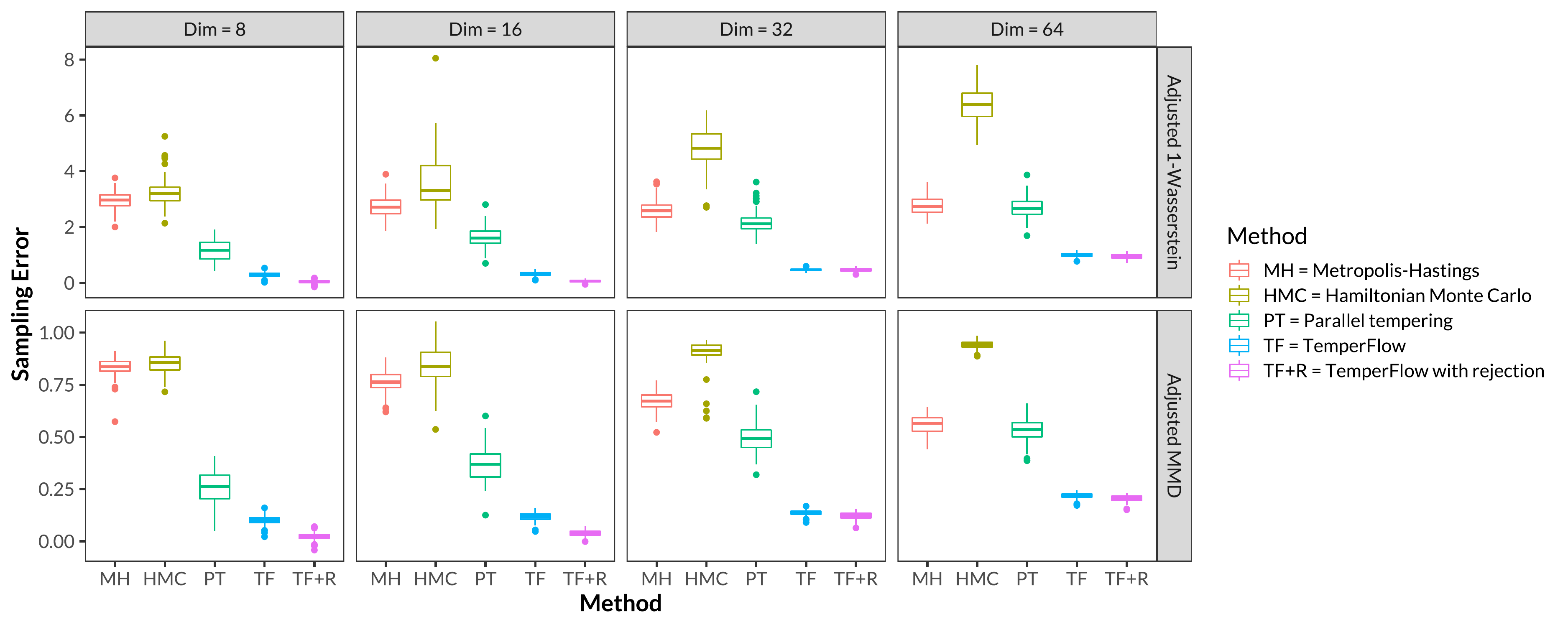}
\par\end{centering}
\caption{\label{fig:experiments_copula_error_mcmc10}Sampling errors of different
methods for the copula-generated distribution, with MCMC run for ten
times longer than in Figure \ref{fig:experiments_copula_error_mcmc10}.}
\end{figure}

\subsubsection{Computing time of sampling methods}

We report the computing time of different samplers in Table \ref{tab:benchmark}
for the distributions in Section \ref{subsec:copula}. All samplers
are implemented and run on GPUs. As introduced in the article, TemperFlow
has a training stage, whose computing time is given in the first row
of Table \ref{tab:benchmark}. Interestingly, the empirical results
show that the training time roughly scales linearly with the dimension
of the target distribution. Once the sampler is trained, the generation
time of TemperFlow can be virtually ignored, even with a post-processing
rejection sampling step. In contrast, MCMC methods have no training
costs, but their generation time is proportional to the number of
Markov chain iterations. It must also be emphasized that the times
for the MCMC methods in Table \ref{tab:benchmark} only reflect the
computing cost after a fixed number of iterations, but clearly Figures
\ref{fig:experiments_copula_error} and \ref{fig:experiments_copula_error_mcmc10}
suggest that their sampling errors are much higher than those of TemperFlow.
Therefore, in practice, a much larger number of iterations is needed
for MCMC samplers, which would dramatically increase the total computing
time.

\begin{table}[h]
\caption{\label{tab:benchmark}Computing time for different sampling methods
and dimensions. The first row shows the training time of TemperFlow,
and the number in the parenthesis is the number of adaptive $\beta$'s
used in Algorithm \ref{alg:temperflow}. Remaining rows show the time
to generate 10,000 points by each algorithm. All timing values are
in seconds.}

\centering{}%
\begin{tabular}{c>{\centering}p{0.15\textwidth}>{\centering}p{0.15\textwidth}>{\centering}p{0.15\textwidth}>{\centering}p{0.15\textwidth}}
\toprule 
 & $p=8$ & $p=16$ & $p=32$ & $p=64$\tabularnewline
\midrule
\midrule 
(TemperFlow Training) & 133.7 (17) & 309.9 (21) & 598.6 (22) & 1475.5 (30)\tabularnewline
TemperFlow/$10^{4}$ & 0.00206 & 0.00449 & 0.00899 & 0.0169\tabularnewline
TemperFlow+Rej./$10^{4}$ & 0.0148 & 0.0263 & 0.0547 & 0.0928\tabularnewline
MH/$10^{4}$ & 11.8 & 16.0 & 15.8 & 15.6\tabularnewline
HMC/$10^{4}$ & 80.9 & 133.0 & 131.0 & 130.0\tabularnewline
Parallel tempering/$10^{4}$ & 24.9 & 32.2 & 32.5 & 31.9\tabularnewline
\bottomrule
\end{tabular}
\end{table}

\subsubsection{Changing the base measure in TemperFlow}

When the target distribution $p(x)$ is supported on the whole Euclidean
space $\mathbb{R}^{d}$, TemperFlow by default uses the standard multivariate
normal as the base measure $\mu_{0}$. If $p(x)$ is known to be supported
on a compact set, for example, a hypercube $[-B,B]^{d}$, then $\mu_{0}$
can also be taken as a distribution supported on $[-B,B]^{d}$, such
as a uniform distribution. In Figure \ref{fig:2d_trunc} we show such
an example. The target distribution is the same as the ``Circle''
Gaussian mixture model in Section \ref{subsec:gmm}, but we restrict
its density to the $[-5,5]^{2}$ square. Then we choose the base measure
$\mu_{0}$ as the uniform distribution on $[-5,5]^{2}$, and use TemperFlow
to estimate the transport map $T$. Figure \ref{fig:2d_trunc} shows
the tempered distribution flow $T_{\sharp}^{(k)}\mu_{0}$ for different
inverse temperature parameters $\beta_{k}$. It can be found that
TemperFlow again estimates the target distribution well.

\begin{figure}[h]
\begin{centering}
\includegraphics[width=0.99\textwidth]{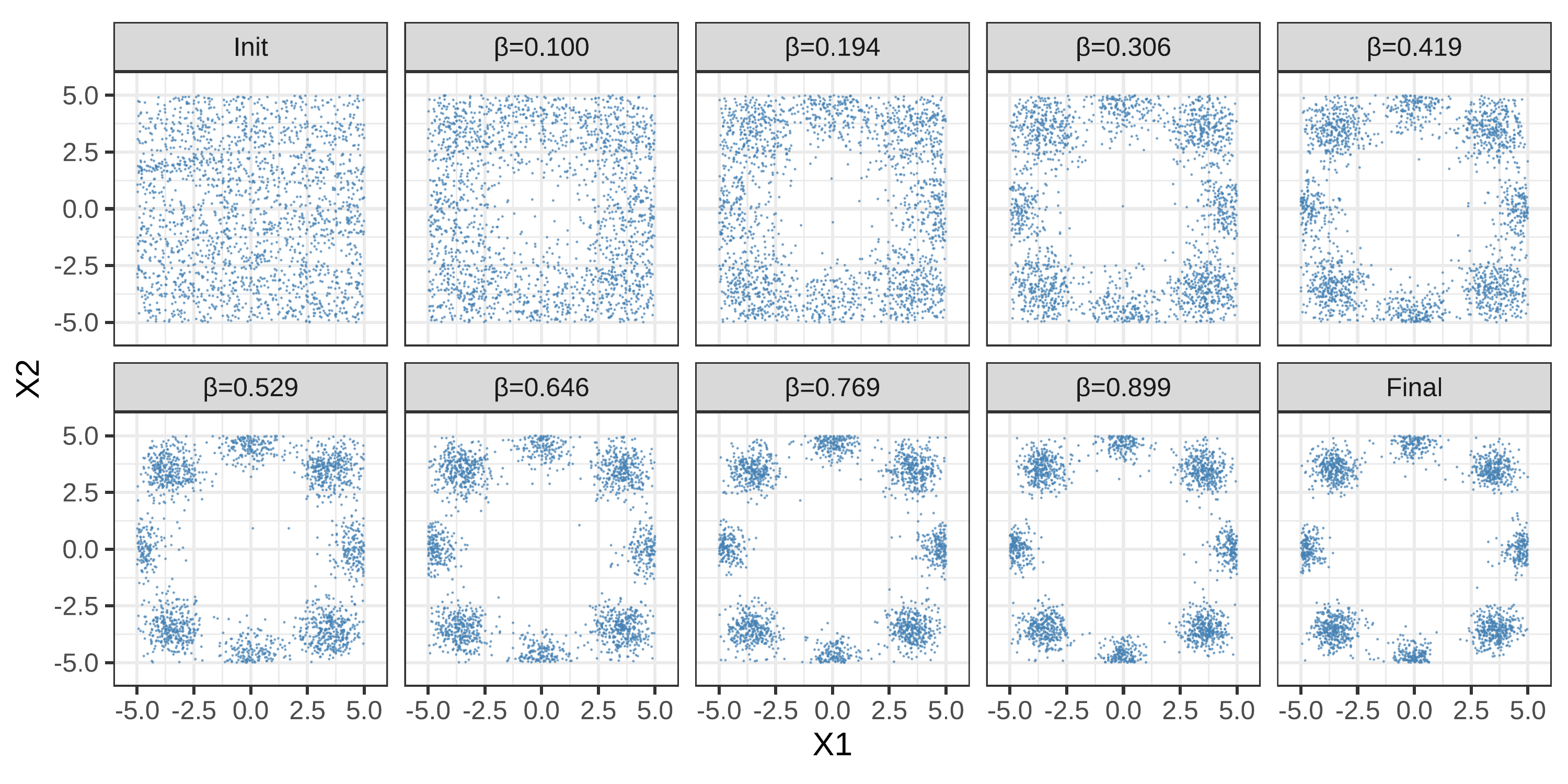}
\par\end{centering}
\caption{\label{fig:2d_trunc}Applying TemperFlow to a distribution supported
on a closed square. The base measure $\mu_{0}$ is taken to be a uniform
distribution.}
\end{figure}

\subsubsection{Comparing different INN architectures}

Theoretically, TemperFlow can utilize any variant of INNs to construct
$T$, and in practice, we find that spline-based models, such as linear
rational splines \citep{dolatabadi2020invertible}, are very powerful,
and we use them as the default choice in TemperFlow. To test the performance
of different INN architectures, we first consider the two-dimensional
distributions studied in Section \ref{subsec:gmm}, and include three
representative INNs for comparison: the inverse autoregressive flows
\citep{kingma2016improving}, the linear rational splines \citep{dolatabadi2020invertible},
and the neural spline flows based on quadratic rational splines \citep{durkan2019neural}.
Figure \ref{fig:2d_maps} shows the generated samples using different
INN architectures. It is clear that the inverse autoregressive flow
is able to capture the outline of the target distribution but is inferior
to the other two in the details. Due to this reason, we only study
the two spline-based INNs for higher-dimensional problems.

\begin{figure}[h]
\begin{centering}
\includegraphics[width=0.8\textwidth]{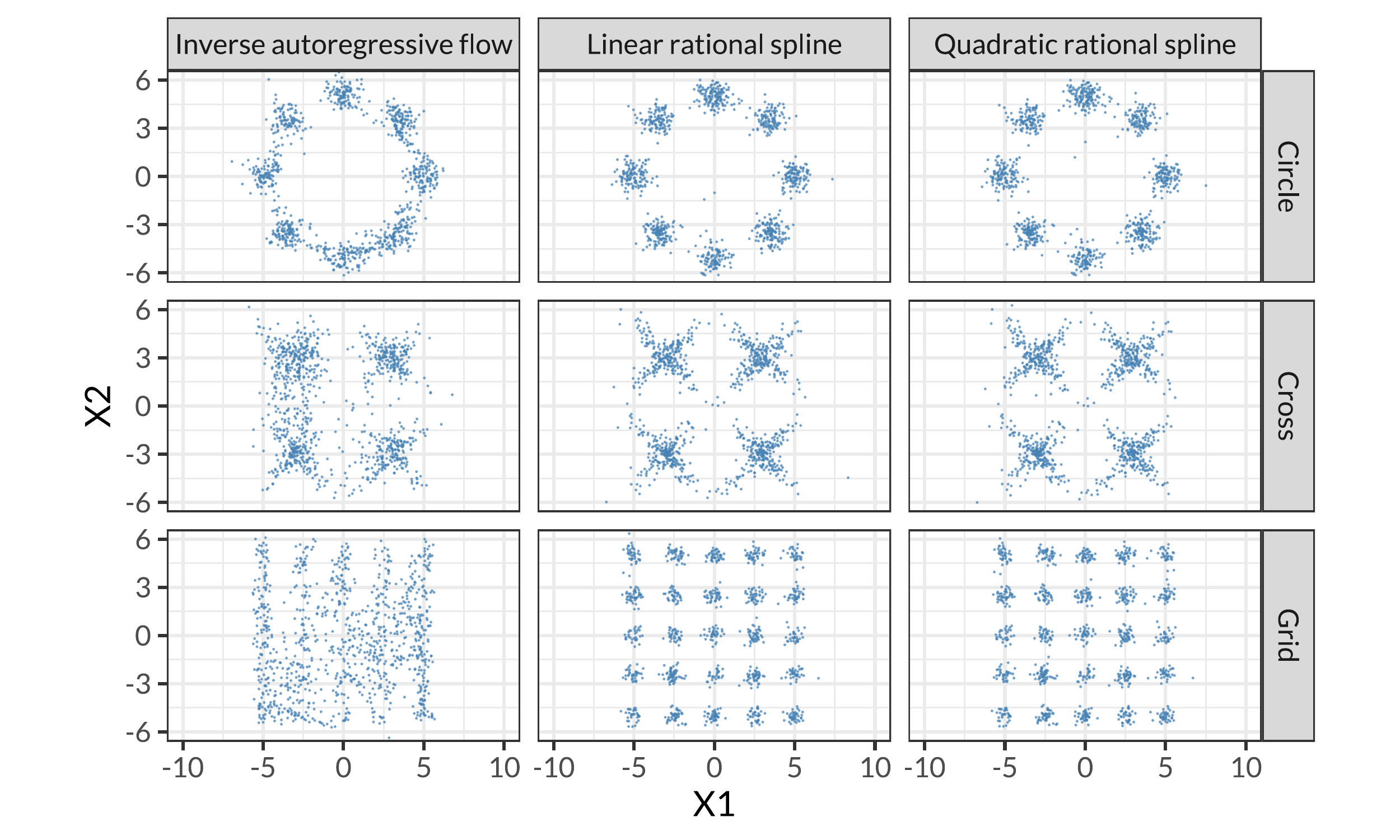}
\par\end{centering}
\caption{\label{fig:2d_maps}TemperFlow samples based on different INN architectures.}
\end{figure}

Figure \ref{fig:clayton_maps} recaps the experiments in Section \ref{subsec:copula},
but compares TemperFlow samplers based on linear rational splines
and quadratic rational splines, respectively. Table \ref{tab:benchmark_splines}
also shows the training and sampling times for each sampler. It can
be observed that the two INNs behave similarly, with linear rational
splines slightly better than quadratic rational splines in terms of
sampling errors, but marginally slower in terms of training time.

\begin{figure}[h]
\begin{centering}
\includegraphics[width=0.99\textwidth]{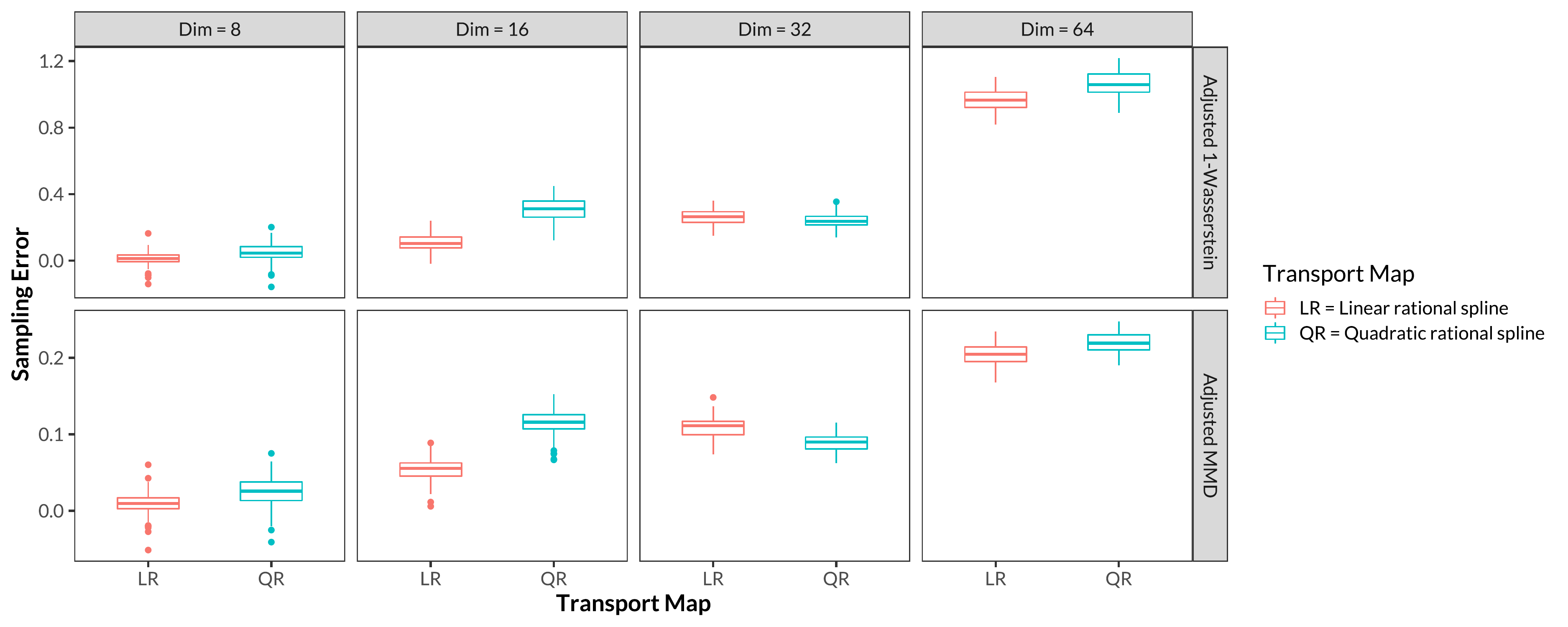}
\par\end{centering}
\caption{\label{fig:clayton_maps}Sampling errors of TemperFlow for the copula-generated
distribution, based on different INN architectures.}
\end{figure}

\begin{table}[h]
\caption{\label{tab:benchmark_splines}Computing time for TemperFlow on the
copula-generated distribution, based on two different INN architectures.
The interpretation of the numbers is similar to that of Table \ref{tab:benchmark}.}

\centering{}%
\begin{tabular}{cc>{\centering}p{0.15\textwidth}>{\centering}p{0.15\textwidth}>{\centering}p{0.15\textwidth}>{\centering}p{0.15\textwidth}}
\toprule 
 &  & $p=8$ & $p=16$ & $p=32$ & $p=64$\tabularnewline
\midrule
\midrule 
\multirow{3}{*}{LR} & Training & 135.6 (17) & 320.6 (21) & 642.0 (23) & 1534.7 (30)\tabularnewline
 & Sampling/$10^{4}$ & 0.00215 & 0.00465 & 0.00892 & 0.0173\tabularnewline
 & Sampling+Rej./$10^{4}$ & 0.0131 & 0.0189 & 0.0401 & 0.0783\tabularnewline
\cmidrule{1-2} \cmidrule{2-2} 
\multirow{3}{*}{QR} & Training & 122.9 (17) & 284.1 (21) & 567.7 (23) & 1366.2 (30)\tabularnewline
 & Sampling/$10^{4}$ & 0.00156 & 0.00406 & 0.00806 & 0.0156\tabularnewline
 & Sampling+Rej./$10^{4}$ & 0.0119 & 0.0199 & 0.0358 & 0.0779\tabularnewline
\bottomrule
\end{tabular}
\end{table}

\subsubsection{Comparing different discounting factors $\alpha$}

In this section we study the impact of the discounting factor $\alpha$
on adaptive temperature selection in TemperFlow. We again consider
the copula-generated distribution in Section \ref{subsec:copula},
and estimate the transport map using $\alpha=0.5,0.6,0.7,0.8,0.9$.
Figure \ref{fig:clayton_alphas} shows the final sampling error of
each method, and Table \ref{tab:benchmark_alphas} gives the training
time and the number of $\beta$'s used. The sampler with $\alpha=0.9$
has abnormal results for dimensions 32 and 64, since we have set the
maximum number of $\beta$'s to be 100, and their computations are
not finished.

In general, by setting a larger $\alpha$ value, one obtains smaller
sampling errors at the expense of more iterations and a longer training
time. There is no definite optimal choice for $\alpha$, but Figure
\ref{fig:clayton_alphas} suggests that TemperFlow is not very sensitive
to the choice of $\alpha$ for $\alpha\ge0.6$. Furthermore, our empirical
results suggest that $\alpha$ can be set small for simple and low-dimensional
distributions, whereas it should be larger for more challenging cases.

\begin{figure}[h]
\begin{centering}
\includegraphics[width=0.99\textwidth]{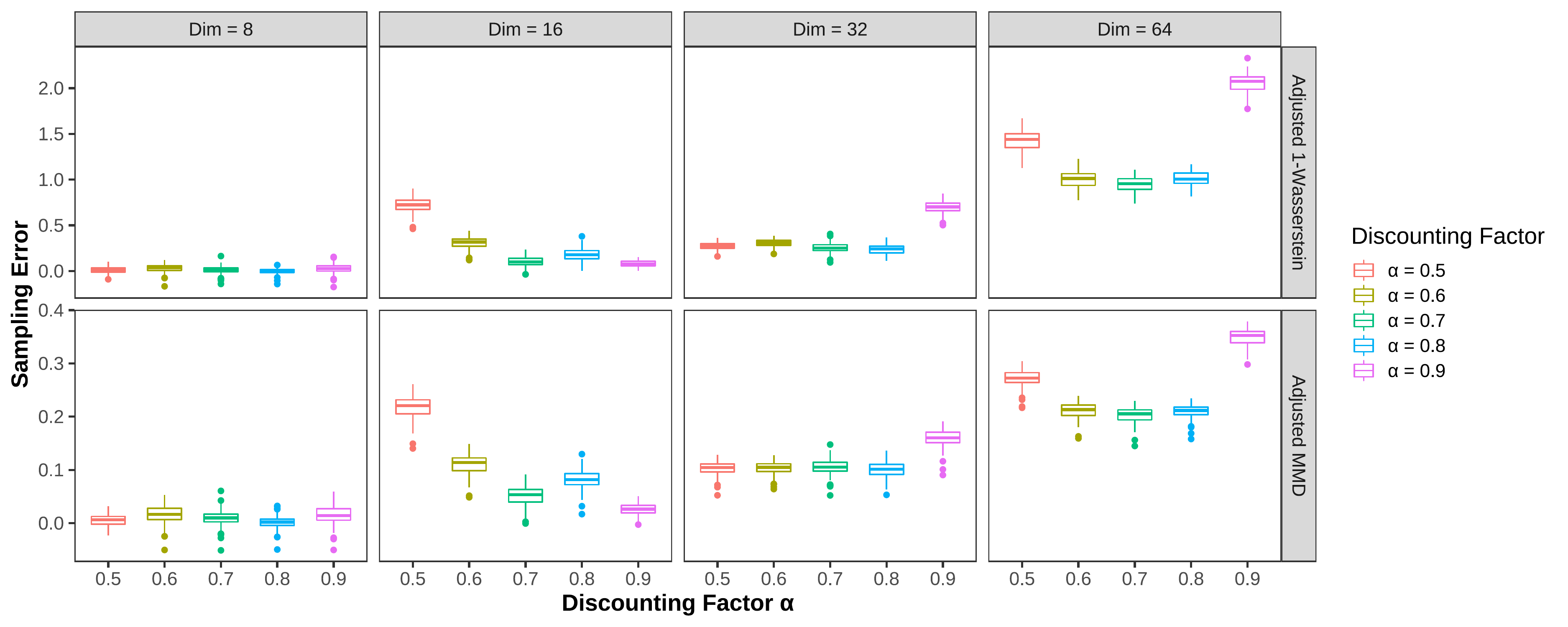}
\par\end{centering}
\caption{\label{fig:clayton_alphas}Sampling errors of TemperFlow for the copula-generated
distribution, based on different INN architectures.}
\end{figure}

\begin{table}[h]
\caption{\label{tab:benchmark_alphas}Training time of TemperFlow on the copula-generated
distribution, based on different $\alpha$ values. The number in the
parenthesis is the number of adaptive $\beta$'s used in Algorithm
\ref{alg:temperflow}, and \textquotedblleft Max\textquotedblright{}
means TemperFlow has used the maximum number of $\beta$'s, which
is set to 100 in the experiments.}

\centering{}%
\begin{tabular}{c>{\centering}p{0.15\textwidth}>{\centering}p{0.15\textwidth}>{\centering}p{0.15\textwidth}>{\centering}p{0.15\textwidth}}
\toprule 
 & $p=8$ & $p=16$ & $p=32$ & $p=64$\tabularnewline
\midrule
\midrule 
$\alpha=0.5$ & 86.3 (10) & 198.4 (12) & 418.4 (14) & 938.3 (17)\tabularnewline
$\alpha=0.6$ & 106.1 (13) & 236.5 (15) & 486.9 (17) & 1170.8 (22)\tabularnewline
$\alpha=0.7$ & 135.6 (17) & 320.6 (21) & 642.0 (23) & 1534.7 (30)\tabularnewline
$\alpha=0.8$ & 206.4 (27) & 459.7 (32) & 984.8 (38) & 2404.4 (50)\tabularnewline
$\alpha=0.9$ & 488.6 (73) & 1104.2 (87) & 2356.0 (Max) & 4730.1 (Max)\tabularnewline
\bottomrule
\end{tabular}
\end{table}

\subsubsection{Variance of the importance sampling estimator}

In TemperFlow, we rely on importance sampling to estimate the normalizing
constant of the tempered density functions. In general, the variance
of importance sampling is largely dependent on the proposal distribution.
For example, to estimate the normalizing constant
\[
Z_{\beta}=\int\exp\{-\beta E(x)\}\mathrm{d}x=\mathbb{E}_{X\sim h(x)}\exp\{-\beta E(X)-\log h(X)\},
\]
the importance sampling estimator $\hat{Z}_{\beta}$ is given by
\[
\hat{Z}_{\beta}=\frac{1}{M}\sum_{i=1}^{M}\frac{e^{-\beta E(X_{i})}}{h(X_{i})}\coloneqq\frac{1}{M}\sum_{i=1}^{M}w(X_{i}),\quad X_{i}\sim h(x),
\]
where $h(x)$ is the proposal distribution. It can be easily verified
that $\mathbb{E}(\hat{Z}_{\beta})=Z_{\beta}$, and
\[
\mathrm{Var}(\hat{Z}_{\beta})=\frac{1}{M}\int\frac{[e^{-\beta E(x)}-Z_{\beta}h(x)]^{2}}{h(x)}\mathrm{d}x=\frac{1}{M}\int[w(x)-Z_{\beta}]^{2}h(x)\mathrm{d}x.
\]
Clearly, if $h(x)\propto e^{-\beta E(x)}$ , then $\mathrm{Var}(\hat{Z}_{\beta})=0$.
In TemperFlow, we use $\hat{r}_{s}(x)\approx Z_{\beta_{s}}^{-1}e^{-\beta_{s}E(x)}$
as the proposal distribution to estimate $Z_{\beta_{s+1}}$, where
$\hat{r}_{s}=T_{\sharp}^{(s)}\mu_{0}$ is the estimated tempered distribution.
Therefore, as long as $\beta_{s}\approx\beta_{s+1}$, the variance
of $\hat{Z}_{\beta_{s+1}}$ would be small.

To verify this claim, we consider the experiment in Section \ref{subsec:copula}
with $d=8$, and extract the $\{\hat{r}_{k}\}$ distributions from
the training process. We then use the following coefficient of variation
to quantify the uncertainty of the importance sampling estimator:
\[
\mathrm{CV}(\hat{r}_{k},\beta_{s})=\frac{\sqrt{\widehat{\mathrm{Var}}(\hat{Z}_{\beta_{s}})}}{\hat{Z}_{\beta_{s}}}\quad\text{with }h(x)=\hat{r}_{k}(x),\quad s=0,1,\ldots,K,\ 0\le k\le s,
\]
where $\widehat{\mathrm{Var}}(\hat{Z}_{\beta})=M^{-1}(M-1)^{-1}\sum_{i=1}^{M}[w(X_{i})-\hat{Z}_{\beta}]^{2}$
is an estimator of $\mathrm{Var}(\hat{Z}_{\beta})$. Intuitively,
we expect that $\mathrm{CV}(\hat{r}_{k},\beta_{s})$ is small if $k$
is close to $s$. Figure \ref{fig:is_variance}(a) shows the sixteen
$\beta$ values selected by TemperFlow, $\beta_{0},\ldots,\beta_{15}$,
and Figure \ref{fig:is_variance}(b) illustrates the values of $\log(\hat{Z}_{\beta_{s}})$
with $\hat{r}_{k}$ used as the proposal distribution. As expected,
each column of the matrix has similar values of $\log(\hat{Z}_{\beta_{s}})$,
since the expectation of $\hat{Z}_{\beta_{s}}$ does not depend on
the proposal distribution. However, different $\hat{r}_{k}$'s result
in different variances of $\hat{Z}_{\beta_{s}}$. Figure \ref{fig:is_variance}(c)
shows the values of $\log[\mathrm{CV}(\hat{r}_{k},\beta_{s})]$ with
different combinations of $(k,s)$, from which we can observe that
the diagonal elements, corresponding to $k=s$, has the smallest uncertainty.
This is consistent with the theoretical analysis of importance sampling.
In TemperFlow, we only have $\hat{r}_{k}$, $k<s$, to estimate $Z_{\beta_{s}}$,
so $\hat{r}_{s-1}$ is to some extent the optimal choice for the proposal
distribution.

\begin{figure}[h]
\begin{centering}
\subfloat[]{\includegraphics[width=0.32\textwidth]{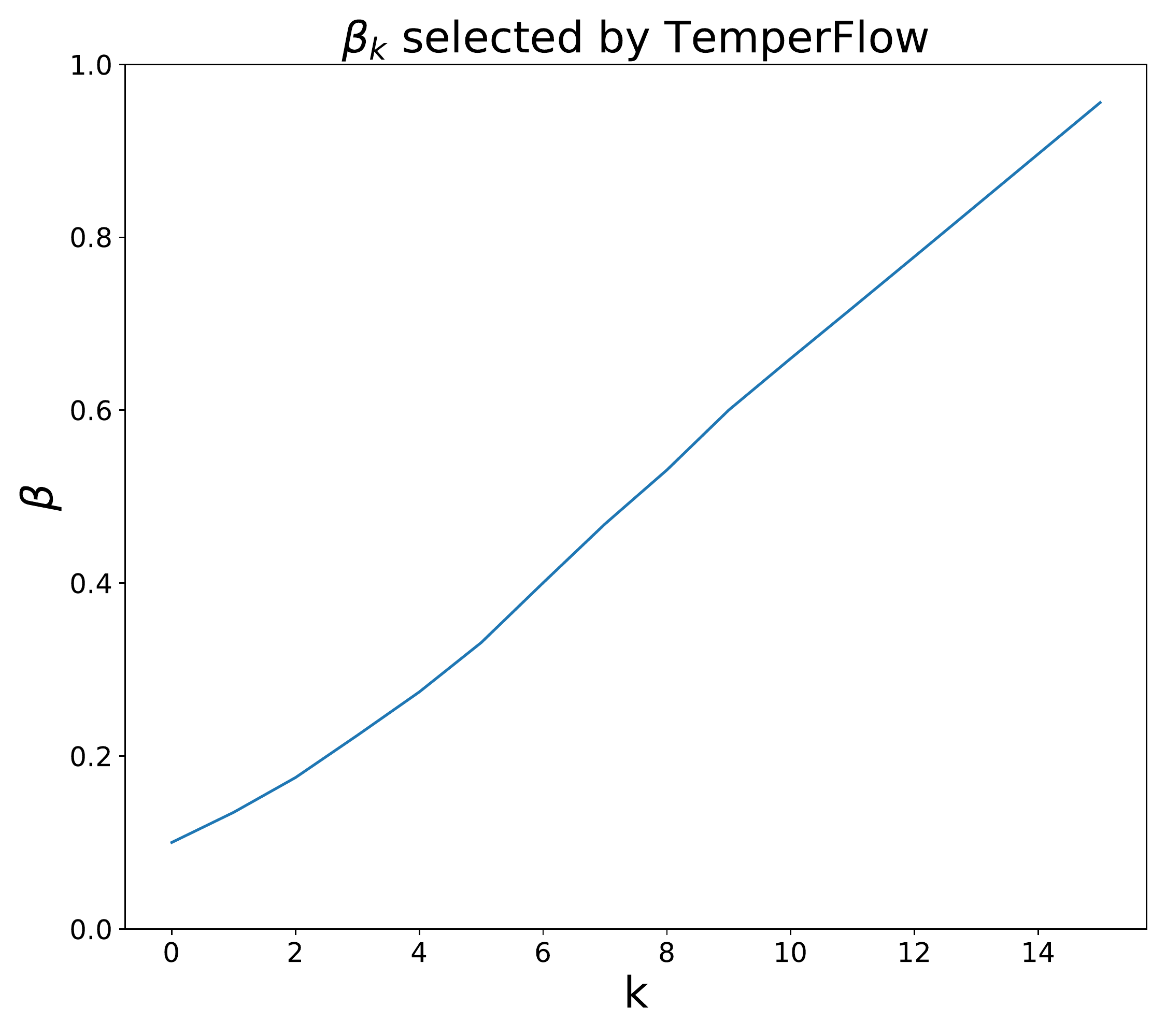}}
\subfloat[]{\includegraphics[width=0.32\textwidth]{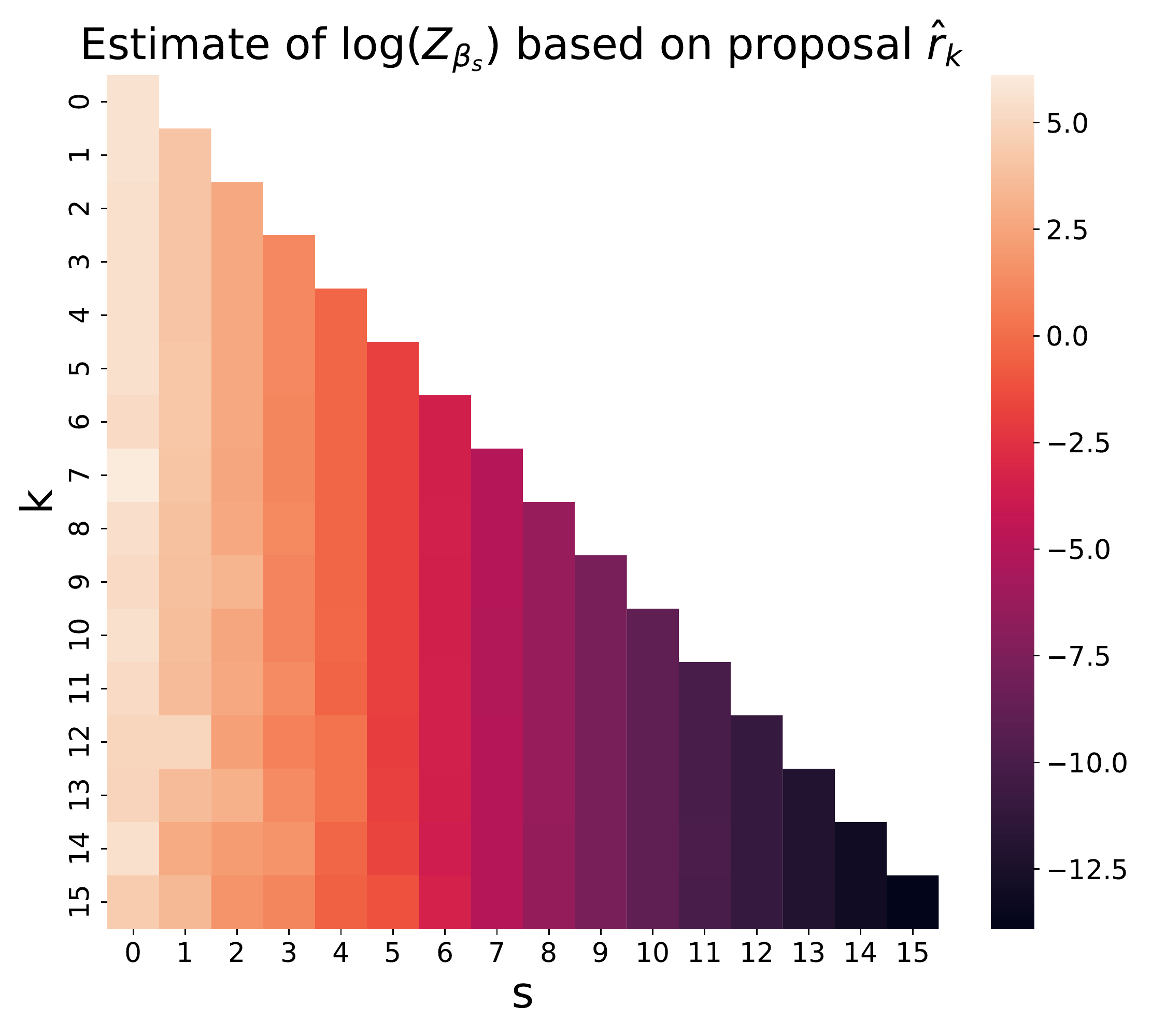}

} \subfloat[]{\includegraphics[width=0.32\textwidth]{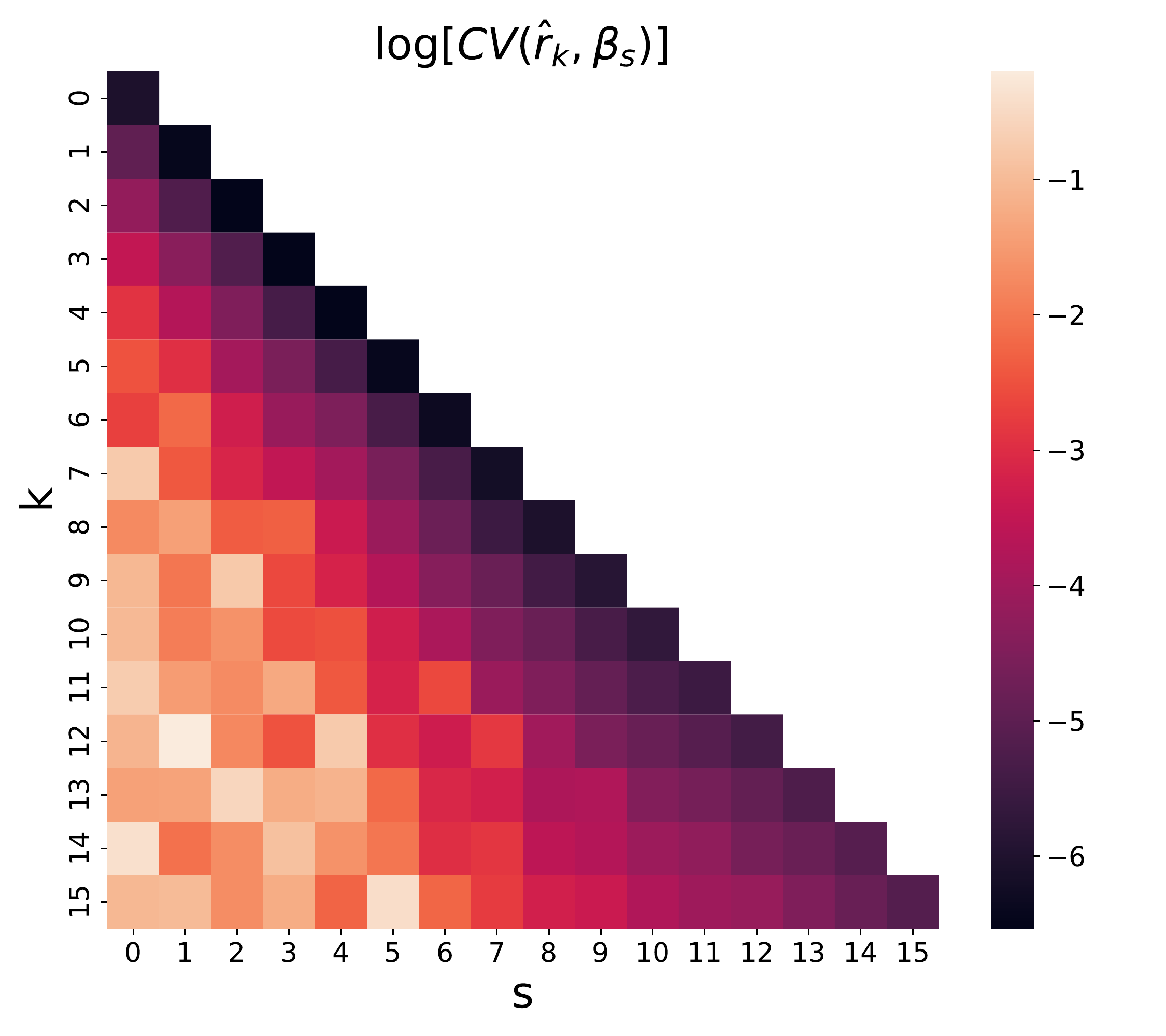}

}
\par\end{centering}
\caption{\label{fig:is_variance}(a) $\beta$ values selected by TemperFlow.
(b) Values of $\log(\hat{Z}_{\beta_{s}})$ with $\hat{r}_{k}$ used
as the proposal distribution. (c) Values of $\log[\mathrm{CV}(\hat{r}_{k},\beta_{s})]$
with different $(k,s)$.}
\end{figure}

\subsection{Proof of Proposition \ref{prop:vector_field}}

By definition, $\mathcal{G}(T)=\mathcal{G}_{1}(T)-\mathcal{G}_{2}(T)$,
where $\mathcal{G}_{1}(T)=\int q(x)\log q(x)\mathrm{d}x$ and $\mathcal{G}_{2}(T)=\int q(x)\log p(x)\mathrm{d}x$.
Let $\omega=T_{\sharp}\mu_{0}$ and $\nu=A_{\sharp}\omega=(A\circ T)_{\sharp}\mu_{0}$
for some mapping $A:\mathbb{R}^{d}\rightarrow\mathbb{R}^{d}$, and
then the density function of $\nu$, denoted by $r$, is given by
$r=(q/\det(\nabla A))\circ A^{-1}$. For any function $f:\mathbb{R}\rightarrow\mathbb{R}$,
we have
\[
\int f(r(x))\mathrm{d}x=\int f\left(\frac{q(A^{-1}(y))}{\det((\nabla A\circ A^{-1})(y))}\right)\mathrm{d}y=\int f\left(\frac{q(x)}{\det(\nabla A(x))}\right)\det(\nabla A(x))\mathrm{d}x.
\]
Take $f(x)=x\log(x)$, and then we obtain
\begin{align*}
\mathcal{G}_{1}(A\circ T) & =\int r(x)\log r(x)\mathrm{d}x=\int q(x)\left[\log q(x)-\log\det(\nabla A(x))\right]\mathrm{d}x\\
 & =\mathcal{G}_{1}(T)-\int\log\det(\nabla A(x))\mathrm{d}\omega(x).
\end{align*}
Let $I:\mathbb{R}^{d}\rightarrow\mathbb{R}^{d}$ denote the identity
map and $I_{d}$ the $d\times d$ identity matrix. Consider $A=I+\varepsilon\Psi\circ T^{-1}$
for an arbitrary $\Psi\in\mathbb{H}$, and then $\nabla A(x)=I_{d}+\varepsilon G(x)$
with $G=\nabla(\Psi\circ T^{-1})$.

Let $l(\varepsilon;x)=\log\det(I_{p}+\varepsilon G(x))$, and then
by Taylor's theorem,
\[
l(\varepsilon;x)=l(0;x)+l'(0;x)\varepsilon+\frac{1}{2}\int_{0}^{\varepsilon}l''(t;x)(\varepsilon-t)\mathrm{d}t.
\]
It can be verified that $l'(\varepsilon;x)=\mathrm{tr}\{U(\varepsilon;x)\}$
and $l''(\varepsilon;x)=-\mathrm{tr}\{[U(\varepsilon;x)]^{2}\}$,
where $U(\varepsilon;x)=[I_{d}+\varepsilon G(x)]^{-1}G(x)$, so
\[
\int\log\det(\nabla A(x))\mathrm{d}\omega(x)=\varepsilon\int\mathrm{tr}\{G(x)\}\mathrm{d}\omega(x)-\frac{1}{2}\int\left[\int_{0}^{\varepsilon}\mathrm{tr}\{[U(t;x)]^{2}\}(\varepsilon-t)\mathrm{d}t\right]\mathrm{d}\omega(x).
\]
Let $\sigma_{1}(M)\ge\cdots\ge\sigma_{d}(M)$ be the ordered singular
values of a matrix $M$, and then by Von Neumann's trace inequality,
we have
\[
\left|\mathrm{tr}\{[U(t;x)]^{2}\}\right|\le\sum_{i=1}^{d}[\sigma_{i}(U(t;x))]^{2}=\sum_{i=1}^{d}[\sigma_{i}(M_{1}^{-1}M_{2})]^{2},
\]
where $M_{1}=I_{d}+tG(x)$ and $M_{2}=G(x)$. Moreover, Theorem 3.3.14(b)
of \citet{horn1991matrix} shows that
\[
\sum_{i=1}^{d}[\sigma_{i}(M_{1}^{-1}M_{2})]^{2}\le\sum_{i=1}^{d}[\sigma_{i}(M_{1}^{-1})\sigma_{i}(M_{2})]^{2}=\sum_{i=1}^{d}\left[\frac{\sigma_{i}(M_{2})}{\sigma_{d-i+1}(M_{1})}\right]^{2}.
\]
By assumption (a), $\sigma_{i}(M_{2})\le\sigma_{1}(M_{2})=\Vert G(x)\Vert_{\mathrm{op}}\le C_{\Psi,T}$,
and Theorem 3.3.16(c) of \citet{horn1991matrix} indicates that $\sigma_{i}(M_{1})\ge1-|t|\sigma_{1}(G(x))\ge1-|t|C_{\Psi,T}$.
Therefore, for all $|t|<1/(2C_{\Psi,T})$, we have $\sigma_{i}(M_{1})\ge1/2$,
and hence
\[
\left|\mathrm{tr}\{[U(t;x)]^{2}\}\right|\le\sum_{i=1}^{d}\left[\frac{\sigma_{i}(M_{2})}{\sigma_{d-i+1}(M_{1})}\right]^{2}\le4dC_{\Psi,T}^{2}.
\]
As a result, for $0\le\varepsilon\le1/(2C_{\Psi,T})$,
\[
\left|\int\left[\int_{0}^{\varepsilon}\mathrm{tr}\{[U(t;x)]^{2}\}(\varepsilon-t)\mathrm{d}t\right]\mathrm{d}\omega(x)\right|\le4dC_{\Psi,T}^{2}\int_{0}^{\varepsilon}(\varepsilon-t)\mathrm{d}t=2dC_{\Psi,T}^{2}\varepsilon^{2},
\]
and the case of $\varepsilon<0$ can be proved similarly. This implies
that
\[
\left.\frac{\mathrm{d}}{\mathrm{d}\varepsilon}\int\log\det(\nabla A(x))\mathrm{d}\omega(x)\right|_{\varepsilon=0}=\int\mathrm{tr}\{G(x)\}\mathrm{d}\omega(x),
\]
where $\int\mathrm{tr}\{G(x)\}\mathrm{d}\omega(x)$ exists since $|\mathrm{tr}(G(x))|\le\sum_{i=1}^{d}\sigma_{i}(G(x))\le dC_{\Psi,T}$.
Consequently,
\[
\left.\frac{\mathrm{d}}{\mathrm{d}\varepsilon}\mathcal{G}_{1}(T+\varepsilon\Psi)\right|_{\varepsilon=0}=\left.\frac{\mathrm{d}}{\mathrm{d}\varepsilon}\mathcal{G}_{1}(A\circ T)\right|_{\varepsilon=0}=-\int\mathrm{tr}\left\{ G(x)\right\} \mathrm{d}\omega(x)=-\int\nabla\cdot(\Psi\circ T^{-1})\mathrm{d}\omega,
\]
where $\nabla\cdot F$ is the divergence of a vector field $F$. For
a scalar-valued function $\varphi$, the divergence operator satisfies
the product rule $\nabla\cdot(\varphi F)=(\nabla\varphi)\cdot F+\varphi(\nabla\cdot F)$,
so
\begin{align*}
\int\nabla\cdot(\Psi\circ T^{-1})\mathrm{d}\omega & =\int q(x)\left(\nabla\cdot(\Psi\circ T^{-1})\right)(x)\mathrm{d}x\\
 & =\int\nabla\cdot(q(\Psi\circ T^{-1}))\mathrm{d}\lambda-\int\langle\nabla q,\Psi\circ T^{-1}\rangle\mathrm{d}\lambda\\
 & =\int\nabla\cdot(q(\Psi\circ T^{-1}))\mathrm{d}\lambda-\int\langle\nabla\log q,\Psi\circ T^{-1}\rangle\mathrm{d}\omega\\
 & =\int\nabla\cdot(q(\Psi\circ T^{-1}))\mathrm{d}\lambda-\int\langle\nabla\log q\circ T,\Psi\rangle\mathrm{d}\mu_{0}.
\end{align*}
Under the assumption that $\lim_{\Vert x\Vert\rightarrow\infty}\Vert q(x)\Psi(T^{-1}(x))\Vert=0$,
we have $\int\nabla\cdot(q(\Psi\circ T^{-1}))\mathrm{d}\lambda=0$,
so $\delta\mathcal{G}_{1}/\delta T=\nabla\log q\circ T$.

For the second part, we have $\mathcal{G}_{2}(T)=\int\log p(T(x))\mathrm{d}\mu_{0}(x)$
according to the definition of the pushforward operator. By Taylor's
theorem,
\[
\log p(T+\varepsilon\Psi)-\log p(T)=\varepsilon\langle\nabla\log p\circ T,\Psi\rangle+\frac{1}{2}\varepsilon^{2}\Psi'\left[\nabla^{2}(\log p)\right](T+c\varepsilon\Psi)\Psi
\]
for some $c\in(0,1)$. Since we have assumed that $\Vert\nabla^{2}(\log p)\Vert_{\mathrm{op}}\le c$,
we get
\[
\int\left|\Psi'\left[\nabla^{2}(\log p)\right](T+c\varepsilon\Psi)\Psi\right|\mathrm{d}\mu_{0}\le c\int\Vert\Psi(x)\Vert^{2}\mathrm{d}\mu_{0}(x)<\infty,
\]
and then
\begin{align*}
\left.\frac{\mathrm{d}}{\mathrm{d}\varepsilon}\mathcal{G}_{2}(T+\varepsilon\Psi)\right|_{\varepsilon=0} & =\lim_{\varepsilon\rightarrow0}\frac{\mathcal{G}_{2}(T+\varepsilon\Psi)-\mathcal{G}_{2}(T)}{\varepsilon}\\
 & =\lim_{\varepsilon\rightarrow0}\varepsilon^{-1}\int\left[\varepsilon\langle\nabla\log p\circ T,\Psi\rangle+\frac{1}{2}\varepsilon^{2}\Psi'\left[\nabla^{2}(\log p)\right](T+c\varepsilon\Psi)\Psi\right]\mathrm{d}\mu_{0}\\
 & =\int\langle\nabla\log p\circ T,\Psi\rangle\mathrm{d}\mu_{0}=\langle\nabla\log p\circ T,\Psi\rangle_{\mathbb{H}}.
\end{align*}
As a result, $\delta\mathcal{G}_{2}/\delta T=\nabla\log p\circ T$.

Combining the results above, it is easy to find that $\delta\mathcal{G}/\delta T=\delta\mathcal{G}_{1}/\delta T-\delta\mathcal{G}_{2}/\delta T=\nabla(\log q-\log p)\circ T$.

\subsection{Proof of Theorem \ref{thm:kl_decay}}

The first part of the theorem is the consequence of several classical
and recent results. We first show that if the probability measure
$\mu$ has a log-concave density $p(x)$, then it satisfies the Poincar\'e
inequality
\begin{equation}
\int f^{2}\mathrm{d}\mu-\left(\int f\mathrm{d}\mu\right)^{2}\le C_{P}\int\Vert\nabla f\Vert^{2}\mathrm{d}\mu\label{eq:poincare_ineq}
\end{equation}
for all locally Lipschitz $f\in L^{2}(\mu)$, where $C_{P}>0$ is
independent of $f$. Indeed, \citet{cheeger1970lower} shows that
(\ref{eq:poincare_ineq}) holds with $C_{P}=4/\psi_{\mu}^{2}$, where
\[
\psi_{\mu}=\inf_{S\subset\mathbb{R}^{d}}\frac{\int_{\partial S}p(x)\mathrm{d}x}{\min\left\{ \mu(S),\mu(S^{c})\right\} }
\]
is called the isoperimetric coefficient of $\mu$. Moreover, Theorem
1 of \citet{chen2021almost} proves that for any integers $d,l\ge1$,
\[
\psi_{\mu}\ge\frac{1}{[c\cdot l(\log(d)+1)]^{l/2}\cdot d^{16/l}\cdot\sqrt{\sigma^{2}}}
\]
for some universal constant $c>0$, where $\sigma^{2}$ is the spectral
norm of the covariance matrix of $\mu$. For convenience, let $l_{1}\equiv l_{1}(d)=\log(d)+1$,
$l_{2}=\log(l_{1})+1$, $L_{1}=\sqrt{l_{1}/l_{2}}$, and $L_{2}=\sqrt{l_{1}l_{2}}$.
Take $l=\lceil L_{1}\rceil$, where $\lceil x\rceil$ is the smallest
integer greater than or equal to $x$. It is easy to verify that $L_{1}$
is increasing in $d$ and $L_{1}\ge1$, so we must have $L_{1}\le l=\lceil L_{1}\rceil\le2L_{1}$.
Also note that $L_{1}\le\sqrt{l_{1}}$ and
\[
\log(L_{1})\le\frac{1}{2}\log(l_{1})<\frac{1}{2}l_{2},\qquad L_{1}\log(l_{1})<L_{1}l_{2}=L_{2},\qquad L_{1}\log(L_{1})<L_{1}\cdot\frac{1}{2}l_{2}=\frac{1}{2}L_{2},
\]
so
\begin{align*}
\log\left\{ [c\cdot l(\log(d)+1)]^{l/2}\cdot d^{16/l}\right\}  & =\frac{l}{2}\cdot\log(c\cdot l\cdot l_{1})+\frac{16}{l}\log(d)\\
 & \le L_{1}\left[\log(2c)+\log(L_{1})+\log(l_{1})\right]+\frac{16}{L_{1}}\cdot l_{1}\\
 & \le\log(2c)L_{1}+\frac{3}{2}L_{2}+16L_{2}.
\end{align*}
Obviously $L_{1}\le L_{2}$, so there exists a constant $c'>0$ such
that $\log\left\{ [c\cdot l(\log(d)+1)]^{l/2}\cdot d^{16/l}\right\} \le c'L_{2}$
and
\[
\psi_{\mu}\ge\frac{1}{e^{c'L_{2}}\cdot\sqrt{\sigma^{2}}}.
\]
Combining the results above, we can take $C_{P}=4\sigma^{2}e^{2c'L_{2}}$.

Next, Corollary 4 of \citet{chewi2020exponential} states that if
$p(x)$ satisfies the Poincar\'e inequality (\ref{eq:poincare_ineq}),
then the law $\{\nu_{t}\}$ of the Langevin diffusion
\[
\mathrm{d}X_{t}=\nabla_{x}\log p(X_{t})\mathrm{d}t+\sqrt{2}\mathrm{d}B_{t},
\]
where $B_{t}$ is the Brownian motion on $\mathbb{R}^{d}$, satisfies
$\mathcal{F}(\nu_{t})\le e^{-2t/C_{P}}\chi^{2}(\nu_{0}\Vert\mu)$.
It is known that the law of the Langevin diffusion is the gradient
flow of the KL divergence \citep{jordan1998variational}, so $\mu_{t}$
and $\nu_{t}$ have the same marginal distributions. Therefore, we
also have $\mathcal{F}(\mu_{t})\le e^{-2t/C_{P}}\chi^{2}(\mu_{0}\Vert\mu)$.
Plugging in the expression of $C_{P}$ and absorbing constants into
$c'$, we get the stated result.

For the second part, (\ref{eq:kl_derivative}) and the equation $\mathbf{v}_{t}(x)=-\nabla\left[\log p_{t}(x)-\log p(x)\right]$
indicate that
\[
\frac{\mathrm{d}}{\mathrm{d}t}\mathcal{F}(\mu_{t})=-\int\left\Vert \nabla\log\frac{\mathrm{d\mu_{t}}}{\mathrm{d}\mu}\right\Vert ^{2}\mathrm{d}\mu_{t}.
\]
Let $f=\sqrt{\mathrm{d}\mu_{t}/\mathrm{d}\mu}$, and then
\[
\left\Vert \nabla\log\frac{\mathrm{d\mu_{t}}}{\mathrm{d}\mu}\right\Vert ^{2}=\left\Vert \nabla\log(f^{2})\right\Vert ^{2}=4\left\Vert \nabla\log f\right\Vert ^{2}=\frac{4}{f^{2}}\left\Vert \nabla f\right\Vert ^{2}.
\]
Therefore,
\[
\frac{\mathrm{d}}{\mathrm{d}t}\mathcal{F}(\mu_{t})=-\frac{4\mathrm{d}\mu}{\mathrm{d}\mu_{t}}\int\Vert\nabla f\Vert^{2}\mathrm{d}\mu_{t}=-4\int\Vert\nabla f\Vert^{2}\mathrm{d}\mu.
\]
By the Poincar\'e inequality (\ref{eq:poincare_ineq}), we have
\begin{align*}
\frac{\mathrm{d}}{\mathrm{d}t}\mathcal{F}(\mu_{t}) & \le-\frac{4}{C_{P}}\left[\int f^{2}\mathrm{d}\mu-\left(\int f\mathrm{d}\mu\right)^{2}\right]=-\frac{4}{C_{P}}\left[\int\frac{\mathrm{d}\mu_{t}}{\mathrm{d}\mu}\mathrm{d}\mu-\left(\int\sqrt{\frac{\mathrm{d}\mu_{t}}{\mathrm{d}\mu}}\mathrm{d}\mu\right)^{2}\right].\\
 & =-\frac{4}{C_{P}}\left[1-\left(\int\sqrt{\frac{\mathrm{d}\mu_{t}}{\mathrm{d}\lambda}\cdot\frac{\mathrm{d}\mu}{\mathrm{d}\lambda}}\mathrm{d}\lambda\right)^{2}\right].
\end{align*}
Note that
\begin{align*}
2H^{2}(\mu,\nu) & =\int(\sqrt{\mathrm{d}\mu/\mathrm{d}\lambda}-\sqrt{\mathrm{d}\nu/\mathrm{d}\lambda})^{2}\mathrm{d}\lambda=\int\mathrm{d}\mu+\int\mathrm{d}\nu-2\int\sqrt{\mathrm{d}\mu/\mathrm{d}\lambda}\cdot\sqrt{\mathrm{d}\nu/\mathrm{d}\lambda}\mathrm{d}\lambda\\
 & =2\left[1-\int\sqrt{\mathrm{d}\mu/\mathrm{d}\lambda}\cdot\sqrt{\mathrm{d}\nu/\mathrm{d}\lambda}\mathrm{d}\lambda\right],
\end{align*}
so
\[
\frac{\mathrm{d}}{\mathrm{d}t}\mathcal{F}(\mu_{t})\le-\frac{4}{C_{P}}\left[1-\left(1-H^{2}(\mu_{t},\mu)\right)^{2}\right]=-\frac{4}{C_{P}}\left[2-H^{2}(\mu_{t},\mu)\right]\cdot H^{2}(\mu_{t},\mu)\le-\frac{4}{C_{P}}H^{2}(\mu_{t},\mu).
\]
The last inequality holds since $0\le H(\mu,\nu)\le1$. Finally, plugging
in the expression of $C_{P}$ yields the result.

\subsection{Proof of Theorem \ref{thm:f_divergence_grad}}

For the brevity of notation we let $q(x)\equiv p_{t^{*}}(x)$, and
$\mathcal{F}$ stands for the $\phi$-divergence-based functional
$\mathcal{F}_{\phi}$. By definition,
\begin{align*}
\mathcal{F}(\mu_{t}) & =\mathcal{D}_{\phi}(\mu_{t}\Vert\mu)=\int p(x)\phi\left(\frac{q(x)}{p(x)}\right)\mathrm{d}x\\
 & =\alpha\int h(x)\phi\left(\frac{q(x)}{p(x)}\right)\mathrm{d}x+(1-\alpha)\int h(x-\mu)\phi\left(\frac{q(x)}{p(x)}\right)\mathrm{d}x\\
 & =\alpha\int h(x)\phi\left(\frac{q(x)}{p(x)}\right)\mathrm{d}x+(1-\alpha)\int h(x)\phi\left(\frac{q(x+\mu)}{p(x+\mu)}\right)\mathrm{d}x,\\
\frac{q(x)}{p(x)} & =\frac{\gamma h(x)+(1-\gamma)h(x-\mu)}{\alpha h(x)+(1-\alpha)h(x-\mu)}=1+\frac{(\gamma-\alpha)(h(x)-h(x-\mu))}{\alpha h(x)+(1-\alpha)h(x-\mu)}.
\end{align*}
Using the elementary inequality $(a+b)/(c+d)\le a/c+b/d$ for any
$a,b,c,d\ge0$, we have
\[
\frac{p(x)}{q(x)}\le\frac{\alpha}{\gamma}+\frac{1-\alpha}{1-\gamma}\coloneqq c_{1}^{-1}(\alpha,\gamma),\quad\frac{q(x)}{p(x)}\le\frac{\gamma}{\alpha}+\frac{1-\gamma}{1-\alpha}\coloneqq c_{2}(\alpha,\gamma).
\]
Therefore, $0<c_{1}\le q(x)/p(x)\le c_{2}<\infty$ for all $x\in\mathbb{R}^{d}$,
and we denote $D_{1}=\max_{x\in[c_{1},c_{2}]}|\phi(x)|$ and $D_{2}=\max_{x\in[c_{1},c_{2}]}|\phi''(x)|$.
Therefore,
\[
\left|\phi\left(\frac{q(x)}{p(x)}\right)\right|\le D_{1},\quad\left|\phi\left(\frac{q(x+\mu)}{p(x+\mu)}\right)\right|\le D_{1},\quad\forall x,\mu\in\mathbb{R}^{d}.
\]
Moreover, by assumption (a), we have for any fixed $x$, $\lim_{\Vert\mu\Vert\rightarrow\infty}\phi(q(x)/p(x))=\phi(\gamma/\alpha)$
and $\lim_{\Vert\mu\Vert\rightarrow\infty}\phi(q(x+\mu)/p(x+\mu))=\phi((1-\gamma)/(1-\alpha))$.
Since $\int Dh(x)\mathrm{d}x=D<\infty$, by the dominated convergence
theorem it holds that
\begin{align*}
\lim_{\Vert\mu\Vert\rightarrow\infty}\int h(x)\phi\left(\frac{q(x)}{p(x)}\right)\mathrm{d}x & =\int h(x)\phi\left(\frac{\gamma}{\alpha}\right)\mathrm{d}x=\phi\left(\frac{\gamma}{\alpha}\right),\\
\lim_{\Vert\mu\Vert\rightarrow\infty}\int h(x)\phi\left(\frac{q(x+\mu)}{p(x+\mu)}\right)\mathrm{d}x & =\int h(x)\phi\left(\frac{1-\gamma}{1-\alpha}\right)\mathrm{d}x=\phi\left(\frac{1-\gamma}{1-\alpha}\right),
\end{align*}
which gives the first result.

For the second part, by definition,
\begin{align*}
\nabla\left(\frac{q(x)}{p(x)}\right) & =(\alpha-\gamma)\cdot\nabla\left(\frac{h(x-\mu)-h(x)}{p(x)}\right)\\
 & =\frac{\alpha-\gamma}{[p(x)]^{2}}\cdot\left\{ p(x)[\nabla h(x-\mu)-\nabla h(x)]-[h(x-\mu)-h(x)]\nabla p(x)\right\} \\
 & =\frac{\alpha-\gamma}{[p(x)]^{2}}\cdot\left[h(x)\nabla h(x-\mu)-h(x-\mu)\nabla h(x)\right]\\
 & =(\alpha-\gamma)\left[\frac{h(x)}{p(x)}\right]^{2}\nabla\left(\frac{h(x-\mu)}{h(x)}\right).
\end{align*}
Define $\delta(x,\mu)=V(x-\mu)-V(x)$, and then
\[
\frac{h(x)}{p(x)}=\frac{h(x)}{\alpha h(x)+(1-\alpha)h(x-\mu)}=\frac{1}{\alpha+(1-\alpha)\exp\{-\delta(x,\mu)\}}.
\]
Take $c=\min\{\alpha,1-\alpha\}$ and define $\sigma(x)=1/(1+e^{-x})$,
and then $h(x)/p(x)\le\sigma(\delta(x,\mu))/c$. Moreover, $\nabla(h(x-\mu)/h(x))=\nabla(e^{-\delta(x,\mu)})=-e^{-\delta(x,\mu)}\nabla_{x}\delta(x,\mu)$.
By assumption (b), $\Vert\nabla_{x}\delta(x,\mu)\Vert=\Vert\nabla V(x-\mu)-\nabla V(x)\Vert\le C_{1}\Vert\mu\Vert^{k}$
for sufficiently large $\Vert\mu\Vert$, so
\begin{equation}
\left\Vert \nabla\left(\frac{q(x)}{p(x)}\right)\right\Vert \le\frac{|\alpha-\gamma|}{c^{2}}\cdot[\sigma(\delta(x,\mu))]^{2}e^{-\delta(x,\mu)}\cdot C_{1}\Vert\mu\Vert^{k}\le\frac{C_{1}|\alpha-\gamma|}{c^{2}}\cdot\Vert\mu\Vert^{k}e^{-|\delta(x,\mu)|},\label{eq:dqp_bound}
\end{equation}
where the second inequality is due to the fact that $[\sigma(x)]^{2}e^{-x}\le e^{-|x|}$.

Next, Lemma 2.1 of \citet{gao2019deep} shows that $(\delta\mathcal{F}/\delta\nu)(\rho)=\phi'(\mathrm{d}\rho/\mathrm{d}\mu)$,
and hence $\mathbf{v}_{t}(x)=-\phi''(p_{t}(x)/p(x))\nabla[p_{t}(x)/p(x)]$.
Plugging in (\ref{eq:dqp_bound}), we have
\[
\Vert\mathbf{v}_{t^{*}}(x)\Vert\le\frac{C_{1}D_{2}|\alpha-\gamma|}{c^{2}}\cdot\Vert\mu\Vert^{k}e^{-|\delta(x,\mu)|}.
\]
Therefore,
\[
0\le-\left.\frac{\mathrm{d}\mathcal{F}(\mu_{t})}{\mathrm{d}t}\right|_{t=t^{*}}=\int q(x)\Vert\mathbf{v}_{t^{*}}(x)\Vert^{2}\mathrm{d}x\le\frac{C_{1}D_{2}|\alpha-\gamma|}{c^{2}}\cdot\Vert\mu\Vert^{k}\int q(x)e^{-|\delta(x,\mu)|}\mathrm{d}x.
\]
By definition,
\begin{align*}
\int q(x)e^{-|\delta(x,\mu)|}\mathrm{d}x & =\gamma\int h(x)e^{-|\delta(x,\mu)|}\mathrm{d}x+(1-\gamma)\int h(x-\mu)e^{-|\delta(x,\mu)|}\mathrm{d}x\\
 & =\gamma\int h(x)e^{-|\delta(x,\mu)|}\mathrm{d}x+(1-\gamma)\int h(x)e^{-|\delta(x+\mu,\mu)|}\mathrm{d}x\\
 & =\gamma\mathbb{E}e^{-|Y_{\mu,1}|}+(1-\gamma)\mathbb{E}e^{-|Y_{\mu,2}|}.
\end{align*}
Let $g_{1}(y)$ be the density function of $|Y_{\mu,1}|$, and then
\begin{align*}
\mathbb{E}e^{-|Y_{\mu,1}|} & =\int_{0}^{C_{2}\Vert\mu\Vert}g_{1}(y)e^{-y}\mathrm{d}x+\int_{C_{2}\Vert\mu\Vert}^{+\infty}g_{1}(y)e^{-y}\mathrm{d}x\\
 & \le\int_{0}^{C_{2}\Vert\mu\Vert}g_{1}(y)\mathrm{d}x+\int_{C_{2}\Vert\mu\Vert}^{+\infty}g_{1}(y)e^{-C_{2}\Vert\mu\Vert}\mathrm{d}x\\
 & =P(|Y_{\mu,1}|\le C_{2}\Vert\mu\Vert)+e^{-C_{2}\Vert\mu\Vert}P(|Y_{\mu,1}|>C_{2}\Vert\mu\Vert)\\
 & \le P(|Y_{\mu,1}|\le C_{2}\Vert\mu\Vert)+e^{-C_{2}\Vert\mu\Vert}.
\end{align*}
Similarly, we can verify that $\mathbb{E}e^{-|Y_{\mu,2}|}\le P(|Y_{\mu,2}|\le C_{2}\Vert\mu\Vert)+e^{-C_{2}\Vert\mu\Vert}$.
Finally, by assumption (c), we obtain
\begin{align*}
C_{1}\Vert\mu\Vert^{k}\int q(x)e^{-|\delta(x,\mu)|}\mathrm{d}x & \le C_{1}\Vert\mu\Vert^{k}e^{-C_{2}\Vert\mu\Vert}+\gamma C_{1}\Vert\mu\Vert^{k}P(|Y_{\mu,1}|\le C_{2}\Vert\mu\Vert)+\\
 & \qquad(1-\gamma)C_{1}\Vert\mu\Vert^{k}P(|Y_{\mu,2}|\le C_{2}\Vert\mu\Vert)\rightarrow0
\end{align*}
as $\Vert\mu\Vert\rightarrow\infty$, which gives the requested conclusion.

\subsection{Proof of Theorem \ref{thm:monotonicity_kl} and (\ref{eq:delta_gamma})}

Let $q_{\beta}(x)=e^{-\beta E(x)}/Z(\beta)$, where $Z(\beta)=\int e^{-\beta E(x)}\mathrm{d}x$.
Define the function $\ell(\beta)=\mathrm{KL}(q_{\beta}\Vert p)$,
and we are going to show that $\ell(\beta)$ is decreasing in $\beta$.
By definition,
\[
\ell(\beta)=\int q_{\beta}(x)\log\frac{q_{\beta}(x)}{p(x)}\mathrm{d}x=\int q_{\beta}(x)[\log q_{\beta}(x)-\log p(x)]\mathrm{d}x.
\]
For simplicity let $\zeta(\beta)=\log Z(\beta)$, and then $\log q_{\beta}(x)=-\beta E(x)-\zeta(\beta)$
and $\log p(x)=-E(x)-\zeta(1)$. Therefore,
\begin{align*}
\ell(\beta) & =\int q_{\beta}(x)[E(x)+\zeta(1)-\beta E(x)-\zeta(\beta)]\mathrm{d}x\\
 & =(1-\beta)\int q_{\beta}(x)E(x)\mathrm{d}x+\zeta(1)-\zeta(\beta),
\end{align*}
and
\[
\ell'(\beta)=-\int q_{\beta}(x)E(x)\mathrm{d}x+(1-\beta)\int\frac{\partial q_{\beta}(x)}{\partial\beta}E(x)\mathrm{d}x-\zeta'(\beta).
\]
Note that
\begin{align*}
\zeta'(\beta) & =\frac{Z'(\beta)}{Z(\beta)}=\frac{1}{Z(\beta)}\int e^{-\beta E(x)}[-E(x)]\mathrm{d}x=-\int q_{\beta}(x)E(x)\mathrm{d}x,
\end{align*}
so $\ell'(\beta)=(1-\beta)\int[\partial q_{\beta}(x)/\partial\beta]E(x)\mathrm{d}x$.
Moreover, since $\partial\log q_{\beta}(x)/\partial\beta=-E(x)-\zeta'(\beta)$,
we have
\[
\frac{\partial q_{\beta}(x)}{\partial\beta}=q_{\beta}(x)\cdot\frac{\partial\log q_{\beta}(x)}{\partial\beta}=-q_{\beta}(x)[E(x)+\zeta'(\beta)].
\]
Consequently,
\begin{align}
\ell'(\beta) & =-(1-\beta)\int q_{\beta}(x)[E(x)+\zeta'(\beta)]E(x)\mathrm{d}x\nonumber \\
 & =-(1-\beta)\left\{ \int q_{\beta}(x)E^{2}(x)\mathrm{d}x-\left[\int q_{\beta}(x)E(x)\mathrm{d}x\right]^{2}\right\} .\label{eq:dl_beta}
\end{align}
By Jensen's inequality, $\int q_{\beta}(x)E^{2}(x)\mathrm{d}x\ge\left[\int q_{\beta}(x)E(x)\mathrm{d}x\right]^{2}$,
and the equal sign holds if and only if $E(x)$ is a constant. Since
the energy function cannot be a constant over $\mathbb{R}^{d}$, the
inequality is strict. As a result, $\ell'(\beta)<0$ when $0<\beta<1$,
and $\ell'(\beta)=0$ only when $\beta=1$. Since $\beta_{t}$ is
increasing in $t$, it follows that $\mathcal{F}(\rho_{t})=\ell(\beta_{t})$
is decreasing in $t$.

Finally, (\ref{eq:delta_gamma}) can be proved by noting that $\mathrm{d}[\log\ell(e^{\gamma})]/\mathrm{d}\gamma=\beta\cdot\ell'(\beta)/\ell(\beta)$,
where $\ell'(\beta)$ is given in (\ref{eq:dl_beta}), and we show
the expression of $\ell(\beta)$ below. By definition,
\[
\ell(\beta)=\int q_{\beta}(x)[\log q_{\beta}(x)-\log p(x)]\mathrm{d}x=\int q_{\beta}(x)[\log q_{\beta}(x)+E(x)+\zeta(1)]\mathrm{d}x,
\]
where
\[
\zeta(1)=\log Z(1)=\log\int\exp\{-E(x)\}\mathrm{d}x=\log\int q_{\beta}(x)\exp\{-E(x)-\log q_{\beta}(x)\}\mathrm{d}x,
\]
and then the equation holds.

\subsection{Proof of Proposition \ref{prop:assump1_normal}}

Since $f(x_{\varepsilon})=\varepsilon^{1/4}$ and $f(x)$ is decreasing
for $x\ge0$, we have $\inf_{Q^{\varepsilon}}f\le\varepsilon^{1/4}$.
Moreover, $\int_{Q^{\varepsilon}}f\mathrm{d}\lambda=2\Phi(x_{\varepsilon}+M\varepsilon^{1/12})-1$,
where $\Phi(\cdot)$ is the c.d.f. of the standard normal distribution.
For sufficiently small $\varepsilon$, $x_{\varepsilon}=\sqrt{-(\log\varepsilon)/2-\log(2\pi)}>\sqrt{\log(1/\varepsilon)/3}>2$.
Then by the well-known inequality $1-\Phi(x)<f(x)/x$ for $x>0$,
we have
\[
\int_{Q^{\varepsilon}}f\mathrm{d}\lambda>2\Phi(x_{\varepsilon})-1>1-2f(x_{\varepsilon})/x_{\varepsilon}=1-2\varepsilon^{1/4}/x_{\varepsilon}>1-\varepsilon^{1/4}
\]
if $\varepsilon$ is small enough. Hence Assumption \ref{assu:balls}(a)
holds.

For convenience let $\delta=M\varepsilon^{1/12}$. Since $f^{-2}(x)=2\pi e^{x^{2}}$
is symmetric about zero and is increasing for $x>0$, the supremum
in Assumption \ref{assu:balls}(b) is achieved when $Q=Q^{\varepsilon}$.
It is known that $h(x)=e^{-x^{2}}\int_{0}^{x}e^{y^{2}}\mathrm{d}y<0.6$,
so
\[
\int_{Q^{\varepsilon}}f^{-2}(x)\mathrm{d}x=4\pi\int_{0}^{x_{\varepsilon}+\delta}e^{x^{2}}\mathrm{d}x=4\pi e^{(x_{\varepsilon}+\delta)^{2}}h(x_{\varepsilon}+\delta)<2.4\pi e^{(x_{\varepsilon}+\delta)^{2}}.
\]
Let $l(x)=e^{x^{2}/2}$, and then $l'(x)=xl(x)$ and $l''(x)=l(x)(x^{2}+1)$,
so
\[
e^{(x_{\varepsilon}+\delta)^{2}/2}=e^{x_{\varepsilon}^{2}/2}\left[1+x_{\varepsilon}\delta+(1+x_{\varepsilon}^{2})\delta^{2}+O((x_{\varepsilon}\delta)^{3})\right]<Ce^{x_{\varepsilon}^{2}/2}=C\varepsilon^{-1/4}
\]
for some constant $C>0$. As a result,
\[
\sup_{\begin{subarray}{c}
Q\subset Q^{\varepsilon}\\
Q\in\mathcal{Q}
\end{subarray}}|Q|^{2/d}\left[\frac{1}{|Q|}\int_{Q}f^{-2}\mathrm{d}\lambda\right]^{1/2}=|Q^{\varepsilon}|^{3/2}\left[\int_{Q^{\varepsilon}}f^{-2}(x)\mathrm{d}x\right]^{1/2}<\sqrt{2.4\pi}C\varepsilon^{-1/4}|Q^{\varepsilon}|^{3/2}.
\]
$|Q^{\varepsilon}|$ is at the order of $\sqrt{\log(1/\varepsilon)}$,
so Assumption \ref{assu:balls}(b) holds with $c=3/4$.

Finally, note that
\[
\int_{Q^{\varepsilon}}\left[\varepsilon^{1/4}-f(x)\right]_{+}^{2}\mathrm{d}x=2\int_{x_{\varepsilon}}^{x_{\varepsilon}+\delta}\left[\varepsilon^{1/4}-f(x)\right]^{2}\mathrm{d}x.
\]
Let $g(x)=\left[\varepsilon^{1/4}-f(x)\right]^{2}$, and then we have
$g'(x)=-2\left[\varepsilon^{1/4}-f(x)\right]f'(x)$ and $g''(x)=2[f'(x)]^{2}-2\left[\varepsilon^{1/4}-f(x)\right]f''(x)$.
Since $g(x_{\varepsilon})=g'(x_{\varepsilon})=0$ and
\[
g''(x_{\varepsilon})=2[f'(x_{\varepsilon})]^{2}=2x_{\varepsilon}^{2}[f(x_{\varepsilon})]^{2}=2x_{\varepsilon}^{2}\varepsilon^{1/2},
\]
by the Taylor expansion we have
\begin{align*}
\frac{1}{2}\int_{Q^{\varepsilon}}\left[\varepsilon^{1/4}-f(x)\right]_{+}^{2}\mathrm{d}x & =g(x_{\varepsilon})\delta+\frac{1}{2}g'(x_{\varepsilon})\delta^{2}+\frac{1}{6}g''(x_{\varepsilon})\delta^{3}+O(\delta^{4})\\
 & =\frac{1}{3}x_{\varepsilon}^{2}\varepsilon^{1/2}\delta^{3}>\frac{M^{3}}{9}\varepsilon^{3/4}\log(1/\varepsilon)>\frac{M^{3}}{9}\varepsilon^{3/4}[\log(1/\varepsilon)]^{3/4}
\end{align*}
for sufficiently small $\varepsilon$, which verifies Assumption \ref{assu:balls}(c).

\subsection{Technical Lemmas}

Let $\lambda$ denote the Lebesgue measure on $\mathbb{R}^{d}$. A
function $w$ is called a weight if it is nonnegative and locally
integrable on $\mathbb{R}^{d}$. A weight $w$ induces a measure,
defined by $w(A)=\int_{A}w\mathrm{d}\lambda$. Denote by $\mathcal{Q}$
the set of balls in $\mathbb{R}^{d}$. A weight $w$ is doubling if
$w(2Q)\le Cw(Q)$ for every $Q\in\mathcal{Q}$, where $2Q$ denotes
the ball with the same center as $Q$ and twice its radius, and $C$
is called the doubling constant for $w$. For a measurable set $A\subset\mathbb{R}^{d}$,
let $|A|=\lambda(A)$, and we use $\Vert f\Vert_{L_{w}^{p}(A)}$ to
denote the $w$-weighted $L^{p}$ norm of a vector-valued function
$f:\mathbb{R}^{d}\rightarrow\mathbb{R}^{d'}$ on $A$, \emph{i.e.},
\[
\Vert f\Vert_{L_{w}^{p}(A)}=\left[\int_{A}\Vert f(x)\Vert^{p}w(x)\mathrm{d}x\right]^{1/p},
\]
where $\Vert\cdot\Vert$ is the Euclidean norm.
\begin{lem}[Theorem 2.14, \citealp{chua1993weighted}]
\label{lem:weighted_sobolev_ineq}Let $1<p\le q<\infty$ and let
$\sigma=v^{-1/(p-1)}$ where $v$ is a weight. Suppose $w$ is a doubling
weight. Then for all $Q_{0}\in\mathcal{Q}$ and Lipschitz continuous
function $f$ on $Q_{0}$,
\[
\Vert f-f_{Q_{0},w}\Vert_{L_{w}^{q}(Q_{0})}\le A(v,w,Q_{0})\Vert\nabla f\Vert_{L_{v}^{p}(Q_{0})}
\]
where
\[
A(v,w,Q_{0})=C(p,q,C_{0})\sup_{\substack{Q\subset Q_{0}\\
Q\in\mathcal{Q}
}
}|Q|^{1/d-1}w(Q)^{1/q}\sigma(Q)^{1/p'}
\]
when $p<q$, and
\[
A(v,w,Q_{0})=C(p,r,C_{0})\sup_{\begin{subarray}{c}
Q\subset Q_{0}\\
Q\in\mathcal{Q}
\end{subarray}}|Q|^{1/d}\left[\frac{1}{|Q|}\int_{Q}w^{r}\mathrm{d}\lambda\right]^{1/(pr)}\left[\frac{1}{|Q|}\int_{Q}\sigma^{r}\mathrm{d}\lambda\right]^{1/(p'r)}
\]
when $p=q$ for any $r>1$. $C_{0}$ is the doubling constant for
$w$, $p'=p/(p-1)$, and $f_{Q_{0},w}=w(Q_{0})^{-1}\int_{Q_{0}}fw\mathrm{d}\lambda$.
\end{lem}
\begin{lem}
\label{lem:weighted_sobolev_density}Let $g(x)$ be a density function
defined on $\mathbb{R}^{d}$. Then for any ball $Q_{0}\in\mathcal{Q}$
and any Lipschitz continuous function $h$ on $Q_{0}$, the following
inequality holds,
\[
\int_{Q_{0}}\left|h(x)-h_{Q_{0}}\right|^{2}\mathrm{d}x\le A(g,d,Q_{0})\int_{Q_{0}}\Vert\nabla h(x)\Vert^{2}g(x)\mathrm{d}x,
\]
where
\[
h_{Q_{0}}=\frac{1}{|Q_{0}|}\int_{Q_{0}}h\mathrm{d}\lambda,\quad A(g,d,Q_{0})=C(d)\sup_{\begin{subarray}{c}
Q\subset Q_{0}\\
Q\in\mathcal{Q}
\end{subarray}}|Q|^{2/d}\left[\frac{1}{|Q|}\int_{Q}g^{-2}\mathrm{d}\lambda\right]^{1/2},
\]
\textup{and $C(d)$ is a constant independent of $g$, $h$, and $Q_{0}$.}
\end{lem}
\begin{proof}
It is easy to verify that $w\equiv1$ is a doubling weight with doubling
constant $C_{0}=2^{d}$. Then take $v=g$, $w=1$, $p=q=r=2$ in Lemma
\ref{lem:weighted_sobolev_ineq}, and the conclusion holds.
\end{proof}
\begin{lem}
\label{lem:inf_bound}Let $f(x)$ be a continuous function defined
on a compact set $S\subset\mathbb{R}^{d}$. Assume $\inf_{S}f(x)<c$
for some constant $c>0$, and denote $M=\int_{S}\left[c-f(x)\right]_{+}^{2}\mathrm{d}x$,
where $g_{+}=\max(0,g)$ stands for the positive part of a function
$g$. Then for any continuous function $e(x)$ satisfying $\int_{S}[e(x)]^{2}\mathrm{d}x\le M$,
we have $\inf_{S}\ [f(x)+e(x)]\le c$.
\end{lem}
\begin{proof}
Since any $e(x)$ can be replaced by $|e(x)|$ to achieve a larger
$\inf_{S}\ [f(x)+e(x)]$ under the same condition $\int_{S}[e(x)]^{2}\mathrm{d}x\le M$,
we can assume that $e(x)\ge0$ without loss of generality. 

Let $h_{c}(x)=\max(c,f(x))$, $e^{*}(x)=h_{c}(x)-f(x)=[c-f(x)]_{+}$,
and $u^{*}(x)=[e^{*}(x)]^{2}$. Then it is easy to see that $h_{c}(x)=c$
on $T\coloneqq\{x:f(x)\le c\}$ and $e^{*}(x)=u^{*}(x)=0$ on $S\backslash T$,
implying $M=\int_{S}u^{*}(x)\mathrm{d}x=\int_{T}u^{*}(x)\mathrm{d}x$.
For any continuous $u(x)\ge0$ such that $\int_{S}u(x)\mathrm{d}x\le M$,
we claim that there must exist some point $x_{0}\in T$ such that
$u(x_{0})\le u^{*}(x_{0})$. If this is not true, then
\[
\int_{S}u(x)\mathrm{d}x>\int_{T}u^{*}(x)\mathrm{d}x+\int_{S\backslash T}u(x)\mathrm{d}x\ge\int_{T}u^{*}(x)\mathrm{d}x+\int_{S\backslash T}u^{*}(x)\mathrm{d}x=\int_{S}u^{*}(x)\mathrm{d}x=M,
\]
which causes a contradiction.

Let $e(x)=\sqrt{u(x)}$, and then it is also true that $e(x_{0})\le e^{*}(x_{0})$
for some $x_{0}\in T$, implying
\[
\inf_{S}\ [f(x)+e(x)]\le f(x_{0})+e(x_{0})\le f(x_{0})+e^{*}(x_{0})=h_{c}(x_{0})=c.
\]
\end{proof}

\subsection{Proof of Theorem \ref{thm:grad_l2}}

It is easy to show that the first variation of $\mathcal{L}$ at $g$
is $(\delta\mathcal{L}/\delta g)(g)=2(g-f)$, and then the vector
field $\mathbf{v}_{t}$ for the continuity equation $\partial g_{t}/\partial t+\nabla\cdot(\mathbf{v}_{t}g_{t})=0$
is given by $\mathbf{v}_{t}=2\nabla(f-g_{t})$. Let $h_{t}=g_{t}-f$,
and then
\[
\frac{\mathrm{d}\mathcal{L}(g_{t})}{\mathrm{d}t}=-\int\Vert2\nabla h_{t}\Vert^{2}g_{t}\mathrm{d}\lambda=-4\int\Vert\nabla h_{t}\Vert^{2}g_{t}\mathrm{d}\lambda.
\]
In what follows, we omit the subscript $t$ in $g_{t}$ and $h_{t}$
and the superscript $\varepsilon$ in $Q_{i}^{\varepsilon}$ whenever
no confusion is caused. Since $|\mathrm{d}\mathcal{F}(g_{t})/\mathrm{d}t|\le\varepsilon$,
we have $\int\Vert\nabla h\Vert^{2}g\mathrm{d}\lambda\le\varepsilon/4$,
and hence
\[
\max_{1\le i\le K}\int_{Q_{i}}\Vert\nabla h\Vert^{2}g\mathrm{d}\lambda\le\varepsilon/4.
\]
By Lemma \ref{lem:weighted_sobolev_density}, we have
\begin{equation}
\int_{Q_{i}}\left|h(x)-h_{Q_{i}}\right|^{2}\mathrm{d}x\le A(g,d,Q_{i})\varepsilon/4,\label{eq:h_bound}
\end{equation}
where
\[
A(g,d,Q_{i})=C(d)\sup_{\begin{subarray}{c}
Q\subset Q_{i}\\
Q\in\mathcal{Q}
\end{subarray}}|Q|^{2/d}\left[\frac{1}{|Q|}\int_{Q}g^{-2}\mathrm{d}\lambda\right]^{1/2}.
\]
By Assumptions \ref{assu:balls}(b) and \ref{assu:lower_bound}, we
have
\[
A(g,d,Q_{i})\le\alpha^{-1}C(d)\sup_{\begin{subarray}{c}
Q\subset Q_{i}\\
Q\in\mathcal{Q}
\end{subarray}}|Q|^{2/d}\left[\frac{1}{|Q|}\int_{Q}f^{-2}\mathrm{d}\lambda\right]^{1/2}\le\alpha^{-1}C(d)M_{1}\varepsilon^{-1/4}[\log(1/\varepsilon)]^{c}.
\]
Since $M_{2}>0$ in Assumptions \ref{assu:balls}(c) is a sufficiently
large constant, we can take $M_{2}=(4\alpha)^{-1}C(d)M_{1}$, and
then (\ref{eq:h_bound}) leads to
\begin{equation}
\int_{Q_{i}}\left|h(x)-h_{Q_{i}}\right|^{2}\mathrm{d}x\le M_{2}\varepsilon^{3/4}[\log(1/\varepsilon)]^{c}.\label{eq:h_norm_bound}
\end{equation}

Let $e_{i}(x)=h(x)-h_{Q_{i}}$, and then by definition, $\int_{Q_{i}}e_{i}\mathrm{d}\lambda=0$
and $\int_{Q_{i}}e_{i}^{2}\mathrm{d}\lambda\le M_{2}\varepsilon^{3/4}[\log(1/\varepsilon)]^{c}$.
In Lemma \ref{lem:inf_bound}, take $S=Q_{i}$, $e=e_{i}$, $M=M_{2}\varepsilon^{3/4}[\log(1/\varepsilon)]^{c}$,
and $c=\varepsilon^{1/4}$, and then by Assumption \ref{assu:balls}(a)
we obtain $\inf_{Q_{i}}\ (f+e_{i})\le\varepsilon^{1/4}$. Since $f+e_{i}$
is continuous and $Q_{i}$ is compact, there must exist a point $x^{*}\in Q_{i}$
such that
\[
f(x^{*})+e_{i}(x^{*})=\inf_{Q_{i}}\ (f+e_{i})\le\varepsilon^{1/4}.
\]
Then we obtain
\begin{equation}
h_{Q_{i}}=h(x^{*})-e_{i}(x^{*})=g(x^{*})-f(x^{*})-e_{i}(x^{*})\ge-[f(x^{*})+e_{i}(x^{*})]\ge-\varepsilon^{1/4}.\label{eq:h_lb}
\end{equation}

On the other hand, $\int_{Q}g\mathrm{d}\lambda\le1$ and $\int_{Q}f\mathrm{d}\lambda\ge1-\varepsilon^{1/4}$
by Assumption \ref{assu:balls}(a), so
\[
\varepsilon^{1/4}\ge\int_{Q}(g-f)\mathrm{d}\lambda=\sum_{i=1}^{K}\int_{Q_{i}}(g-f)\mathrm{d}\lambda=\sum_{i=1}^{K}\int_{Q_{i}}(e_{i}+h_{Q_{i}})\mathrm{d}\lambda.
\]
Note that $\int_{Q_{i}}e_{i}\mathrm{d}\lambda=0$ and $\int_{Q_{i}}h_{Q_{i}}\mathrm{d}\lambda=h_{Q_{i}}|Q_{i}|$,
so we have
\begin{equation}
h_{Q_{i}}|Q_{i}|\le\sum_{i=1}^{K}h_{Q_{i}}|Q_{i}|\le\varepsilon^{1/4}.\label{eq:h_ub}
\end{equation}
Combining (\ref{eq:h_lb}) and (\ref{eq:h_ub}), and then we get $|h_{Q_{i}}|\le\max(1,|Q_{i}|^{-1})\cdot\varepsilon^{1/4}.$

Finally, (\ref{eq:h_norm_bound}) indicates that $\Vert h-h_{Q_{i}}\Vert_{L^{2}(Q_{i})}\le M_{2}^{1/2}\varepsilon^{3/8}[\log(1/\varepsilon)]^{c/2}$,
so
\begin{align*}
\Vert h\Vert_{L^{2}(Q_{i})} & \le\Vert h-h_{Q_{i}}\Vert_{L^{2}(Q_{i})}+\Vert h_{Q_{i}}\Vert_{L^{2}(Q_{i})}\le M_{2}^{1/2}\varepsilon^{3/8}[\log(1/\varepsilon)]^{c/2}+h_{Q_{i}}|Q_{i}|^{1/2}\\
 & \le M_{2}^{1/2}\varepsilon^{3/8}[\log(1/\varepsilon)]^{c/2}+\max(|Q_{i}|^{1/2},|Q_{i}|^{-1/2})\cdot\varepsilon^{1/4}.
\end{align*}

\setstretch{1}

\bibliographystyle{apalike}
\bibliography{ref}

\end{document}